%% file: main.tex
\newcommand{\iffull}[2]{%
  \ifcsname full\endcsname%
    #1%
  \else
    #2%
  \fi%
}
\documentclass[twoside,leqno,twocolumn]{article}
\usepackage{ltexpprt}

\usepackage[utf8]{inputenc}
\usepackage{etex}
\usepackage[english]{babel}
\usepackage{amsmath}
\usepackage{amssymb}
\usepackage{bibentry}
\usepackage{enumitem}
\usepackage{xytree}
\usepackage{cleveref}
\usepackage{pgfplotstable}
\usetikzlibrary{pgfplots.groupplots}
\usetikzlibrary{calc}

\crefname{property}{Property}{Properties}
\crefname{@theorem}{Theorem}{Theorems}
\crefname{claim}{Claim}{Claims}
\crefname{lemma}{Lemma}{Lemmas}
\crefname{fact}{Fact}{Facts}
\crefname{proposition}{Proposition}{Propositions}
\crefname{corollary}{Corollary}{Corollaries}
\crefname{Definition}{Definition}{Definitions}

\newcommand{\oracle}{\textsc{AIY}}
\newcommand{\tsweep}{\textsc{2-Sweep}}
\newcommand{\ifub}{\textsc{iFub}}

\newcommand{\ssh}{\textsc{SumSH}}
\renewcommand{\ss}{\textsc{SumS}}
\newcommand{\bcm}{\textsc{BCM}}
\newcommand{\rw}{\textsc{RW}}

\newcommand{\iid}{i.i.d.}
\newcommand{\aas}{a.a.s.}
\newcommand{\whp}{w.h.p.}

\newcommand{\farn}[1]{f(#1)}
\newcommand{\lfarn}[2]{\tilde{f}_{#1}(#2)}
\newcommand{\clos}[1]{c(#1)}
\newcommand{\C}{c}

\newcommand{\Td}[2]{T\left(#1 \rightarrow #2\right)}
\newcommand{\F}[2]{F\left(#1 \rightarrow #2\right)}

\newcommand{\pu}{\texttt{\upshape{P1}}}
\newcommand{\pd}{\texttt{\upshape{P2}}}

\newcommand{\N}[2]{\boldsymbol{N}^{#1}(#2)}

\newcommand{\G}[2]{\boldsymbol{\Gamma}^{#1}(#2)}
\newcommand{\g}[2]{\boldsymbol{\gamma}^{#1}(#2)}

\newcommand{\gu}[2]{\boldsymbol{\delta}_1^{#1}(#2)}
\newcommand{\gd}[2]{\boldsymbol{\delta}^{#1}(#2)}

\newcommand{\Gu}[2]{\boldsymbol{\Delta}_1^{#1}(#2)}
\newcommand{\Gd}[2]{\boldsymbol{\Delta}^{#1}(#2)}

\newcommand{\Nd}[2]{\boldsymbol{\Theta}^{#1}(#2)}
\newcommand{\nd}[2]{\boldsymbol{\theta}^{#1}(#2)}

\newcommand{\Iu}[1]{\boldsymbol{I}_1^{#1}}
\newcommand{\Id}[1]{\boldsymbol{I}_2^{#1}}

\newcommand{\au}{\boldsymbol{\alpha}_1}
\newcommand{\ad}{\boldsymbol{\alpha}_2}

\renewcommand{\b}{\boldsymbol{b}}

\newcommand{\n}[2]{\boldsymbol{n}^{#1}(#2)}

\renewcommand{\d}[2]{\boldsymbol{\delta}^{#1}(#2)}
\newcommand{\D}[2]{\boldsymbol{\Delta}^{#1}(#2)}

\renewcommand{\t}{\boldsymbol{\tau}}

\newcommand{\timp}[2]{\t_{#1}\left(#2\right)}
\newcommand{\timpd}[2]{\t_{#1}'\left(#2\right)}
\newcommand{\avedist}[3]{\dist_{\text{\upshape avg}}\left(#1\right)}

\newcommand{\w}[1]{\rho_{#1}}
\newcommand{\z}[1]{\boldsymbol{Z}^{#1}}
\newcommand{\zc}[1]{\tilde{\boldsymbol{Z}}^{#1}}
\newcommand{\Zc}{\tilde{\boldsymbol{Z}}}

\renewcommand{\o}{o}
\renewcommand{\O}{\mathcal O}
\renewcommand{\P}{\mathbb P}
\newcommand{\NN}{\mathbb N}
\newcommand{\E}{\mathbb E}

\newcommand{\X}{\boldsymbol{X}}
\renewcommand{\S}{\boldsymbol{S}}
\newcommand{\Y}{\boldsymbol{Y}}
\newcommand{\Z}{\boldsymbol{Z}}

\DeclareMathOperator{\dist}{dist}
\DeclareMathOperator{\poi}{Poisson}

\DeclareMathOperator{\var}{Var}
\DeclareMathOperator{\ecc}{ecc}

\renewcommand{\epsilon}{\varepsilon}

\date{}
\title{\Large An Axiomatic and an Average-Case Analysis of Algorithms and Heuristics for Metric Properties of Graphs\iffull{}{\thanks{The extended version of this paper, with full proofs of all the results, is available at \cite{Borassi2016d}.}}}

\author{Michele Borassi\thanks{IMT School for Advanced Studies Lucca} \and
Pierluigi Crescenzi\thanks{University of Florence}
\and
Luca Trevisan \thanks{EECS Department and Simons Institute, U.C. Berkeley}}
\date{}

\begin{document}

\maketitle
\sloppy

\begin{abstract}
In recent years, researchers proposed several algorithms that compute metric quantities of real-world complex networks, and that are very efficient in practice, although there is no worst-case guarantee. 

In this work, we propose an axiomatic framework to analyze the performances of these algorithms, by proving that they are efficient on the class of graphs satisfying certain properties. Furthermore, we prove that these properties are verified asymptotically almost surely by several probabilistic models that generate power law random graphs, such as the Configuration Model, the Chung-Lu model, and the Norros-Reittu model. Thus, our results imply average-case analyses in these models.

For example, in our framework, existing algorithms can compute the diameter and the radius of a graph in subquadratic time, and sometimes even in time $n^{1+\o(1)}$. Moreover, in some regimes, it is possible to compute the $k$ most central vertices according to closeness centrality in subquadratic time, and to design a distance oracle with sublinear query time and subquadratic space occupancy. 

In the worst case, it is impossible to obtain comparable results for any of these problems, unless widely-believed conjectures are false.
\end{abstract}

\section{Introduction.}

We study problems motivated by network analysis, such as computing the diameter of a graph, the radius, the closeness centrality, and so on. All these problems admit polynomial-time algorithms, based on computing the distance between all pairs of vertices. These algorithms, however, do not terminate in reasonable time if the input is a real-world graph with millions of nodes and edges. Such worst-case inefficiency is probably due to complexity-theoretic bottlenecks: indeed, a faster algorithm for any of these problems would falsify widely believed conjectures  \cite{Thorup2005,Roditty2013,Abboud2015a,Bergamini2015,Borassi2016,Abboud2016}.

In practice, these problems are solved via heuristics and algorithms that do not offer any performance guarantee, apart from empirical evidence. These algorithms are widely deployed, and they are implemented in major graph libraries, like Sagemath \cite{Stein2005}, Webgraph \cite{Boldi2004}, NetworKit \cite{Staudt2014}, and SNAP \cite{Snap}.

In this work, we develop a theoretical framework in which these algorithms can be evaluated and compared. Our framework is axiomatic in the sense that we define some properties, we experimentally show that these properties hold in most real-world graphs, and we perform a worst-case analysis on the class of graphs satisfying these properties. The purpose of this analysis is threefold: we validate the efficiency of the algorithms considered, we highlight the properties of the input graphs that are exploited, and we perform a comparison that does not depend on the specific dataset used for the evaluation. A further confirmation of the validity of this approach comes from the results obtained, that are very similar to existing empirical results.

Furthermore, we show that the properties are verified on some models of random graphs, asymptotically almost surely (\aas), that is, with probability that tends to $1$ as the number of nodes $n$ goes to infinity: as a consequence, all results can be turned into average-case analyses on these models, with no modification. This modular approach to average-case complexity analysis has two advantages: since our properties are verified by different models, we can prove results in all these models with a single worst-case analysis. Furthermore, we clearly highlight which properties of random graphs we are using: this way, we can experimentally validate the choice of the probabilistic model, by showing that these properties are reflected by real-world graphs.

In the past, most average-case analyses were performed on the Erd\"os-Renyi model, which is defined by fixing the number $n$ of nodes, and connecting each pair of nodes with probability $p$ \cite{Frieze1996,Rossman2010,Vijayaraghavan2012, Mann2013}. However many algorithms that work well in practice have poor average-case running time on this model.\footnote{The poor performances of some of these algorithms in the Erd\"os-Renyi model were empirically shown in \cite{Crescenzi2013}, and they can be proved with a simple adaptation of the analysis in this paper.} Indeed, these algorithms are efficient if there are some nodes with very high degree, and such nodes are not present in Erd\"os-Renyi graphs. Conversely, most real-world graphs contain such nodes, because their degree distribution is power law \cite{Barabasi1999}, that is, the number of vertices with degree $d$ is close to $\frac{n}{d^{\beta}}$ for some $\beta>1$. For this reason, we only consider models that generate power law random graphs. Our framework encompasses almost all values of $\beta$, and many of these models: the Configuration Model \cite{Bollobas1980}, and Rank-1 Inhomogeneous Random Graph models (\cite{Hofstad2014}, Chapter 3), such as the Chung-Lu \cite{Chung2006} and the Norros-Reittu model \cite{Norros2006}. 

Our approach is based on four properties: one simply says that the degree distribution is power law, and the other three study the behavior of $\timp s{n^x}$, which is defined as the smallest integer $\ell $ such that the number of vertices at distance $\ell $ from $s$ is at least $n^x$. The first of these properties describes the typical and extremal behavior of $\timp s{n^x}$, where $s$ ranges over all vertices in the graph. The next two properties link the distance between two vertices $s$ and $t$  with $\timp s{n^x}+\timp t{n^y}$: informally, $\dist(s,t)$ is close to $\timp s{n^x}+\timp t{n^{1-x}}$. We prove that these properties are verified in the aforementioned graph models.

The definition of these properties is one of the main technical contributions of this work: they do not only validate our approach, but they also provide a very simple way of proving other metric properties of random graphs, and their use naturally extends to other applications. Indeed, the proof of our probabilistic analysis is very simple, when one assumes these properties. On the other hand, the proof of the properties is very technical, and it uses different techniques in the regimes $\beta>2$, and $1<\beta<2$. In the regime $\beta>2$, the main technical tool used is branching processes: it is well-known that the size of neighborhoods of a given vertex in a random graph resembles a branching process \cite{Norros2006,Fernholz2007,Bollobas2007,Hofstad2014,Hofstad2014a}, but this tool was almost always used either as an intuition \cite{Fernholz2007,Hofstad2014,Hofstad2014a} (and different techniques were used in the actual proof), or it was applied only for specific models, such as the Norros-Reittu model \cite{Norros2006}. Conversely, in this work, we provide a quantitative result, from which we deduce the proof of the properties. In the regime $\beta<2$, the branching process approximation does not hold anymore (indeed, the distribution of the branching process is not even defined). For this reason, in the past, very few results were obtained in this case \cite{Esker2005}. In this work, we overcome this difficulty through a different technique: we prove that the graph contains a very dense ``core'' made by the nodes with highest degree, and the distance between two nodes $s,t$ is almost always the length of a shortest path from $s$ to the core, and from the core to $t$. This technique lets us compute the exact value of the diameter, and it lets us prove that the asymptotics found in \cite{Esker2005} for the Configuration Model also hold in other models.

Assuming the four properties, we can easily prove consequences on the main metric properties of the graphs $G=(V,E)$ under consideration: we start by estimating the eccentricity of a given vertex $s$, which is defined as $\ecc(s)=\max_{t \in V} \dist(s,t)$. From this result, we can estimate the diameter $D=\max_{s \in V} \ecc(s)$. Similarly, we can estimate the \emph{farness} $\farn s$ of $s$, that is, $\sum_{t \in V} \dist(s,t)$, the \emph{closeness centrality} of $s$, which is defined as $\frac{1}{\farn s}$, and the average distance between two nodes. By specializing these results to the random graph models considered, we retrieve known asymptotics for these quantities, and we prove some new asymptotics in the regime $1<\beta<2$.

After proving these results, we turn our attention to the analysis of many heuristics and algorithms, by proving all the results in \Cref{tab:summary} (a plot of the results is available in \Cref{fig:summary}).\footnote{Some of the results contain a value $\o(1)$: this value comes from the four properties, which depend on a parameter $\epsilon$. In random graphs, this notation is formally correct: indeed, we can let $\epsilon$ tend to $0$, since the properties are satisfied \aas\ for each $\epsilon$. In real-world graphs, we experimentally show that these properties are verified for small values of $\epsilon$, and with abuse of notation we write $\o(1)$ to denote a function bounded by $c\epsilon$, for some constant $c$.} For approximation algorithms, we usually know the running time and we analyze the error; conversely, for exact algorithms, we bound the running time. All algorithms analyzed are exactly the algorithms published in the original papers, apart from the \ssh\ and the \ss, where we need a small variation to make the analysis work.

\begin{table*}[htb]
\caption{a summary of the results of our probabilistic analyses. The value of the constant $C$ is $\frac{2\avedist n11}{D-\avedist n11}$, where $D$ is the diameter of the graph, $\avedist n11$ is the average distance. The values marked with $(*)$ are proved using further characteristics of the probabilistic models.}
\label{tab:summary}
\centering
\begin{footnotesize}
\begin{tabular}{|l|l|c|c|c|}
\hline
Parameter& Algorithm & \multicolumn{3}{c|}{Running time} \\
         &           & $\beta>3$ & $2<\beta<3$ & $1<\beta<2$, \\
         \hline
Diameter & BFS from  $n^\gamma$  & $\Theta(n^{1+\gamma})$ & $\Theta(n^{1+\gamma})$ & $\Theta(mn^{\gamma})$ \\
(lower bound)& random nodes &  $\epsilon_{rel}=\frac{1-\gamma+\o(1)}{2+C}$ & $\epsilon_{rel}=\frac{1-\gamma+\o(1)}{2}$ & $\epsilon_{abs}=\left\lfloor\frac{2(\beta-1)}{2-\beta}\right\rfloor-\left\lfloor\frac{(\gamma+1)(\beta-1)}{2-\beta}\right\rfloor$ \\
\hline
Diameter & \tsweep\ \cite{Magnien2009}  & $\Theta(n)$ & $\Theta(n)$ & $\Theta(m)$ \\
(lower bound)& &  $\epsilon_{rel}=\o(1)$ & $\epsilon_{rel}=\o(1)$ & $\epsilon_{abs} \leq \begin{cases}
1 & D \text{ even } \\
2 & D \text{ odd }
\end{cases}$ \\
\hline
Diameter & \rw\ \cite{Roditty2013}  & $\Theta(n^{\frac 32}\log n)$ & $\Theta(n^{\frac 32}\log n)$ & $\Theta(m\sqrt{n}\log n)$ \\
(lower bound)& &  $\epsilon_{rel}=\o(1)$ & $\epsilon_{rel}=\o(1)$ & $\epsilon_{abs} \leq \begin{cases}
1 & D \text{ even } \\
2 & D \text{ odd }
\end{cases}$ \\
\hline
All eccentricities & \ssh\ \cite{Borassi2014} & $n^{1+\o(1)}$ & $n^{1+\o(1)}$ & $\leq mn^{1-\frac{2-\beta}{\beta-1}\left(\left\lfloor\frac{\beta-1}{2-\beta}-\frac{3}{2}\right\rfloor-\frac{1}{2}\right)}$ \\
(lower bound) & &  $\epsilon_{abs}=0$ & $\epsilon_{abs}=0$ & $\epsilon_{abs}=0$ \\
\hline
Diameter & \ifub\ \cite{Crescenzi2013}    & $\leq n^{1+\left(\frac{1}{2}-\frac{1}{\beta-1}\right)C+\o(1)}$ & $n^{1+\o(1)}$ & $\leq mn^{1-\frac{2-\beta}{\beta-1}\left\lfloor \frac{\beta-1}{2-\beta}-\frac{1}{2}\right\rfloor+\o(1)}$ \\
\hline
Diameter & \ss\ \cite{Borassi2014,Borassi2014a}      & $\leq n^{1+\frac{C}{C+\frac{\beta-1}{\beta-3}}}$ $(*)$ & $n^{1+\o(1)}$ & $\leq mn^{1-\frac{2-\beta}{\beta-1}\left(\left\lfloor\frac{\beta-1}{2-\beta}-\frac{3}{2}\right\rfloor-\frac{1}{2}\right)}$ \\
\hline
Radius   & \ss\ \cite{Borassi2014,Borassi2014a}       & $n^{1+\o(1)}$    & $n^{1+\o(1)}$ &  $\leq mn^{1-\frac{2-\beta}{\beta-1}\left(\left\lfloor\frac{\beta-1}{2-\beta}-\frac{3}{2}\right\rfloor-\frac{1}{2}\right)}$ \\
\hline
Top-$k$ closeness & \bcm\ \cite{Bergamini2015} & $n^{2-\frac{1}{\beta-1}}$ $(*)$ & $n^{2-\o(1)}$ & $m^{1+\o(1)}$ \\
\hline
Distance oracle & \oracle\ \cite{Akiba2013} & $n^{1-\o(1)}$ & $\leq n^{f(\beta)}$ (*) & $\leq n^{\frac{1}{2}+\o(1)}$ \\
(query time) & &  & (no closed form) &  \\
(space needed) &  & $n^{2-\o(1)}$ & $\leq n^{1+f(\beta)}$ (*) & $\leq n^{\frac{3}{2}+\o(1)}$ \\
\hline
\end{tabular}
\end{footnotesize}
\end{table*}

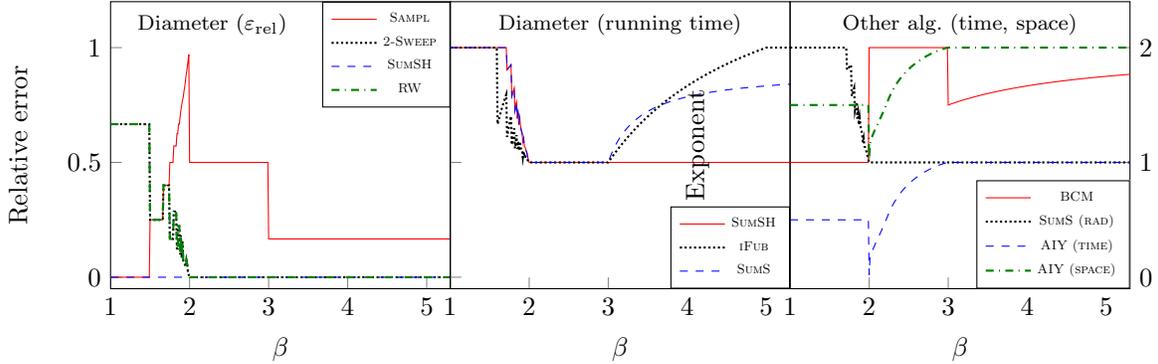
\begin{figure*}[htb]
\input{EfficiencyPlot}
\caption{plot of the running time and relative errors of the heuristics and algorithms considered. The constant $C$ was set to $3$, and the $\o(1)$ were ignored.} \label{fig:summary}
\end{figure*}

In many regimes, our results improve corresponding worst-case bounds: indeed, under reasonable complexity assumptions, for any $\epsilon>0$, there is no algorithm that computes a $\frac{3}{2}-\epsilon$-approximation of the diameter or the radius in $\O\left(n^{2-\epsilon}\right)$  \cite{Roditty2013,Borassi2016,Abboud2016}, the complexity of computing the most closeness central vertex is $\Omega\left(n^{2-\epsilon}\right)$ \cite{Abboud2016}, and there are hardness results on the possible tradeoffs between space needed and query time in distance oracles \cite{Thorup2005,Sommer2009}. The difference is very significant, both from a theoretical and from a practical point of view: for instance, we can compute the diameter and the radius of a graph in linear time, in many regimes. This means that, on standard real-world graphs with millions of nodes, the heuristic are millions of times faster than the standard quadratic algorithms.
It is also worth mentioning that our results strongly depend on the exponent $\beta$: in particular, there are two phase transitions corresponding to $\beta=2$ and $\beta=3$. This is due to the fact that, if $1<\beta<2$, the average degree is unbounded, if $2<\beta<3$, the average degree is finite, but the variance is unbounded, while if $\beta>3$ also the variance is finite. Furthermore, all the results with $\beta>3$ can be easily generalized to any degree distribution with finite variance, but the results become more cumbersome and dependent on specific characteristics of the distribution, such as the maximum degree of a vertex (for this reason, we focus on the power law case). Conversely, in the case $\beta<3$, our results strongly depend on the degree distribution to be power law, because random graphs generated with different degree distributions can have very different behaviors. The only open cases are $\beta=2$ and $\beta=3$, which are left for future work (note that, if $\beta \leq 1$, the degree distribution is not well defined).

\paragraph{Approximating the diameter.} We confirm the empirical results in \cite{Magnien2009}, proving that the \tsweep\ heuristic is significantly better than the basic sampling algorithm, which returns the maximum eccentricity of a random set of vertices. Furthermore, we show that the \ssh\ is even better than the \tsweep, confirming the experimental results in \cite{Borassi2014,Borassi2014a}. Finally, we analyze the well-known \rw\ algorithm, which provides a guaranteed $\frac{3}{2}$-approximation of the diameter in time $\Theta(m\sqrt{n})$. In our framework, it does not improve the \tsweep\ algorithm (which is much faster): this might theoretically explain why many graph libraries implement (variations of) the \tsweep\ heuristic, but not the \rw\ algorithm (for instance, Sagemath \cite{Stein2005}, Webgraph \cite{Boldi2004}, NetworKit \cite{Staudt2014}).

\paragraph{Computing the diameter.}The aforementioned heuristics can be turned into exact algorithms, that always provide the correct result, but that can be inefficient in the worst case. We analyze two of these algorithms, proving that, for small values of $\beta$, both the \ifub\ and the \ss\ algorithm are very efficient; for big values of $\beta$, the \ss\ algorithm is usually better, because it is always subquadratic. These results explain the surprisingly small running time on most graphs, and the reason why the \ss\ algorithm is usually faster on ``hard'' instances, as observed in \cite{Borassi2014,Borassi2014a}. 
It is interesting to note that all the running times for $\beta>3$ depend on the same constant $C$, which we prove to be close to $\frac{2\avedist n11}{D-\avedist n11}$, where $D$ is the diameter and $\avedist n11$ is the average distance of two nodes in the input graph. Intuitively, if this ratio is small, it means that there are ``few far vertices'', and the algorithms are quite efficient because they only need to analyze these vertices (the only exception is the sampling algorithm, which is not able to find these vertices, and hence achieves better performances when $C$ is large). For $2<\beta<3$, a very similar argument applies, but in this case $C=0$, because $D=\O(\log n)$ and $\avedist n11 = \O(\log\log n)$.

\paragraph{Other algorithms.} Our framework lets us also analyze algorithms for computing other quantities. For example, the \ss\ algorithm is also able to compute the radius: in this case, in all regimes, the running time is almost linear, confirming the results in \cite{Borassi2014a}, where it is shown that the algorithm needed at most $10$ BFSes on all inputs but one, and in the last input it needed $18$ BFSes. The other two algorithms analyzed are the \bcm\ algorithm, to compute the $k$ most central vertices according to closeness centrality \cite{Borassi2015b}, and the distance oracle \oracle\ in \cite{Akiba2013}. In the first case, we show significant improvements with respect to the worst-case in the regime $1<\beta<2$ and $\beta>3$, and we show that the algorithm is not efficient if $2<\beta<3$: this is the \emph{only} result in this paper which is not reflected in practice. The problem is that our analysis relies on the fact that $\avedist n11=\Theta(\log\log n)$ tends to infinity, but the experiments were performed on graphs where $n<10\,000\,000$, and consequently $\log \log (n)<4$. The last probabilistic analysis confirms the efficiency of \oracle: we show that, if $\beta<3$, the expected time needed to compute the distance between two random nodes is sublinear, and the space occupied is subquadratic.

Finally, as a side result of our analysis, we compute for the first time the diameter of random graph in the regime $1<\beta<2$.

\subsection{Related Work.}

This work combines results in several research fields: the analysis of random graphs, axiomatic approaches in the study of social networks, the design of heuristics and algorithms that are efficient on real-world graphs, the average and worst-case analysis of algorithms. Since it is impossible to provide a comprehensive account of the state-of-the-art in all these areas, here we just point the reader to the most recent and comprehensive surveys.

There are several works that study metric properties of random graphs: most of these results are summarized in \cite{Hofstad2014,Hofstad2014a}. In the regime $\beta>2$, we take inspiration from the proofs in \cite{Fernholz2007} for the Configuration Model, and in \cite{Norros2006,Bollobas2007} for Inhomogeneous Random Graphs. In this setting, we give a formal statements that links neighborhood sizes with branching processes (\iffull{\Cref{thm:branchprocess}}{see the full version of this paper \cite{Borassi2016d}}): although it was used very often as a heuristic argument \cite{Hofstad2014,Hofstad2014a}, or used in specific settings \cite{Norros2006}, as far as we know, it was never formalized in this general setting. Furthermore, in the regime $1<\beta<2$, we use new techniques to prove the four properties, and as a result we obtain new asymptotics for the diameter. As far as we know, the only work that addresses the latter case is \cite{Esker2005}, which only computes the typical distance between two nodes.

Furthermore, our work relies on several works that outline the main properties of complex networks, and that develop models that satisfy such properties: for example, the choice of the power law degree distribution is validated by extensive empirical work (see \cite{Newman2010} for a survey).

Despite this large amount of research on models of real-world graphs, few works have tried to address the problem of evaluating heuristics and algorithms on realistic models. For example, several works have addressed the efficiency of computing shortest paths in road networks \cite{Goldberg2005,Sanders2005,Delling2013,Delling2013a}. In \cite{Abraham2010}, the authors provide an explanation of their efficiency, based on the concept of \emph{highway dimension}. Another example is the algorithm in \cite{Boldi2004}, which is used to compress a web graph: in \cite{Chierichetti2009}, the authors prove that in most existing models no algorithm can achieve good compression ratio, and they provide a new model where the algorithm in \cite{Boldi2004} works well. Also in \cite{Gupta2016}, the authors develop an axiomatic framework, but they study triangle density, and not distances (that is, they assume that the input graph contains many triangles, a characteristic that is shared by most real-world graphs). That paper sets forth the research agenda of defining worst-case conditions on graphs generalizing all the popular generative models: it discusses the main advantages and disadvantages of the approach, and it leaves as an open problem to find algorithms that are more efficient on the class of triangle-dense graphs. Another related work is \cite{Brach2016}, where the authors develop an axiomatic approach that is similar to ours: assuming only that the degree distribution is power law, they manage to analyze some algorithms, and to prove that their analysis improves the worst-case analysis. Our work is orthogonal to their work: indeed, they only assume the degree distribution to be power law, using a variation of our \Cref{ax:deg}. Their properties are weaker than ours, since they only assume a variation of our \Cref{ax:deg}: for this reason, they manage to analyze algorithms that compute local properties, such as patterns in subgraphs, but not the global metric properties considered in this paper (indeed, graphs with the same degree distribution can have very different metric properties). An approach that is more similar to ours is provided in \cite{Borassi2016b}: among other results, it is proved that it is possible to compute a shortest path between two nodes $s,t$ in sublinear time, in the models considered in this paper. The running-time is $\O(n^{\frac{1}{2}+\epsilon})$ if $\beta>3$, and $\O(n^{\frac{4-\beta}{2}})$ if $2<\beta<3$ (while, in the worst-case, this task can be performed in $\O(n)$).

Finally, some works have tried to explain and motivate the efficiency of some heuristics and algorithms for diameter and radius computation. A first attempt uses the \emph{Gromov hyperbolicity} of the input graph \cite{Gromov1987}: for example, the \tsweep\ heuristic provides good approximations of the diameter of hyperbolic graphs \cite{Chepoi2008}. However, this approach cannot be applied to some algorithms, like the \ifub, and when it can be applied, the theoretical guarantees are still far from the empirical results, because real-world graphs are usually not hyperbolic according to Gromov's definition \cite{Borassi2015a}.

\subsection{Structure of the Paper.}

In \Cref{sec:axioms}, we state the four properties considered; in \Cref{sec:overviewmain}, we define the models considered and we sketch the proof that they satisfy the four properties. In \Cref{sec:experiments} we experimentally show that they are satisfied by real-world graphs. Then, in \Cref{sec:diamecc}, we prove some consequences of the axioms, that are extensively used throughout the paper, such as asymptotics for the diameter, the average distance, etc. In \Cref{sec:sampl,sec:tsweep,sec:rw,sec:ssh,%
sec:ifub,sec:ss}, we perform the probabilistic analysis for diameter and radius algorithms (for the other two algorithms, \bcm\ and \oracle, we refer to \iffull{\Cref{sec:bcm,sec:oracle}}{the full version of this paper \cite{Borassi2016d}}). \Cref{sec:conc} concludes the paper.

\section{The Four Properties.} \label{sec:axioms}

In this section, we define the four properties used in our framework. Let us start with some definitions.

\begin{Definition}
Given a graph $G=(V,E)$, if $s \in V$, let $\G \ell s$ be the set of vertices at distance exactly $\ell $ from $s$, let $\g \ell s=|\G \ell s|$, let $\N \ell s$ be the set of vertices at distance at most $\ell $ from $s$, and let $\n \ell s=|\N \ell s|$. We define $\timp sk=\min\{\ell  \in \NN:\g \ell s>k\}$, and $\Td dk$ as the average number of steps for a node of degree $d$ to obtain a neighborhood of $k$ nodes. More formally, $\Td dk$ is the average $\timp sk$ over all vertices $s$ of degree $d$ (note that, since the diameter is $\O(\log n)$, $\Td d{n^x}$ is defined for each $x<1$).
\end{Definition}

Our properties depend on a parameter $\epsilon$: for instance, the first property bounds the number of vertices such that $\timp s{n^x} \geq (1+\epsilon)\Td d{n^x}$. Intuitively, one can think of $\epsilon$ as a constant which is smaller than any other constant appearing in the proofs, but bigger than $\frac{1}{n}$, or any other infinitesimal function of $n$. Indeed, in random graphs, we prove that if we fix $\epsilon,\delta>0$, we can find $n_{\epsilon,\delta}$ such that the properties hold for each $n>n_{\epsilon,\delta}$, with probability at least $1-\delta$. In real-world graphs, we experimentally show that the four properties are verified with good approximation for $\epsilon=0.2$. In our analyses, the time bounds are of the form $n^{c+\O(\epsilon)}$, and the constants in the $\O$ are quite small. Since, in our dataset, $n^{0.2}$ is between $6$ and $19$, we can safely consider $n^{c+\O(\epsilon)}$ close to $n^{c}$.

The first property analyzes the typical and extremal values of $\timp s{n^x}$, where $s$ is any vertex.

\begin{property} \label{ax:dev}
There exists a constant $\C$ such that:
\begin{itemize}
\item for each vertex $s$ with degree $d>n^\epsilon$, $\timp s{n^x} \leq (1+\epsilon)\left(\Td d{n^x}+1\right)$;
\item the number of vertices verifying $\timp s{n^x} \geq (1+\epsilon)\left(\Td d{n^x}+\alpha\right)$ is $\O\left(n\C^{\alpha-x}\right)$;
\item the number of vertices verifying $\timp s{n^x} \geq (1-\epsilon)\left(\Td 1{n^x}+\alpha\right)$ is $\Omega\left(n\C^{\alpha-x}\right)$. 
\end{itemize}
\end{property}


In random graphs, the values of $\Td d{n^x}$ depend on the exponent $\beta$ (see \Cref{tab:tc}). In many of our analyses, we do not use the actual values of $\Td d{n^x}$, but we use the following properties:
\begin{itemize}
\item $\Td d{n^{x+\epsilon}} \leq \Td d{n^{x}}(1 + \O(\epsilon))$;
\item $\sum_{d=1}^\infty |\{v \in V: \deg(v)=d\}| \Td {d}{n^{x}}=(1+\o(1))n\Td 1{n^x}$;
\item $\Td {1}{n^{x}}+\Td{1}{n^{1-x}}-1=(1+\o(1))\avedist n{1}{1}$, where $\avedist n{1}{1}$ is a function not depending on $x$ (this function is very close to the average distance, as we prove in \Cref{sec:diamecc}).
\end{itemize}
\begin{table*}[htb]
\caption{the values of $\Td d{n^x}$, $\avedist n{d_1}{d_2}$ and $c$, depending on the value of $\beta$.}
\label{tab:tc}
\centering
\begin{tabular}{|l|c|c|c|}
\hline
Regime & $\Td d{n^x}$ & $\avedist n{d_1}{d_2}$ & $\C$ \\
\hline
$1<\beta<2$ & $1$ if $d\geq {n^x}$, $2$ otherwise & $3$ & $n^{-\frac{2-\beta}{\beta-1}(1+\o(1))}$ \\
\hline
$2<\beta<3$ & $(1+\o(1))\log_{\frac{1}{\beta-2}} \frac{\log {n^x}}{\log d} \text{ if } {n^x}<n^{\frac{1}{\beta-1}}$ & $(2+\o(1))\log_{\frac{1}{\beta-2}} \log n$ & $\eta(1)+\o(1)$ \\
 & $(1+\o(1))\log_{\frac{1}{\beta-2}} \frac{\log {n^x}}{\log d}+\O(1) \text{ if } {n^x}>n^{\frac{1}{\beta-1}}$ & & \\
\hline
$\beta>3$ & $(1+\o(1))\log_{M_1(\mu)} \frac{{n^x}}{d}$ & $(1+\o(1))\log_{M_1(\mu)} n$ & $\eta(1)+\o(1)$ \\
\hline
\end{tabular}
\end{table*}

The next two properties relate the distance between two vertices $s,t$ with the values of $\timp s{n^x}$, $\timp t{n^y}$, where $x,y$ are two reals between $0$ and $1$. The idea behind these two properties is to apply the ``birthday paradox'', assuming that $\G{\timp s{n^x}}s$ and $\G{\timp t{n^y}}t$ are random sets of $n^x$ and $n^y$ vertices. In this idealized setting, if $x+y>1$, there is a vertex that is common to both, and $\dist(s,t) \leq \timp s{n^x}+\timp t{n^y}$; conversely, if $x+y<1$, $\dist(s,t)$ is likely to be bigger than $\timp s{n^x}+\timp t{n^y}$. Let us start with the simplest property, which deals with the case $x+y>1$.

\begin{property} \label{ax:touch}
Let us fix two real numbers $0<x,y<1$ such that $x+y>1+\epsilon$. For each pair of vertices $s,t$, $\dist(s,t) < \timp s{n^x}+\timp t{n^y}$.
\end{property}


The next property is a sort of converse: the main idea is that, if the product of the size of two neighborhoods is smaller than $n$, then the two neighborhoods are usually not connected. The simplest way to formalize this is to state that, for each pair of vertices $s,t$, $\dist(s,t) \geq \timp s{n^x}+\timp t{n^y}$. However, there are two problems with this statement: first, in random graphs, if we fix $s$ and $t$, $\dist(s,t) \geq \timp s{n^x}+\timp t{n^y}$ \aas, not \whp, and hence there might be $\o(n)$ vertices $t$ such that $\dist(s,t) < \timp s{n^x}+\timp t{n^y}$ (for example, if $s$ and $t$ are neighbors, they do not verify $\dist(s,t) \geq \timp s{n^x}+\timp t{n^y}$). To solve this, our theorem bounds the number of vertices $t$ verifying $\dist(s,t) \geq \timp s{n^x}+\timp t{n^y}$. The second problem is more subtle: for example, if $s$ has degree $1$, and its only neighbor has degree $n^{\frac{1}{2}}$, $\timp s{n^{\frac{1}{4}}}=\timp s{n^{\frac{1}{2}}}=2$, and the previous statement cannot hold for $x=\frac{1}{4}$. However, this problem does not occur if $x \geq y$: the intuitive idea is that we can ``ignore'' vertices with degree bigger than $n^x$. Indeed, if a shortest path from $s$ to $t$ passes through a vertex $v$ with degree bigger than $n^x$, then $\timp s{n^x} \leq \dist(s,v)+1$, $\timp t{n^y} \leq \dist(t,v) + 1$, and hence $\dist(s,t)=\dist(s,v)+\dist(v,t) \geq \timp s{n^x}+\timp t{n^y}-2$.

\begin{property} \label{ax:untouch}
Let $s$ be any vertex, let $0<z\leq y<x<1$, let $x+y \geq 1+\epsilon$, and let $\alpha$, $\omega$ be integers. If $T_{\alpha,\omega,z}$ is the set of vertices $t$ such that $\timp t{n^{z}}$ is between $\alpha$ and $\omega$, there are at most $|T_{\alpha,\omega,z}|\frac{n^{x+y+\epsilon}}{n}$ vertices $t \in T$ such that $\dist(s,t) < \timp s{n^x}+\timp t{n^y}-2$.
\end{property}


Finally, in some analyses, we also need to use the fact that the degree distribution is power law. To this purpose, we add a further property (in random graphs, this result is well-known \cite{Hofstad2014,Hofstad2014a}).

\begin{property} \label{ax:deg}
The number of vertices with degree bigger than $d$ is $\Theta\left(\frac{n}{d^{\max(1,\beta-1)}}\right)$.
\end{property}



Although the definition of the four properties is quite complicated, the intuition is natural. Indeed, \Cref{ax:touch,ax:untouch} simply say in a formal way that $\dist(s,t)\approx \timp s{n^x}+\timp t{n^{1-x}}$, and this is the property which is used in all the probabilistic analysis. As far as we know, in the context of the analysis of real-world graphs, this property was never stated or formalized: we believe that it can give further insight in the field of the analysis of real-world graphs. A further confirmation of the importance of this property is that the algorithms considered are not very efficient on graphs where this property is not satisfied, such as road networks \cite{Borassi2014,Borassi2014a,Bergamini2015}.

Conversely, \Cref{ax:dev,ax:deg} are more specific, and they are specifically suited to the analysis of the real-world networks and the random graphs under consideration. They were chosen because they are satisfied by the graphs under consideration, but one might be interested in using variations of these properties on different kinds of networks, since the proofs usually do not depend on the specific values of the parameters considered.

\section{Validity of the Properties in Random Graphs: Overview.} \label{sec:overviewmain}

In order to transform the axiomatic worst-case analyses into average-case analyses on random graphs, we use the following theorem.

\begin{theorem} \label{thm:main}
For each fixed $\epsilon>0$, \Cref{ax:dev,ax:touch,ax:untouch,ax:deg} are verified in the random graphs defined in all the models considered, \aas.
\end{theorem}

In other words, for each $\epsilon,\delta>0$, there exists $n_{\epsilon,\delta}$ such that the probability that a random graph with $n>n_{\epsilon,\delta}$ nodes does not verify the four properties is at most $1-\delta$.

In this section, we sketch the proof of this theorem, while we provide the complete proof in \iffull{\Cref{sec:proof}}{the full version of the paper \cite{Borassi2016d}}.

\subsection{The Models}

The models considered are the Configuration Model (CM) and Rank-1 Inhomogeneous Random Graphs (IRG), such as the Norros-Reittu model and the Chung-Lu model. All these models fix a set $V$ of $n$ vertices, and they assign a weight $\w v$ to each vertex $v \in V$ (we choose the weights $\w v$ according to a power law distribution with exponent $\beta$). Then, we create edges in a way that the degree of $v$ is close to $\w v$: in the CM, this is done by associating to $v$ $\w v$ half-edges, and pairing these half-edges at random, while in IRG, an edge between vertices $v$ and $w$ exists with probability close to $\frac{\w v \w w}{M}$, where $M=\sum_{v \in V} \w v$.

Furthermore, we need to consider only the \emph{giant component} of the graph considered, and, differently from other works, we do not assume the graph generated through the CM to be simple (anyway, multiple edges and self-loops have no effect on distances). For more details of the models considered, and for some additional technical assumptions used to avoid pathological cases, we refer to \iffull{\Cref{sec:model}}{the full version of this paper \cite{Borassi2016d}}.

\subsection{\Cref{ax:touch,ax:untouch,ax:deg}.}

It is quite easy to prove that \Cref{ax:deg} holds: indeed, it is enough to show that the degree of a vertex $v$ is close to its weight $\w v$, and this can be done through a Chernoff-type probability bound. 

Then, we need to prove that \Cref{ax:touch,ax:untouch} hold: these two properties bound $\dist(s,t)$ with $\timp s{n^x}+\timp t{n^y}$. Let us assume that $\g \ell s=n^x$, and $\g{\ell '}t=n^y$: if all vertices are in $\G \ell s$ with the same probability, $\G \ell s$ will be a random subset of the set of vertices, and the probability that a vertex in $\G \ell s$ is also in $\G{\ell '}t$ is close to $\frac{\g{\ell '}t}{n}=\frac{1}{n^{1-y}}$. Hence, the probability that $\dist(s,t) \geq \ell +\ell '$ is related to the probability that $\G \ell s$ does not intersect $\G{\ell '}t$, which is close to $\left(1-\frac{1}{n^{1-y}}\right)^{n^x} \approx e^{-n^{x+y-1}}$. For $x+y>1$, this means that $\dist(s,t) \leq \ell +\ell '$ \whp, and this is very close to the statement of \Cref{ax:touch}. For $x+y<1$, $\G \ell s$ does not intersect $\G{\ell '}t$ with probability $e^{-n^{x+y-1}}\approx 1-n^{x+y-1}$, and hence $\dist(s,t) \leq \ell +\ell '$ with probability close to $n^{x+y-1}$. The proof that \Cref{ax:untouch} holds is then concluded by applying concentration inequalities, exploiting the fact that $T$ is ``enough random''.

\subsection{\Cref{ax:dev}, $\beta>2$.}

The proof that \Cref{ax:dev} holds is much more complicated: in the proof, we have to distinguish between the case $\beta<2$ and $\beta>2$. In the case $\beta>2$, we use two different techniques.
\begin{enumerate}
\item When $\g \ell s=|\G \ell s|$ is small (say, smaller than $n^\epsilon$), we show that the behavior of $\g \ell s$ is well approximated by a $\mu$-distributed branching process, where $\mu$ is the residual distribution of $\lambda$ (the definition of residual distribution depends on the model, and it is provided in \iffull{\Cref{def:residual}}{the full version of the paper \cite{Borassi2016d}}). Furthermore, if $s$ and $t$ are two different vertices, and if $\g \ell s$ and $\g {\ell '}t$ are small,  the behavior of $\G \ell s$ and the behavior of $\G{\ell '}t$ are ``almost'' independent. \label{item:smallneigh}
\item When $\g \ell s$ is large, the branching process approximation and the independence do not hold anymore. We need a different technique: since $\g \ell s>n^\epsilon$, a Chernoff-type probability bound gives guarantees of the form $e^{-n^\epsilon}$, which is bigger than any polynomial in $n$. This way, we can prove very precise bounds on the size of $\g{\ell +1}s$ given the size of $\g \ell s$, and through a union bound we can show that these bounds hold for any vertex $s$. \label{item:bigneigh}
\end{enumerate}

The second technique was already used in some works \cite{Chung2006,Norros2006,Fernholz2007}; however, the formalization of the connection between neighborhood expansion and branching processes is original\iffull{ (\Cref{thm:branchprocess})}{}, it formalizes existing intuitive explanations \cite{Hofstad2014,Hofstad2014a}, and it generalizes proofs that were performed in restricted classes of models \cite{Norros2006,Bollobas2007}. Let us provide some more details: we define a branching process $\gd \ell s$ coupled with $\g \ell s$ (that is, $\g \ell s$ and $\gd \ell s$ are defined on the same probability space, and the probability that they are equal is high). Then, we analyze the size of $\gd \ell s$: if the first moment $M_1(\mu)$ of the distribution $\mu$ of the branching process is finite (or, equivalently, if $M_2(\lambda)$ is finite), it is well known \cite{Athreya1972} that the expected size of $\gd \ell s$ is $\gd 1s M_1(\mu)^{\ell -1}=\deg(s)M_1(\mu)^{\ell -1}$; if $\lambda$ is a power law distribution with $2<\beta<3$, the typical size of $\gd \ell s$ is close to $\gd 1s^{\left(\frac{1}{\beta-2}\right)^{\ell -1}}=\deg(s)^{\left(\frac{1}{\beta-2}\right)^{\ell -1}}$. Hence, heuristically, we can estimate $\timp s{n^x}$, by setting $\deg(s)M_1(\mu)^{\ell -1}=n^x$ if $M_1(\mu)$ is finite and strictly bigger than $1$, and $\deg(s)^{\left(\frac{1}{\beta-2}\right)^{\ell -1}}=n^x$ if $\mu$ is power law with exponent $1<\beta<2$. Solving with respect to $\ell$, we obtain the values in \Cref{tab:tc}. 

Through a more refined analysis, we can use the branching process approximation to estimate the deviations from these value: first, we remove from the branching process all branches that have a finite number of descendants, since they have little impact on the total size of the branching process (if the whole branching process is finite, it means that the starting vertex is not in the \emph{giant component}, and we can ignore it). It is proved in \cite[1.D.12]{Athreya1972} that we obtain another branching process, with distribution $\eta$ that depends only on $\mu$, and such that $\eta(0)=0$, so that all branches are infinite. Then, we prove that the ``worst'' that can happen is that $\gd \ell s$ is $1$ for a long time, and then it grows normally: this means that $\P\left(\timp s{n^x}>\Td{\deg(s)}{n^x}+k\right) \approx \eta(1)^k$, and, since the growths of different vertices are almost independent, we obtain that the number of vertices verifying $\timp s{n^x}>\Td{\deg(s)}{n^x}+k$ is approximately $n\eta(1)^k$.

Summarizing, we sketched the proof that the values appearing in \Cref{tab:tc} are correct, and that \Cref{ax:dev} holds, at least when $x$ is small. For big values of $x$, the branching process approximation does not hold anymore: however, as soon as $\g \ell s$ is large enough, we can prove directly that $\g{\ell +1}s\approx \g \ell sM_1(\mu)$ if $M_1(\mu)$ is finite, and $\g{\ell +1}s \approx \g \ell s^{\frac{1}{\beta-2}}$ if $\lambda$ is power law with exponent $2<\beta<3$, \whp. This way, we can prove results on $\timp s{n^x}$ by proving the same results for $\timp s{n^y}$ for some small $y$, and extending the result to $\timp s{n^x}$ using this argument. This concludes the proof that the values appearing in \Cref{tab:tc} are correct, and that \Cref{ax:dev} holds.

\subsection{\Cref{ax:dev}, $\beta<2$.}

In this case, the branching process approximation does not hold: indeed, the residual distribution $\mu$ cannot be even defined! We use a completely different technique. First, we consider the $N$ vertices with highest weight, where $N$ is a big constant: using order statistics, we can prove that each of these vertices has weight $\Theta(M)$, where $M=\sum_{v \in V}\w v$. From this, we can prove that each vertex with degree at least $n^\epsilon$ is connected to each of these $N$ vertices, and these $N$ vertices have degree $\Theta(n)$. This is enough to characterize the size of neighbors of any vertex $v$ with degree bigger than $n^\epsilon$: there are $\deg(v)$ vertices at distance $1$ and $\Theta(n)$ vertices at distance $2$.

Let us now consider the neighborhood growth of other vertices: given a vertex $v$, the probability that it is not connected to any vertex $w$ with weight smaller than $n^\epsilon$ is approximately $\prod_{\w w <n^\epsilon}\left(1-\frac{\w v\w w}{M}\right) \approx 1-\frac{\w v}{M}\sum_{\w w <n^\epsilon} \w w\approx 1-\frac{n}{n^{\frac{1}{\beta-1}}}$ (it is possible to prove that $M\approx n^{\frac{1}{\beta-1}}$). As a consequence, the probability that a vertex $v$ is connected to another vertex with weight $w<n^\epsilon$ is quite small, being approximately $n^{-\frac{2-\beta}{\beta-1}}=\C$. Let us consider three cases separately.
\begin{enumerate}
\item If $v$ is connected to a vertex $w$ with degree at least $n^\epsilon$, we deduce results on neighbors of $v$ from results on neighbors of $w$.
\item If $v$ is not connected to a vertex $w$ with degree at least $n^\epsilon$, the following cases might occur:
\begin{enumerate}
\item if $v$ is not connected to a vertex with weight smaller than $n^\epsilon$, we can ignore it, because it is not in the giant component;
\item the last case is that $v$ is connected to another vertex $w$ with weight smaller than $n^\epsilon$, which occurs with probability $\C$; in this case, we iterate our argument with $w$, until we hit a vertex with degree at least $n^\epsilon$.\label{item:iter}
\end{enumerate}
\end{enumerate}

In particular, the probability that \Cref{item:iter} occurs $\ell$ times before hitting a vertex with degree at least $n^\epsilon$ is approximately $\C^\ell$: this means that the number of vertices whose neighbors reach size $n^x$ after $\ell$ steps is at most $n\C^{\ell+\O(1)}$. Through a more thorough analysis of the constant $\O(1)$, we obtain the results in \Cref{tab:tc}, proving upper bounds for \Cref{ax:dev}. For lower bounds, surprisingly, we only have to consider vertices with degree $1$ and $2$: in particular, the probability that a vertex with degree $1$ is linked to another vertex of degree $2$ turns out to be approximately $\C$. For this reason, there are at least $n\C^{\ell}$ vertices of degree $1$ that are starting points of a path of length $\ell$, which terminates in a vertex with larger degree. This concludes the proof that \Cref{ax:dev} holds.

\section{Validity of the Properties in Real-World Graphs.} \label{sec:experiments}

\begin{table*}[t!]
\centering
\begin{tabular}{|l|r|r|r|r|r|r|r|r|r|}
\hline
 Network               & $n^{0.2}$ & Vert. & $k=-2$ & $k=-1$ & $k=0$ & $k=1$ & $k=2$ & $k=3$  & $k=4$  \\
\hline
 p2p-Gnutella09         & 6.1 & 2811 & 0.00\% & 61.37\% & 38.63\% & 0.00\% & 0.00\% & 0.00\% & 0.00\% \\ 
 oregon1-010526         & 6.5 & 640 & 0.00\% & 58.75\% & 41.25\% & 0.00\% & 0.00\% & 0.00\% & 0.00\% \\ 
 ego-gplus              & 7.5 & 348 & 0.00\% & 2.87\% & 97.13\% & 0.00\% & 0.00\% & 0.00\% & 0.00\% \\ 
 oregon2-010526         & 6.5 & 1113 & 0.00\% & 55.17\% & 44.83\% & 0.00\% & 0.00\% & 0.00\% & 0.00\% \\ 
 ca-HepTh               & 6.1 & 1987 & 2.21\% & 48.97\% & 43.48\% & 4.98\% & 0.25\% & 0.00\% & 0.10\% \\ 
 ca-CondMat             & 7.3 & 6519 & 0.00\% & 45.25\% & 51.20\% & 3.27\% & 0.23\% & 0.05\% & 0.00\% \\ 
 ca-HepPh               & 6.5 & 4644 & 0.00\% & 46.32\% & 50.39\% & 2.84\% & 0.45\% & 0.00\% & 0.00\% \\ 
 email-Enron            & 8.0 & 6354 & 0.00\% & 69.00\% & 30.33\% & 0.66\% & 0.02\% & 0.00\% & 0.00\% \\ 
 loc-brightkite   & 8.9 & 9929 & 0.00\% & 69.45\% & 29.94\% & 0.42\% & 0.18\% & 0.00\% & 0.00\% \\ 
 email-EuAll            & 11.8 & 2654 & 0.00\% & 59.08\% & 40.66\% & 0.23\% & 0.00\% & 0.00\% & 0.04\% \\ 
 ca-AstroPh             & 7.1 & 9812 & 0.00\% & 58.55\% & 41.10\% & 0.18\% & 0.16\% & 0.00\% & 0.00\% \\ 
 gowalla-edges          & 11.5 & 33263 & 0.00\% & 65.69\% & 34.07\% & 0.23\% & 0.01\% & 0.00\% & 0.00\% \\ 
 munmun-twitter  & 13.6 & 6670 & 0.00\% & 70.57\% & 29.43\% & 0.00\% & 0.00\% & 0.00\% & 0.00\% \\ 
 com-dblp               & 12.6 & 33363 & 1.65\% & 63.03\% & 32.41\% & 2.57\% & 0.32\% & 0.01\% & 0.00\% \\ 
 com-lj.all.cmty        & 12.5 & 5258 & 0.51\% & 65.96\% & 32.98\% & 0.53\% & 0.02\% & 0.00\% & 0.00\% \\ 
 enron                  & 9.7 & 7792 & 0.00\% & 77.71\% & 21.79\% & 0.37\% & 0.13\% & 0.00\% & 0.00\% \\ 
 com-youtube            & 16.3 & 46471 & 0.00\% & 79.01\% & 20.32\% & 0.45\% & 0.15\% & 0.04\% & 0.02\% \\ 
 wiki-Talk              & 18.9 & 27536 & 0.00\% & 62.63\% & 37.37\% & 0.00\% & 0.00\% & 0.00\% & 0.00\% \\
 \hline
\end{tabular}
\caption{the percentage of vertices with degree at least $n^{0.2}$ that verify $\timp s{n^{\frac{1}{2}}}-\left\lceil\Td d{n^{\frac{1}{2}}}\right\rceil=k$ (the other values of $k$ are $0$, for each graph in the dataset).} \label{tab:devbig}
\end{table*}

\begin{figure*}[t!]
\input{Tail}
\caption{the percentage of vertices verifying $\timp s{n^{\frac{1}{2}}}-\Td {\deg(s)}{n^{\frac{1}{2}}} \geq k$, in all the graphs in our dataset.} \label{fig:devsm}
\end{figure*}
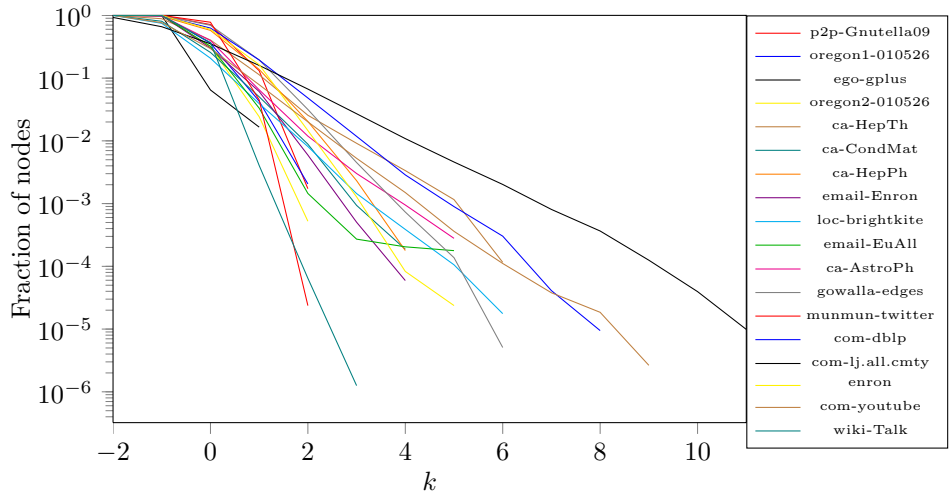

In this section, we experimentally show that the first three properties hold in real-world graphs, with good approximation (we do not perform experiments on the fourth property, because it is well known that the degree distribution of many real-world graphs is power law \cite{Barabasi1999,Newman2003}). To this purpose, we consider a dataset made by 18 real-world networks of different kinds (social networks, citation networks, technological networks, and so on), taken from the well-known datasets SNAP (\url{snap.stanford.edu/}) and KONECT (\url{http://konect.uni-koblenz.de/networks/}). Then, for each of the properties, we compute the quantities considered, on all graphs in the dataset, and we show that the actual behavior reflects the predictions.

We start with \Cref{ax:dev}: to verify the first claim, we consider all vertices with degree at least $n^{0.2}$, which is between $6$ and $19$ in our inputs. For each of these vertices, we compute $\timp s{n^{\frac{1}{2}}}-\Td {\deg(s)}{n^{\frac{1}{2}}}$ (in this paper, we show the results for $x=\frac{1}{2}$, but very similar results hold for all values of $x$). The results obtained are represented in \Cref{tab:devbig}.

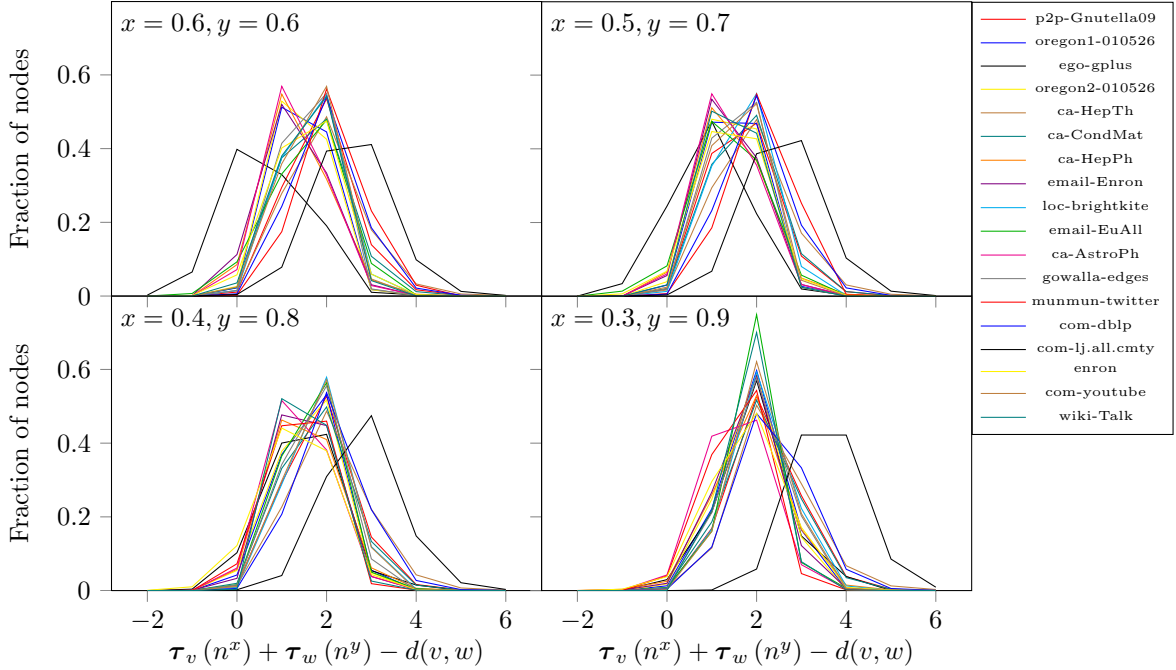
\begin{figure*}
\input{Touch}
\caption{the values of $\timp s{n^x}+\timp t{n^y}-\dist(s,t)$ for $10\,000$ pairs of vertices in each graph.} \label{fig:plottouch}
\end{figure*}

The table shows that in all the graphs considered, the first statement of \Cref{ax:dev} is verified with good approximation: almost all vertices with degree at least $n^{0.2}$ verify $\timp s{n^{\frac{1}{2}}}-\left\lceil\Td  {\deg(s)}{n^{\frac{1}{2}}}\right\rceil \leq 2$; the percentage of vertices verifying $\timp s{n^{\frac{1}{2}}}-\left\lceil\Td  {\deg(s)}{n^{\frac{1}{2}}}\right\rceil = 2$ is always below $0.5\%$, and the percentage of vertices verifying $\timp s{n^{\frac{1}{2}}}-\left\lceil\Td {\deg(s)}{n^{\frac{1}{2}}}\right\rceil = 1$ is always below $5\%$. 

For the other two points of \Cref{ax:dev}, for each vertex $s$, we have computed $\timp s{n^{\frac{1}{2}}}-\Td {\deg(s)}{n^{\frac{1}{2}}}$. We want to prove that the number of vertices that verify $\timp s{n^{\frac{1}{2}}}-\Td {\deg(s)}{n^{\frac{1}{2}}} \geq k$ is close to $nc^k$, for some constant $c$ smaller than $1$. For this reason, we have plotted the fraction of vertices verifying this inequality in logarithmic scale, in \Cref{fig:devsm}.

This plot confirms the last two points of \Cref{ax:dev}: indeed, in logarithmic scale, the number of vertices satisfying $\timp s{n^{\frac{1}{2}}}-\Td {\deg(s)}{n^{\frac{1}{2}}} \geq k$ decreases almost linearly with $k$, when $k>0$.

Then, let us validate \Cref{ax:touch}, which says that, whenever $x+y>1+\epsilon$, for each pair of vertices $s,t$, $\dist(s,t) < \timp s{n^x}+\timp t{n^y}$: we have tested this condition with $(x,y)=(0.3,0.9), (0.4,0.8), (0.5, 0.7), (0.6,0.6)$. For each graph $G=(V,E)$ in the dataset, and for each of the aforementioned pairs $(x,y)$, we have chosen a set $T \subseteq V$ made by $10\,000$ random vertices (or the whole $V$ if $|V|<10\,000$), and for each $i$ we have plotted the percentages of pairs $(s,t) \in T^2$ such that $\timp s{n^x}+\timp t{n^y}-\dist(s,t)=i$. The plots are shown in \Cref{fig:plottouch}.

From the figure, it is clear that $\timp s{n^x}+\timp t{n^y}$ is almost always at least $\dist(s,t)$, as predicted by \Cref{ax:touch}. However, in some cases, $\dist(s,t)=\timp s{n^x}+\timp t{n^y}$: we think that this is due to the fact that, in our random graph models, the guarantee is $\O\left(e^{-n^{\epsilon}}\right)$, and for $\epsilon=0.2$, this value is not very small (for instance, if $n=10\,000$, $e^{-n^{\epsilon}}=0.012$). However, this value tends to $0$ when $n$ tends to infinity, and this is reflected in practice: indeed, the fit is better when the number of nodes is larger. Overall, we conclude that \Cref{ax:touch} is valid with good approximation on the networks in the dataset, and we conjecture that the correspondance is even stronger for bigger values of $n$.

\begin{figure*}[t!]
\input{Untouch}
\caption{the values of $1+\frac{\log \frac{N_z}{|T|}}{\log n}$, as a function of $z$.} \label{fig:plotuntouch}
\end{figure*}
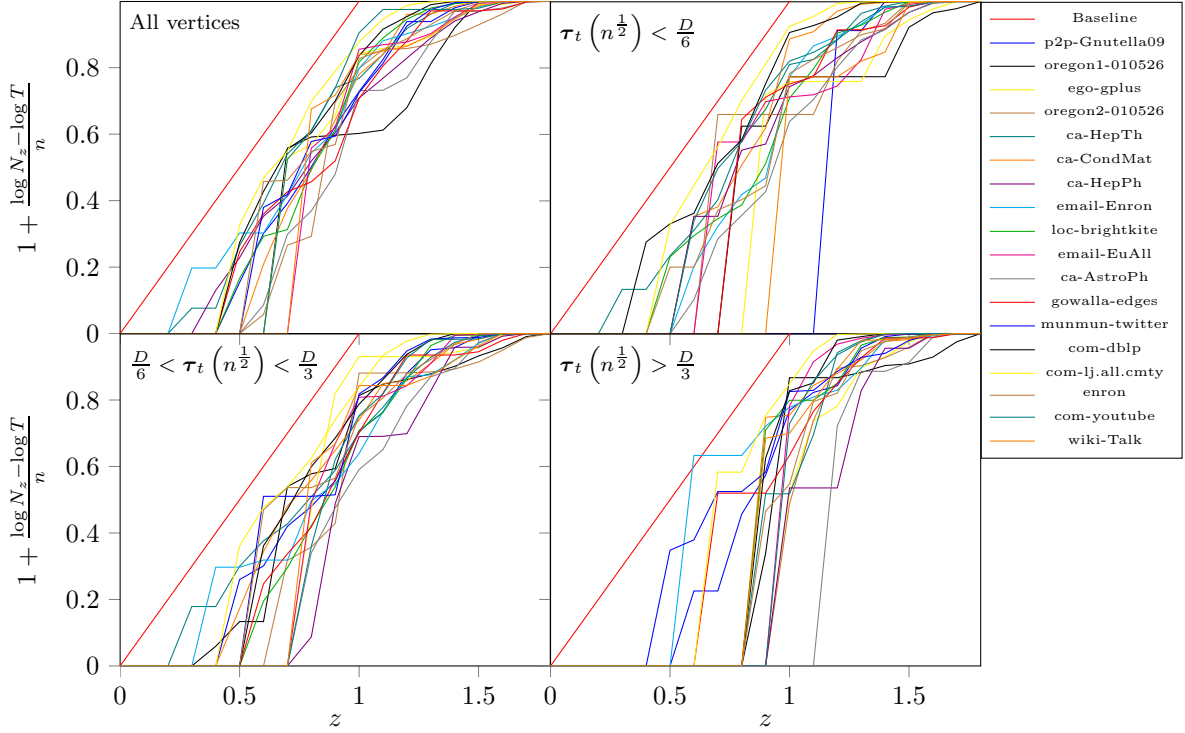
Finally, we need to validate \Cref{ax:untouch}, which says that, given a vertex $s$, for ``many'' sets of vertices $T$, $|\{t \in T: \timp s{n^x}+\timp t{n^y}<\dist(s,t)+2\}| \leq |T|n^{1-x-y+\epsilon}$. Hence, we have chosen a random vertex $s$ and a random set $T$ made by $10\,000$ vertices, and for each $t \in T$, we have computed $z_t=\min\{x+y:x>y,\timp s{n^x}+\timp t{n^y}<\dist(s,t)+2\}$. If the number $N_z$ of vertices $t$ such that $z_t<z$ is at most $|W|n^{-1+z+\epsilon}$, then we can guarantee that the theorem holds for each $x$ and $y$. Solving with respect to $z$, we want that $N_z\leq |T|n^{-1+z+\epsilon}$, that is, $\log \frac{N_z}{|T|} \leq (-1+z+\epsilon)\log n$, that is, $z \geq 1-\epsilon+\frac{\log \frac{N_z}{|T|}}{\log n}$. Hence, \Cref{fig:plotuntouch} shows the values of the function $1+\frac{\log \frac{N_z}{|T|}}{\log n}$, for each graph in our dataset. Furthermore, since \Cref{ax:untouch} also deals with sets $T$ defined depending on $\timp t{n^x}$, we have also repeated the experiment on sets $T$ containing only vertices $t$ verifying $0 \leq \timp t{n^{\frac{1}{2}}} < \frac{D}{6}$, $\frac{D}{6} \leq \timp t{n^{\frac{1}{2}}} < \frac{D}{3}$, $\timp t{n^{\frac{1}{2}}} > \frac{D}{3}$, where $D$ is the diameter of the graph.

From the plot, it is clear the claim is verified even with $\epsilon=0$, by all but one case. Also the latter case is verified with a very small value of $\epsilon$.

For the validation of \Cref{ax:deg}, we rely on extensive studies that show that the degree distribution of many real-world graphs is power law (for more information, we refer to \cite{Newman2003} and the references therein).

\section{Technical Preliminaries on Diameter, Eccentricity, Closeness Centrality, and Average Distance.} \label{sec:diamecc}

In this section, we prove some basic metric properties in the graphs satisfying our four properties. By specializing these results to random graphs, we obtain a new proof of known asymptotics, and we prove new asymptotics in the case $1<\beta<2$. In all the following lemmas, with abuse of notation, we write $\O(\epsilon)$ even if $\epsilon$ is a constant, in order to indicate a function bounded by $c\epsilon$ for some constant $c$.

\begin{lemma} \label{lem:tails}
All vertices $s$ with degree $d$ verify $\timp s{n^x} \leq \left\lfloor(1+\O(\epsilon))\left(\Td d{n^x}+\frac{\log n}{-\log \C}+x\right)\right\rfloor$. Moreover, for each $\delta>0$, there are $\Omega\left(n^{\delta}\right)$ vertices $s$ with degree $1$ verifying $\timp s{n^x} \geq \left\lceil(1-\epsilon-\delta)\left(\Td 1{n^x}+\frac{\log n}{-\log \C}-1+x\right)\right\rceil$.
\end{lemma}
\begin{proof}
By \Cref{ax:dev} applied with $\alpha=\left(1+\epsilon\right)\frac{\log n}{-\log \C}+x$, there are $\O\left(n\C^{\alpha-x}\right)=\O\left(n\C^{\left(1+\epsilon\right)\frac{\log n}{-\log \C}}\right)\leq \O\left(n^{-\epsilon}\right)<1$ vertices $s$ such that $\timp s{n^x} \geq (1+\epsilon)\left(\Td d{n^x}+\alpha\right)=(1+\epsilon)\left(\Td d{n^x}+(1+\epsilon)\frac{\log n}{-\log \C}+x\right)$. By observing that $\timp s{n^x}$ is an integer, we obtain the first claim.

For the other inequality, let us apply \Cref{ax:dev} with $\alpha=(1-\delta)\frac{\log n}{-\log \C}-1+x$: there are $\Omega\left(n\C^{\alpha+1-x}\right)=\Omega\left(n\C^{\left(1-\delta\right)\frac{\log n}{-\log \C}}\right) = \Omega\left(n^{\delta}\right)$ vertices $s$ such that $\timp s{n^x} \geq (1-\epsilon)\left(\Td 1{n^x}+\alpha\right)=(1-\epsilon)\left(\Td 1{n^x}+(1-\delta)\frac{\log n}{-\log \C}-1+x\right)\geq (1-\epsilon-\delta)\left(\Td 1{n^x}+\frac{\log n}{-\log \C}-1+x\right)$. By observing that $\timp s{n^x}$ is an integer, the second claim is proved.

\end{proof}

By combining the previous lemma with \Cref{ax:untouch,ax:touch}, we can estimate the eccentricity of each vertex.

\begin{theorem} \label{thm:ecc}

For each vertex $s$ and for each $x$ between $0$ and $1$,

\begin{multline*}
\ecc(s) \leq \timp s{n^x} +  \\ \left\lfloor(1+\O(\epsilon))\left(\Td 1{n^{1-x}}+\frac{\log n}{-\log \C}-x\right)\right\rfloor.
\end{multline*}
Furthermore, for each $s$ and for each $x \geq \frac{1}{2}$:
\begin{multline*}
\ecc(s) \geq \timp s{n^{x}}+ \\ \left\lceil(1+\O(\epsilon))\left(\Td 1{n^{1-x}}+\frac{\log n}{-\log \C}-x\right)\right\rceil-2.
\end{multline*}
\end{theorem}
\begin{proof}
By \Cref{ax:touch}, for each vertex $t$, $\dist(s,t)\leq \timp s{n^x}+\timp t{n^{1-x+\epsilon}}-1$. By \Cref{lem:tails}, for each $t$, \begin{multline*}
\timp t{n^{1-x+\epsilon}} \leq \left\lfloor(1+\O(\epsilon)) \phantom{\left(\frac{\log n}{-\log \C}\right)}\right. \\
\left.\left(\Td {\deg(t)}{n^{1-x+\epsilon}}+\frac{\log n}{-\log \C}+1-x+\epsilon\right)\right\rfloor,
\end{multline*}
and consequently $\ecc(s)=\max_{t \in V}\dist(s,t)\leq \timp s{n^x}+\left\lfloor(1+\O(\epsilon))\left(\Td {1}{n^{1-x+\epsilon}}+\frac{\log n}{-\log \C}+1-x+\epsilon\right)\right\rfloor-1=\left\lfloor(1+\O(\epsilon))\left(\Td {1}{n^{1-x}}+\frac{\log n}{-\log \C}-x\right)\right\rfloor$.

For the other inequality, if $x \geq \frac{1}{2}$, let $y=1-x-\epsilon<x$, and let $T$ be the set of vertices $t$ such that $\timp t{n^y} \geq \left\lceil(1-3\epsilon)\left(\Td 1{n^y}+\frac{\log n}{-\log \C}-1+y\right)\right\rceil$ (by \Cref{lem:tails}, $|T| \geq n^{2\epsilon}$). By \Cref{ax:untouch}, there is at least a vertex $t \in T$ verifying $\dist(s,t) \geq \timp s{n^{x}}+\timp t{n^y}-2 \geq \timp s{n^x}+\left\lceil(1-3\epsilon)\left(\Td 1{n^y}+\frac{\log n}{-\log \C}-1+y\right)-2\right\rceil$. The second claim follows.
\end{proof}

Thanks to this lemma, we can compute the diameter of a graph as the maximum eccentricity.

\begin{theorem} \label{thm:diameter}
For each $x$, the diameter of our graph is 
$D = \left\lfloor(1+\O(\epsilon))\left(\avedist n11+\frac{2\log n}{-\log \C}\right)\right\rfloor$.
\end{theorem}
\begin{proof}
By combining the upper bounds in \Cref{thm:ecc} and \Cref{lem:tails}, we can prove that $D \leq \left\lfloor(1+\O(\epsilon))\left(\Td 1{n^x}+\frac{\log n}{-\log \C}+x\right)\right\rfloor+\left\lfloor(1+\O(\epsilon))\left(\Td 1{n^{1-x}}+\frac{\log n}{-\log \C}-x\right)\right\rfloor$. If we choose $x$ such that $(1+\O(\epsilon))\left(\Td 1{n^x}+\frac{\log n}{-\log \C}+x\right)=i-\epsilon$, we obtain that $D \leq i-1+\left\lfloor(1+\O(\epsilon))\left(\Td 1{n^{1-x}}+\frac{\log n}{-\log \C}-x\right)\right\rfloor\leq \left\lfloor i+(1+\O(\epsilon))\left(\Td 1{n^{1-x}}+\frac{\log n}{-\log \C}-x-1\right)\right\rfloor \leq \left\lfloor (1+\O(\epsilon))\left(\avedist n11+\frac{2\log n}{-\log \C}\right)\right\rfloor$.



Let us combine the lower bounds in \Cref{thm:ecc} and \Cref{lem:tails}: we obtain that $D \geq \left\lceil(1-\O(\epsilon))\left(\Td 1{n^x}+\frac{\log n}{-\log \C}+x-1\right)\right\rceil + \left\lceil(1-\O(\epsilon))\left(\Td 1{n^{1-x}}+\frac{\log n}{-\log \C}-x-1\right)\right\rceil-1$. For all but a constant number of values of $x$, this value is equal to $\left\lfloor(1-\O(\epsilon))\left(\Td 1{n^x}+\frac{\log n}{-\log \C}+x\right)\right\rfloor + \left\lfloor(1-\O(\epsilon))\left(\Td 1{n^{1-x}}+\frac{\log n}{-\log \C}-x\right)\right\rfloor-1$. Furthermore, if $(1-\O(\epsilon))\left(\Td 1{n^x}+\frac{\log n}{-\log \C}+x\right)=i+\epsilon$ for some integer $i$, this value is $\left\lfloor i+(1-\O(\epsilon))\left(\Td 1{n^{1-x}}+\frac{\log n}{-\log \C}-x\right)\right\rfloor-1 \geq \left\lfloor(1-\O(\epsilon))\left(\avedist n11+\frac{2\log n}{-\log \C}\right)\right\rfloor$. A similar argument can be applied if the second term is $i+\epsilon$: hence, it only remains to prove that we can find a value of $x$ between $\frac{1}{2}$ and $1$ such that one of the two parts is close to an integer. This is true because $\Td 1{n^x}+x$ is continuous and increasing with respect to $x$, and $\Td 1{n^{1-x}}-x$ is continuous and decreasing. Since the incease and the decrease are at least $\frac{1}{2}$, the sum of the two is at least $1$.
\end{proof}

Given \Cref{thm:main}, and given the values in \Cref{tab:tc}, \Cref{thm:diameter} gives diameter bounds for power law graphs generated through the models considered (since the four properties hold for each $\epsilon$, we can safely let $\epsilon$ tend to $0$, and transform $\O(\epsilon)$ into $\o(1)$). As far as we know, the bound for $1<\beta<2$ is new, while the other bounds are already known \cite{Fernholz2007,Bollobas2007}.

\begin{corollary}
If $\lambda$ is a power law degree distribution with exponent $\beta$, the diameter of a random graph with degree distribution $\lambda$ is:
\begin{itemize}
\item if $1<\beta<2$, $D=\left\lfloor3+\frac{\beta-2}{\beta-1}\right\rfloor$;
\item if $2<\beta<3$, $D=(1+\o(1))\left(\frac{2}{-\log \eta(1)}\right)\log n$;
\item if $\beta>3$, \[D=(1+\o(1))\left(\frac{2}{-\log \eta(1)}+\frac{1}{\log M_1(\mu)}\right)\log n.\]
\end{itemize}
\end{corollary}

All the previous results deal with ``extremal'' properties of the distance distribution. Instead, the next results deal with properties that hold on average. Let us start by estimating the farness of a node $s$, that is, $\sum_{t \in V} \dist(s,t)$.

\begin{theorem}\label{thm:farnessupper}
For each vertex $s$ and for each $0<x<1$, the \emph{farness} $\farn s$ of $s$ verifies
\begin{multline*}
\farn s \leq n(1+\O(\epsilon))(\timp s{n^x}-\Td 1{n^x} \\ 
+\avedist n11)-\deg(s).
\end{multline*}
\end{theorem} 
\begin{proof}
By \Cref{ax:touch}, for each vertex $t$, $\dist(s,t)\leq \timp s{n^{x}} + \timp t{n^{1-x+\epsilon}}-1$, and hence 
\begin{align*}
&\sum_{t \in V} \dist(s,t) \\
&\leq n\left(\timp s{n^{x}}-1\right) +\sum_{d=1}^{+\infty}\sum_{\deg(t)=d} \timp t{n^{1-x+\epsilon}} \\
&=n\left(\timp s{n^{x}}-1\right) +\sum_{d=1}^{+\infty}|\{t \in V:\deg(t)=d\}| \cdot \\
&\quad \Td d{n^{1-x+\epsilon}} -n\\
&\leq n\left(\timp s{n^{x}}-1\right) + n(1+\o(1))\Td 1{n^{1-x+\epsilon}} \\
&\quad+ n\Td 1{n^{x-\epsilon}}-n\Td 1{n^{x-\epsilon}}\\
&\leq n(1+\O(\epsilon))\left(\timp s{n^{x}}-\Td 1{n^x}+\avedist nd1\right).
\end{align*}
We need to subtract $\deg(s)$ from this result. To this purpose, we observe that in the first estimate, all neighbors of $s$ with degree at most $n^{1-x}$ were given a distance $\timp s{n^{x}} + \timp t{n^{1-x+\epsilon}}-1=\timp s{n^{x}}+2-1 \geq 2$, and consequently the other estimates remain correct if we subtract the number of neighbors of $s$ with degree at most $n^{1-x}$, or equivalently if we subtract $\deg(s)$ and we sum the number of neighbors of $s$ with degree at least $n^{1-x}$. Since $|E| \leq n^{1+\epsilon}$ by \Cref{ax:deg}, the number of vertices with degree at least $n^{1-x}$ is at most $n^{x+\epsilon}$, and the latter contribution is negligible.
\end{proof}

\begin{theorem} \label{thm:farnesslower}
For each vertex $s$ and for each $\frac{1}{2}\leq x<1$,
\begin{multline*}
\farn s \geq n(1-\O(\epsilon)) \cdot\\
\left(\timp s{n^x}-\Td 1{n^x}+\avedist n11-1\right).
\end{multline*}\end{theorem}
\begin{proof}
Let $s$ be any vertex, and let us apply \Cref{ax:untouch} with $T=V$: there are at most $\O\left(n^{1-\epsilon}\right)$ vertices $t \in V$ such that $\dist(s,t) < \timp s{n^x} + \timp t{n^{1-x-2\epsilon}}-2$. Let $T':=\{t \in V:\dist(s,t) \geq \timp s{n^x} + \timp t{n^{1-x-2\epsilon}}-2\}$. 
\begin{align*}
\farn s &= \sum_{t \in V} \dist(s,t) \\
&\geq \sum_{t \in V'} \timp s{n^x} + \timp t{n^{1-x-2\epsilon}}-2 \\
&= n(1-\o(1))\left(\timp s{n^x}-2\right)\\
&\quad +\sum_{t \in V} \timp t{n^{1-x-2\epsilon}}- \timp t{n^{1-x-2\epsilon}} \\
&=n(1-\O(\epsilon))(\timp s{n^x}-\Td 1{n^x}+\avedist n11 \\
&\quad -1)-\sum_{i=1}^{+\infty} |V-V'|\O(\log n) \\
&=n(1-\O(\epsilon))(\timp s{n^x}-\Td 1{n^x} \\
&\quad +\avedist n11-1).
\end{align*}
\end{proof}

By computing the inverse of the farness, we can compute the closeness centrality of a vertex.

\begin{corollary}
For each $x$ such that $\frac{1}{2}\leq x<1$, the closeness centrality of a vertex $s$ verifies $\frac{1-\O(\epsilon)}{n\left(\timp s{n^x}-\Td 1{n^x}+\avedist n11\right)-\deg(s)} \leq \clos s \leq \frac{1+\O(\epsilon)}{n\left(\timp s{n^x}-\Td 1{n^x}+\avedist n11-1\right)}.$
\end{corollary}

\begin{corollary} \label{cor:avedist}
The average distance between two vertices is between $(1-\O(\epsilon))\avedist n11-1$ and $(1+\O(\epsilon))\avedist n11$. 
\end{corollary}
\begin{proof}
The average distance is the sum of the farness of all vertices, divided by $n(n-1)$. By the two previous theorems, for each $x\geq\frac{1}{2}$, $n(1+\O(\epsilon))\left(\timp s{n^x}-\Td 1{n^x}+\avedist n11-1\right) \leq \farn s \leq n(1+\O(\epsilon))\left(\timp s{n^x}-\Td 1{n^x}+\avedist n11\right)$.

Let us compute $\sum_{s \in V} \timp s{n^x}=\sum_{d = 1}^{+\infty} \sum_{\deg(s)=d} \timp s{n^x}=\sum_{d=1}^\infty |\{s:\deg(s)=d\}|\Td d{n^x}=(1+\o(1))\Td 1{n^x}$. Combining this estimate with the previous equation, we obtain:

\begin{multline*}
(1-\O(\epsilon))\avedist n11-1 \leq \frac{1}{n(n-1)}\sum_{s \in V}\farn s \\ \leq (1+\O(\epsilon))\avedist n11.
\end{multline*}
\end{proof}

Again, assuming \Cref{thm:main}, and given the values in \Cref{tab:tc}, we have proved the following asymptotics for the average distance in random graphs.

\begin{corollary}
If $\lambda$ is a power law degree distribution with exponent $\beta$, the average distance in a random graph with degree distribution $\lambda$ is:
\begin{itemize}
\item if $1<\beta<2$, $2-\o(1) \leq \dist_{\text{\upshape avg}} \leq 3+\o(1)$;
\item if $2<\beta<3$, $\dist_{\text{\upshape avg}}=(2+\o(1))\left(\log_{\frac{1}{\beta-1}} \log n \right)$;
\item if $\beta>3$, $\dist_{\text{\upshape avg}}=(1+\o(1))\frac{\log n}{\log M_1(\mu)}$.
\end{itemize}
\end{corollary}

\section{Bounding the Diameter Through Sampling.} \label{sec:sampl}

The first algorithm we analyze is very simple: it lower bounds the diameter of a graph by performing $k$ BFSes from random nodes $s_1,\dots,s_k$, and returning $\max_{i = 1,\dots,k} \ecc(s_i)$. Clearly, the running time is $\O(mk)$: we want to analyze the error of this method on graphs that satisfy our assumptions. The main idea behind this analysis is that $\ecc(s)$ is strongly correlated with $\timp s{n^x}$, and the number of vertices verifying $\timp s{n^x}>\alpha$ decreases exponentially with respect to $\alpha$. This means that the number of vertices with high eccentricity is very small, and it is difficult to find them by sampling: this means that the error should be quite big.

More formally, by \Cref{thm:ecc}, the eccentricity of a vertex $s$ verifies $\ecc(s) \leq \timp s{n^x} + \left\lfloor(1+\O(\epsilon))\left(\Td 1{n^{1-x}}+\frac{\log n}{-\log \C}-x\right)\right\rfloor$, and consequently the output is at most $\max_{i = 1,\dots,k} \timp {s_i}{n^x} + \left\lfloor(1+\O(\epsilon))\left(\Td 1{n^{1-x}}+\frac{\log n}{-\log \C}-x\right)\right\rfloor$. We want to estimate $\max_{i = 1,\dots,k} \timp {s_i}{n^x}$ through \Cref{ax:dev}: the number of vertices $s$ verifying $\timp s{n^x} \geq (1+\O(\epsilon))\left(\Td 1{n^x}+\alpha\right)$ is at most $n\C^{\alpha-x}$, and consequently a random set of $k$ vertices does not contain any such vertex, \aas, if $k \leq \frac{n^{1-\epsilon}}{n\C^{\alpha-x}} \ll\frac{n}{n\C^{\alpha-x}}$. Solving the first inequality with respect to $\alpha$, we obtain $\alpha \geq x+\frac{\epsilon\log n+\log k}{-\log \C}$. 

We conclude that, \aas, if $\alpha=x+\frac{\epsilon\log n+\log k}{-\log \C}$, we do not perform any BFS from a vertex $s$ such that $\timp s{n^x} \geq (1+\epsilon)\left(\Td 1{n^{x}}+\alpha\right)=(1+\O(\epsilon))\left(\Td 1{n^{x}}+\frac{\gamma\log n}{-\log \C}+x\right)$. This means that, for a suitable choice of $x$, the output is smaller than:
\begin{align*}
&\max_{s \in X} \timp s{n^x} \\
&+ \left\lfloor(1+\O(\epsilon))\left(\Td 1{n^{1-x}}+ \frac{\log n}{-\log \C}-x\right)\right\rfloor \\
&\leq \left\lfloor(1+\O(\epsilon))\left(\Td 1{n^{x}}+\frac{\gamma\log n}{-\log \C}+x\right)\right\rfloor \\
&\quad+ \left\lfloor(1+\O(\epsilon))\left(\Td 1{n^{1-x}}+\frac{\log n}{-\log \C}-x\right)\right\rfloor\\
&\leq \left\lfloor(1+\O(\epsilon))\left(\avedist n11+\frac{(1+\gamma)\log n}{-\log \C}-1\right)\right\rfloor.
\end{align*}

By replacing the values in \Cref{tab:tc}, we obtain the desired results. In order to obtain a lower bound on the error, it is enough to perform similar computations after replacing $\epsilon$ with $-\epsilon$.

\section{The 2-Sweep Heuristic.} \label{sec:tsweep}

The \tsweep\ heuristic \cite{Magnien2009} finds a lower bound on the diameter, by performing a BFS from a vertex $s$, finding a vertex $t$ that maximizes the distance from $s$, and returning the eccentricity of $t$ (since only 2 BFSes are performed, the running time is linear in the input size). Following the intuitive proof, let us show that $\timp t{n^{x}}$ is high, and consequently the eccentricity of $t$, which depends on $\timp t{n^x}$, is high as well.

\begin{lemma}\label{lem:eccv}
For each vertex $s$, let $t$ be a vertex maximizing the distance from $s$. Then, for each $x \geq \frac{1}{2}$, $\timp t{n^{1-x}} \geq \left\lceil(1-\O(\epsilon))\left(\Td 1{n^{1-x}}+\frac{\log n}{-\log \C}-1-x\right)\right\rceil$.
\end{lemma}
\begin{proof}
By \Cref{ax:touch,thm:ecc}, $\timp s{n^{x+\epsilon}} + \left\lceil(1-\O(\epsilon))\left(\Td 1{n^{1-x-2\epsilon}}+\frac{\log n}{-\log \C}-2-x\right)\right\rceil \leq \ecc(s)=\dist(s,t) \leq \timp s{n^{x+\epsilon}}+\timp t{n^{1-x}}-1$. From this inequality, we obtain that $\timp t{n^{1-x}} \geq \left\lceil(1-\O(\epsilon))\left(\Td 1{n^{1-x}}+\frac{\log n}{-\log \C}-1-x\right)\right\rceil$. 
\end{proof}

By \Cref{thm:ecc,thm:diameter}, if $t$ is the vertex maximizing the distance from $s$:

\begin{align*}
&\ecc(t) \geq \timp t{n^{\frac{1}{2}}} \\
&\quad+ \left\lceil (1-\O(\epsilon))\left(\Td 1{n^{\frac{1}{2}}}+\frac{\log n}{-\log \C}-{\frac{1}{2}}\right)\right\rceil-2 \\
&\geq \left\lceil(1-\O(\epsilon))\left(\Td 1{n^{{\frac{1}{2}}}}+\frac{\log n}{-\log \C}-\frac{3}{2}\right)\right\rceil \\ 
&\quad + \left\lceil (1-\O(\epsilon))\left(\Td 1{n^{\frac{1}{2}}}+\frac{\log n}{-\log \C}-\frac{5}{2}\right)\right\rceil \\
&\geq 2\left\lfloor(1-\O(\epsilon))\left(\Td 1{n^{\frac{1}{2}}}+\frac{\log n}{-\log \C}-\frac{1}{2}\right)\right\rfloor -1 \\
&\geq 2\left\lfloor(1-\O(\epsilon))\frac{D}{2}\right\rfloor-1
\end{align*}
(in this analysis, we used that $\Td{1}{n^{\frac{1}{2}}}+\frac{\log n}{-\log \C}$ is not an integer, and that $\epsilon$ is small enough).

We conclude that the output of the \tsweep\ heuristic is $2\left\lfloor(1-\O(\epsilon))\frac{D}{2}\right\rfloor-1$, proving the results in \Cref{tab:summary}.

\section{The RW Algorithm.} \label{sec:rw}

The \rw\ algorithm \cite{Roditty2013} is a randomized algorithm that computes a $\frac{3}{2}$-approximation of the diameter of a graph, in time $\Theta(m\sqrt{n}\log n)$. The algorithm works as follows: we choose $k=\Theta(\sqrt{n}\log n)$ vertices $s_1,\dots,s_k$, and we perform a BFS from each of these vertices. Then, we compute the vertex $t$ maximizing $\min_{i=1,\dots,k} \dist(s_i,t)$, and we let $t_1,\dots,t_k$ be the $k$ vertices closest to $t$. Then, if there exist $i,j$ such that $s_i=t_j$, we return the maximum eccentricity among the $s_i$s and the $t_j$s, otherwise the algorithm fails, and we can decide to run it again (anyway, the probability that it fails are small). 

For the worst-case analysis of this algorithm, we refer to \cite{Roditty2013}: in this work, we analyze its performances in our framework. The running time is still $\Theta(m\sqrt{n}\log n)$, since the algorithm requires $\Theta(\sqrt{n}\log n)$ BFSes from the vertices $s_i$, $t_j$, but the approximation factor can be better than the worst-case. Intuitively, this algorithm is quite similar to the \tsweep\ heuristic (if $k$ was $1$, the algorithm would be the \tsweep\ heuristic), and we conjecture that also its behavior should be similar.

For a formal proof, let $v$ be any vertex: since the vertices $s_i$ are random, $\P\left(\forall i, s_i \in \N{\timp v{n^x}}v\right) \geq \left(1-\frac{n^x}{n}\right)^{\sqrt{n}\log n}=e^{n^{x-\frac{1}{2}+\o(1)}}$, and similarly $\P\left(\forall i, s_i \notin \N{\timp v{n^x}-1}v\right) \leq \left(1-\frac{\O\left(n^x\log n\right)}{n}\right)^{\sqrt{n}\log n}=e^{n^{x-\frac{1}{2}+\o(1)}}$ (because there are at most $n^x \timp v{n^x} = \O(n^x \log n)$ vertices in $\N{\timp v{n^x}-1}v$). This means that $\min_{i=1,\dots,k} \dist(s_i,v) \leq \timp v{n^{\frac{1+\epsilon}{2}}}$ \whp, and $\min_{i=1,\dots,k} \dist(s_i,v) \geq \timp v{n^{\frac{1-\epsilon}{2}}}-1$ \aas. Hence, if $v$ is one of the vertices maximizing $\timp v{n^{\frac{1}{2}}}$, $\min_{i=1,\dots,k} \dist(s_i,v) \geq \timp v{n^{\frac{1-\epsilon}{2}}}-1 \geq \left\lceil(1-\O(\epsilon))\left(\Td 1{n^{\frac 12}}+\frac{\log n}{-\log \C}-\frac 12\right)\right\rceil-1$ \aas, by \Cref{lem:tails}.

This means that the vertex $t$ maximizing $\min_{i=1,\dots,k} \dist(s_i,t)$ verifies $\timp t{n^{\frac{1+\epsilon}{2}}} \geq \min_{i=1,\dots,k} \dist(s_i,t) \geq  \min_{i=1,\dots,k} \dist(s_i,v) \geq \left\lceil(1-\O(\epsilon))\left(\Td 1{n^{\frac 12}}+\frac{\log n}{-\log \C}-\frac 12\right)\right\rceil$. This means that $\ecc(t) \geq \timp t{n^{\frac{1}{2}}}+\left\lceil(1-\O(\epsilon))\left(\Td 1{n^{\frac{1}{2}}}+\frac{\log n}{-\log \C}-\frac{1}{2}\right)\right\rceil-2\geq \left\lceil(1-\O(\epsilon))\left(\Td 1{n^{\frac 12}}+\frac{\log n}{-\log \C}-\frac 12\right)\right\rceil-1+\left\lceil(1-\O(\epsilon))\left(\Td 1{n^{\frac{1}{2}-\epsilon}}+\frac{\log n}{-\log \C}-\frac{1}{2}\right)\right\rceil-2=2\left\lceil(1-\O(\epsilon))\left(\Td 1{n^{\frac 12}}+\frac{\log n}{-\log \C}-\frac 12\right)\right\rceil-3$, and this value is exactly the same value we obtained for the \tsweep\ heuristic.

Note that this analysis only uses the vertex $t$, and a more refined analysis could (in principle) obtain better bounds.

\section{The SumSweep Heuristic.} \label{sec:ssh}

The \ss\ heuristic \cite{Borassi2014,Borassi2014a} provides a lower bound on the eccentricity of all vertices, by performing some BFSes from vertices $t_1,\dots,t_k$, and defining $L(v)= \max_{i=1,\dots,k} \dist(v,t_i) \leq \ecc(v)$ for each vertex $v$. The vertices $t_i$ are chosen as follows: we start from a random vertex $s_1$, then we choose $t_1$ as the vertex maximizing $\dist(s_1,t_1)$. Then, we choose again $s_2$ as a random vertex, and we choose $t_2$ as the vertex in $V-\{t_1\}$ maximizing $\dist(s_1,t_2)+\dist(s_2,t_2)$. In general, after $2i$ BFSes are performed, we choose a random vertex $s_{i+1}$, we perform a BFS from $s_{i+1}$, and we choose $t_{i+1}$ as the vertex in $V-\{t_1,\dots,t_i\}$ maximizing $\sum_{j=1}^{i+1}\dist(s_j,t_{i+1})$.\footnote{Actually, in the original \ss\ heuristic, there is no distinction between the vertices $s_i$ and $t_i$: we simply choose $t_{i+1}$ as the vertex maximizing $\sum_{j=1}^i \dist(t_j,t_{i+1})$. However, for our analyses, we need to use this variation.}

The idea behind the analysis of the \tsweep\ heuristic and the \rw\ algorithm is to exploit the existence of few vertices $t$ with big values of $\timp t{n^x}$: both algorithms find a single vertex $t$ such that $\timp t{n^x}$ is high, and they lower bound the diameter with the eccentricity of this vertex, which is peripheral. Instead, the \ss\ heuristic highlights all vertices $t_i$ having big values of $\timp t{n^x}$, and it performs a BFS from each of these vertices. Then, for each vertex $v$, if $t$ is the vertex farthest from $v$, $\timp t{n^x}$ is big, and this means that $t=t_i$ for some small $i$. Consequently, if we lower bound $\ecc(s) \geq \max_{i=1,\dots,k} \ecc(t_i)$, the lower bounds obtained are tight after few steps. Let us formalize this intuition: first, we need to prove that the vertices $t$ with high value of $\timp t{n^{1-x}}$ are chosen soon by this procedure.

\begin{lemma}
Let $S$ be a random set of vertices, let $\timp S{n^y}=\frac{1}{|S|}\sum_{s \in S} \timp s{n^y}$, and let $t$ be any vertex in the graph. Then, $\frac{\sum_{s \in S} \dist(s,t)}{|S|} \leq \timp t{n^x}+\timp S{n^{1-x+\epsilon}}-1$. Furthermore, if $|S|>n^{3\epsilon}$, $x \geq \frac{1}{2}$, $\frac{\sum_{s \in S} \dist(s,t)}{|S|} \geq (1-\o(1))\left(\timp t{n^x}+\timp S{n^{1-x-2\epsilon}}-1\right)$, \whp.
\end{lemma}
\begin{proof}
For the upper bound, by \Cref{ax:touch}, $\sum_{s \in S} \dist(s,t) \leq \sum_{s \in S} \timp t{n^x}+\timp s{n^{1-x+\epsilon}}-1=|S|\left(\timp t{n^x}+\timp S{n^{1-x+\epsilon}}-1\right)$.

For the lower bound, by \Cref{ax:untouch}, for each vertex $t$, the number of vertices $s \in V$ verifying $\dist(s,t) < \timp t{n^x}+\timp s{n^{1-x-2\epsilon}}-2$ is at most $n^{1-\epsilon}$. Let $S' \subseteq S$ be the set of vertices verifying $\dist(s,t) \geq \timp t{n^x}+\timp s{n^{1-x-2\epsilon}}-2$: since $S$ is random, the probability that a vertex $s \in S$ does not belong to $S'$ is at least $n^{-\epsilon}$. From this bound, we want to prove that $|S'|\geq (1-\O(n^{-\epsilon}))|S|$, using Hoeffding's inequality. For each $s \in S$, let $\X_s=1$ if $\dist(s,t) \geq \timp t{n^x}+\timp s{n^{1-x-2\epsilon}}-2$, $0$ otherwise: clearly, $|S'|=\sum_{s \in S} \X_s$, the variables $\X_s$ are independent, and $\P(\X_s=1) \geq 1-n^{-\epsilon}$. By Hoeffding's inequality, $\P\left(\sum_{s \in S} \X_s<\E\left[\sum_{s \in S} \X_s\right]-\lambda \right) \leq e^{-\frac{\lambda^2}{|S|}}$. Since $\E\left[\sum_{s \in S} \X_s\right] \geq |S|(1-n^{-\epsilon})$, if we choose $\lambda=|S|n^{-\epsilon}$, we obtain that $\P\left(|S'|<(1-2n^{-\epsilon})|S| \right) \leq e^{-|S|n^{-2\epsilon}}$. We proved that, \whp, $|S'|\geq |S|(1-\O(n^{-\epsilon}))$. As a consequence:
\begin{align*}
&\sum_{s \in S} \dist(s,t) \\
&\geq \sum_{s \in S'} \dist(s,t) \\
& \geq \sum_{s \in S'} \timp t{n^x}+\timp s{n^{1-x-2\epsilon}}-2 \\
& \geq |S'|\left(\timp{t}{n^x}+\timp {S}{n^{1-x-2\epsilon}}-2\right)-\sum_{s \in S-S'} \O(\log n) \\
& \geq |S'|\left(\timp{t}{n^x}+\timp {S}{n^{1-x-2\epsilon}}-2\right)-\O\left(n^{-\epsilon}|S|\log n\right) \\
& \geq \left(1-\o\left(1\right)\right)|S|\left(\timp t{n^x}+\timp {S}{n^{1-x-2\epsilon}}-2\right).
\end{align*}
\end{proof}

By \Cref{lem:eccv}, if $t$ maximizes $\dist(u,t)$ for some $u \in V$, then for each $y\geq \frac{1}{2}$, $\timp t{n^{1-y-2\epsilon}} \geq \left\lceil(1-\O(\epsilon))\left(\Td 1{n^{1-y}}+\frac{\log n}{-\log \C}-1-y\right)\right\rceil$. If we choose $x=y=\frac{1}{2}$, the previous lemma proves that, if a vertex $v$ is chosen before $t$ in this procedure, then $\timp v{n^{\frac{1}{2}+3\epsilon}}+\timp S{n^{\frac{1}{2}-2\epsilon}}-1 \geq \frac{\sum_{s \in S} \dist(s,v)}{|S|} \geq \frac{\sum_{s \in S} \dist(s,t)}{|S|} \geq \left(1-\o\left(1\right)\right)\left(\timp t{n^{\frac{1}{2}}}+\timp S{n^{\frac{1}{2}-2\epsilon}}-2\right)$. Rearranging this inequality, we obtain $\timp v{n^{\frac{1}{2}+3\epsilon}} \geq (1-\o(1))\left(\timp t{n^{\frac{1}{2}}}-1\right)\geq \left\lceil(1-\O(\epsilon))\left(\Td 1{n^{\frac{1}{2}}}+\frac{\log n}{-\log \C}-\frac{5}{2}\right)\right\rceil=\left\lfloor(1-\O(\epsilon))\left(\Td 1{n^{\frac{1}{2}}}+\frac{\log n}{-\log \C}-\frac{3}{2}\right)\right\rfloor$.

If we apply \Cref{ax:dev} with the value of $\alpha$ verifying $(1+\epsilon)\left\lfloor\Td 1{n^{\frac{1}{2}+3\epsilon}}+\alpha\right\rfloor=\left\lfloor(1-\O(\epsilon))\left(\Td 1{n^{\frac{1}{2}}}+\frac{\log n}{-\log \C}-\frac{3}{2}\right)\right\rfloor$, we obtain that the number of vertices $v$ satisfying the latter equation is $\O(n\C^{\alpha-x}) = n\C^{\left\lfloor(1-\O(\epsilon))\left(\Td 1{n^{\frac{1}{2}}}+\frac{\log n}{-\log \C}-\frac{3}{2}\right)\right\rfloor-\Td 1{n^{\frac{1}{2}+3\epsilon}}-\frac{1}{2}}$. If $\beta>2$, this value is simply $\O(\epsilon)$, while if $1<\beta<2$, this value is $n^{1-\frac{2-\beta}{\beta-1}\left(\left\lfloor\frac{\beta-1}{2-\beta}-\frac{3}{2}\right\rfloor-\frac{1}{2}\right)}$, if $\epsilon$ is small enough.

We conclude that, after $n^{3\epsilon}+\O(n\C^{\alpha-x})$ BFSes, we have performed a BFS from all vertices $t$ that maximize $\dist(u,t)$ for some $u\in V$: this means that the lower bounds on all eccentricities are tight.

\section{The iFub Algorithm.} \label{sec:ifub}

The \ifub\ algorithm is an exact algorithm to compute the diameter of a graph. Its worst-case running time is $\O(mn)$, but it performs much better in practice \cite{Crescenzi2013}. It works as follows: it performs a BFS from a vertex $v$, and it uses the fact that, if $D=\dist(s,t)$, either $\dist(s,v) \geq \frac{D}{2}$, or $\dist(v,t) \geq \frac{D}{2}$. Hence, after the first BFS from $v$, the \ifub\ algorithm computes the eccentricity of all the other vertices, in decreasing order of distance from $v$. During this process, it keeps track of the maximum eccentricity found $D_L$, which is a lower bound on the diameter. As soon as we are processing a vertex $s$ such that $\dist(v,s) \leq \frac{D_L}{2}$, we know that, for each pair $(s,t)$ of unprocessed vertices, $\dist(s,t) \leq 2\frac{D_L}{2}=D_L$: this means that we have processed at least one of the vertices in a diametral pair, and $D_L=D$. The running time is $\O\left(mN_{\frac{D}{2}}(v)\right)$, where $N_{\frac{D}{2}}(v)$ is the number of vertices at distance at least $\frac{D}{2}$ from $v$ (indeed, the algorithm performs a BFS from each of these vertices).

For our analysis, we only need to estimate $N_{\frac{D}{2}}(v)$. Intuitively, the diameter is the sum of two contributions: one is $\avedist n11$, which is close to the average distance between two nodes, and the other is twice the maximum deviation from this value, that is, $2\frac{\log n}{-\log \C}$. Hence, $N_{\frac{D}{2}}(v)$ is the number of vertices at distance $\frac{\avedist n11}{2}+\frac{\log n}{-\log \C}$ from $v$: if the second term is dominant (for instance, if $2<\beta<3$), we are considering only vertices with very big deviations, and the time is much smaller than $n$. Conversely, if the deviation is smaller than $\frac{\avedist n11}{2}$, we expect this number to be $\O(n)$ (for instance, if $\beta>3$ and $\eta(1)$ is small).

Let us formalize this intuition. By \Cref{thm:diameter}, $\frac{D}{2} \geq \frac{1}{2}\left\lfloor\left(1+\O(\epsilon)\right)\left(\avedist n11+\frac{2\log n}{-\log \C}\right)\right\rfloor$. In order to estimate $N_{\frac{D}{2}}(v)$, we use the fact that if $x+y\geq1+\epsilon$, $\dist(v,w) \leq \timp v{n^x}+\timp w{n^{y}}-1$, and consequently, if $\dist(v,w) \geq \frac{D}{2}$, then $\timp w{n^{y}} \geq \left\lceil\frac{D}{2}+1-\timp v{n^x}\right\rceil \geq \left\lfloor\frac{1}{2}\left(1+\O(\epsilon)\right)\left(\avedist n11+\frac{2\log n}{-\log \C}\right)\right\rfloor+1-\timp v{n^x}$.

Let us apply \Cref{ax:dev} with $\alpha$ such that $(1+\epsilon)\left(\Td d{n^y}+\alpha\right)=\left\lfloor\frac{1}{2}\left(1+\O(\epsilon)\right)\left(\avedist n11+\frac{2\log n}{-\log \C}\right)\right\rfloor+1-\timp v{n^x}$: we obtain that $N_{\frac{D}{2}}(v) \leq nc^{\alpha-y} = n^{1-\frac{\log \C}{\log n} (\alpha-y)}$. Let us estimate:
\begin{align*}
&\alpha-y \\
& \leq (1+\O(\epsilon))\left(\left\lfloor\frac{1}{2}\left(1+\O(\epsilon)\right)\left(\avedist n11+\frac{2\log n}{-\log \C}\right)\right\rfloor \right.\\
&\quad\left.\phantom{\left(\frac{2\log n}{-\log \C}\right)}+1-\timp v{n^x}-\Td d{n^y}\right) \\
&\leq (1+\O(\epsilon))\left(\left\lfloor\frac{1}{2}\left(\avedist n11+\frac{2\log n}{-\log \C}\right)\right\rfloor\right. \\
&\quad\left.\phantom{\left(\frac{2\log n}{-\log \C}\right)}-\avedist n11+\Td 1{n^{\max\left(1-\epsilon,\frac{1}{\beta-1}\right)}}\right).
\end{align*}
For $\beta>3$, the number of BFSs is $n^{(1+\O(\epsilon))\frac{-\log \C}{\log n}\left(\frac{1}{2}\avedist n11-\Td 1{n^{\frac{1}{\beta-1}}}\right)}=n^{(1+\O(\epsilon))\frac{-\log \eta(1)}{\log n}\left(\frac{1}{2}-\frac{1}{\beta-1}\right)\frac{\log n}{\log M_1(\mu)}} = n^{\left(\frac{1}{2}-\frac{1}{\beta-1}+\O(\epsilon)\right)\frac{-\log \eta(1)}{\log M_1(\mu)}}$. Hence, the running time is $\O\left(mN_{\frac{D}{2}}(v)\right) = n^{1+\left(\frac{1}{2}-\frac{1}{\beta-1}+\O(\epsilon)\right)\frac{-\log \eta(1)}{\log M_1(\mu)}}$.

For $2<\beta<3$, the computation is similar, but $M_1(\mu)$ is infinite: the running time is $n^{1+\O(\epsilon)}$. 

Finally, for $1<\beta<2$, if $v$ is the maximum degree vertex, this value is at most $n^{1-(1+\O(\epsilon))\frac{2-\beta}{\beta-1}\left\lfloor\frac{3}{2}+\frac{\beta-1}{2-\beta}-2\right\rfloor}=n^{1+\O(\epsilon)-\frac{2-\beta}{\beta-1}\left\lfloor\frac{\beta-1}{2-\beta}-\frac{1}{2}\right\rfloor}$. The running-time is $mn^{1-\frac{2-\beta}{\beta-1}\left\lfloor\frac{\beta-1}{2-\beta}-\frac{1}{2}\right\rfloor+\O(\epsilon)}=n^{2-\frac{2-\beta}{\beta-1}\left\lfloor\frac{\beta-1}{2-\beta}-\frac{1}{2}\right\rfloor+\O(\epsilon)}$.


\section{The Exact SumSweep Algorithm.} \label{sec:ss}

The \textsc{SumS} algorithm \cite{Borassi2014,Borassi2014a} is based on keeping lower bounds $L(v)$ and upper bounds $U(v)$ on the eccentricity of each vertex $v$. In particular, assume that we have performed BFSes from vertices $s_1,\dots,s_k$: we can set an upper (resp., lower) bound $U(v)=\min_{i=1,\dots,k}(\ecc(s_i)+\dist(s_i,v))$ (resp., $L(v)=\max_{i=1,\dots,k}(\dist(v_i,s))$) on the eccentricity of $v$. Furthermore the algorithm keeps a lower bound $D_L$ (resp., an upper bound $R_U$) on the diameter (resp., radius), defined as the maximum (resp., minimum) eccentricity of a processed vertex $s_i$. As soon as $D_L \geq \min_{v \in V} U(v)$, we can safely output $D_L$ as the diameter; similarly, as soon as $R_U \leq \min_{v \in V} L(v)$, we know that $R_U$ is the exact radius. It remains to define how the vertices $s_1,\dots,s_k$ are chosen: we start by performing a \ssh, and after that we alternatively maximize $L$ and minimize $U$ (obviously, we never choose the same vertex twice). Actually, in order to perform the analysis in the case $\beta>3$, we also need to perform a BFS from a vertex maximizing the degree every $k$ steps, for some constant $k$ (differently from the original \ss\ algorithm).

The analysis for the radius computation is very simple: after the initial \ssh, all lower bounds are tight \whp, and consequently it is enough to perform a further BFS from a vertex minimizing $L$ to obtain the final value of $R_U$. Then, the running time is the same as the running time of the \ssh. For the diameter, the analysis is more complicated, because we have to check when all upper bounds are below the diameter, and the upper bounds are not tight, in general.

Intuitively, if $\beta<3$, the radius is very close to half the diameter, and the first BFS is performed from a radial vertex $s$: consequently, after the first BFS, the upper bound of a vertex $v$ becomes $\ecc(s)+\dist(s,v) \leq D$ if $\dist(s,v) \leq D-\ecc(s) = D-R \approx R=\ecc(s)$. This means that, after this BFS, we have to perform a BFS from each vertex whose distance from $s$ is approximately the eccentricity of $s$, and there are not many such vertices, as shown by \Cref{lem:eccv}. Hence, we obtain that, in this regime, the running time of the \textsc{ExactSumS} algorithm is the same as the running time of the initial \ssh. Conversely, if $\beta>3$, a BFS from a vertex $s$ sets upper bounds smaller than $D$ to all vertices closer to $s$ than $D-\ecc(s)$, and the number of such vertices is close to $M_1(\mu)^{D-\ecc(s)}$. Since $D-\ecc(s)$ is usually $\O(\log n)$, a BFS sets correct bounds to $M_1(\mu)^{\O(\log n)}=n^{\O(1)}$ vertices: hence, we expect the number of BFSes needed to be subquadratic.

\subsection{The Case $1<\beta<3$.}

As we said before, the first BFS is performed from a radial vertex $s$: by \Cref{thm:ecc}, if $\overline{s}$ is a vertex maximizing $\timp s{n^x}$,
$\ecc(s) \leq \ecc(\overline{s}) \leq \timp {\overline{s}}{n^x} + \left\lfloor(1+\epsilon)\left(\Td 1{n^{1-x}}+\frac{\log n}{-\log \C}-x\right)\right\rfloor$. Let $x:=1-\epsilon$, if $1<\beta<2$, $x:=\frac{1}{\beta-1}-\epsilon$ if $2<\beta<3$: this value is at most $\left\lfloor(1+2\epsilon)\left(2+\frac{\log n}{-\log \C}\right)\right\rfloor$. As a consequence, after the first BFS, the algorithm sets upper bounds smaller than $D$ to any vertex closer to $s$ than 
\begin{align*}
D-\ecc(s) &\geq \left\lfloor(1-\O(\epsilon))\left(\avedist n11+\frac{2\log n}{-\log \C}\right)\right\rfloor\\
&\quad-\left\lfloor(1+2\epsilon)\left(2+\frac{\log n}{-\log \C}\right)\right\rfloor  \\
 &\geq (1-\O(\epsilon))\left(\left\lfloor \frac{2\log n}{-\log \C}\right\rfloor-\left\lfloor \frac{\log n}{-\log \C}\right\rfloor+1\right).
\end{align*}

This means that we only have to analyze vertices $v$ such that $D-\ecc(s) \leq \dist(s,v) \leq \timp s{n^x}+\timp v{n^{1-x+\epsilon}}-1$, that is:
\begin{align*}
&\timp v{n^{1-x+\epsilon}} \\
&\geq D-\ecc(s)-\timp s{n^x}+1 \\
&\geq (1-\O(\epsilon))\left(\left\lfloor \frac{2\log n}{-\log \C}\right\rfloor-\left\lfloor \frac{\log n}{-\log \C}\right\rfloor+1\right).
\end{align*}
Let us apply \Cref{ax:dev} with $\alpha$ such that $(1+\epsilon)\left(\Td 1{n^{1-x+\epsilon}}+\alpha\right)=(1-\O(\epsilon))\left(\left\lfloor \frac{2\log n}{-\log \C}\right\rfloor-\left\lfloor \frac{\log n}{-\log \C}\right\rfloor+1\right)$: we obtain that the number of vertices $v$ that do not receive bounds smaller than $D$ is at most 
\begin{align*}
&n\C^{\alpha-1+x-\epsilon}\\
&=n\C^{(1-\O(\epsilon))\left(\left\lfloor \frac{2\log n}{-\log \C}\right\rfloor-\left\lfloor \frac{\log n}{-\log \C}\right\rfloor+1\right)-\Td 1{n^{1-x+\epsilon}}-1+x-\epsilon}\\
&=n\C^{\alpha-1+x-\epsilon} \\
&=n\C^{(1-\O(\epsilon))\left(\left\lfloor \frac{2\log n}{-\log \C}\right\rfloor-\left\lfloor \frac{\log n}{-\log \C}\right\rfloor\right)-\Td 1{n^{1-x+\epsilon}}-1},
\end{align*}
which is smaller than the number of iteration of the \ssh. Hence, the total running time is bounded by the time needed to perform the initial \ssh.

\subsection{The Case $\beta>3$.}

In the case $\beta>3$, the previous argument does not work, because $\avedist n11$ can be small. We need a different technique: we prove that, for each vertex $v$ and for some $x$, either $\timp v{n^x}$ is quite large, or there is a vertex $s$ with high degree that is ``not far'' from $s$. After $\O(k)$ steps, we have performed a BFS from the $k$ vertices with highest degree, and consequently all vertices which are quite close to one of these vertices have bounds smaller than $D$: this means that there are few vertices with upper bound bigger then $D$. Then, since every $\O(1)$ steps, the number of vertices with upper bound bigger than $D$ decreases by $1$, after few more BFSes, all upper bounds are smaller than or equal to $D$.

More formally, let $s_1,\dots,s_k$ be all the vertices with degree bigger than $n^x$: by \Cref{ax:deg}, $k=\frac{n^{1\pm\O(\epsilon)}}{n^{x(\beta-1)}}$, and after at most $\frac{n^{1+\O(\epsilon)}}{n^{x(\beta-1)}}$ BFSes (apart from the initial \ssh), we have performed a BFS from each of these vertices.

We start by estimating $\ecc(s_i)$, because, after the BFS from $s_i$, for each vertex $v$, $U(v) \leq \dist(v,s_i)+\ecc(s_i)$. By \Cref{thm:ecc}, $\ecc(s_i) \leq \timp {s_i}{n^x} + \left\lfloor(1+\epsilon)\left(\Td 1{n^{1-x}}+\frac{\log n}{-\log \C}-x\right)\right\rfloor \leq 1+(1+\epsilon)\left((1-x)\frac{\log n}{\log M_1(\mu)}+\frac{\log n}{-\log \C}\right)\leq (1+2\epsilon)\left((1-x)\frac{\log n}{\log M_1(\mu)}+\frac{\log n}{-\log \C}\right)$. Hence, after the BFS from vertex $s_i$, the upper bound of any vertex $v$ is smaller than $\dist(v,s_i)+(1+2\epsilon)\left((1-x)\frac{\log n}{\log M_1(\mu)}+\frac{\log n}{-\log \C}\right)$, which is smaller than $D$ if $\dist(v,s_i) \leq D-(1+2\epsilon)\left((1-x)\frac{\log n}{\log M_1(\mu)}+\frac{\log n}{-\log \C}\right)=(1+4\epsilon)\left(\frac{x\log n}{\log M_1(\mu)}+\frac{\log n}{-\log \C}\right)$ by \Cref{thm:diameter}. 

Now, we want to compute the number of vertices that are at distance at least $(1+4\epsilon)\left(\frac{x\log n}{\log M_1(\mu)}+\frac{\log n}{-\log \C}\right)$ from each $s_i$. To estimate this quantity, we use the following lemma, which does not follow directly from the four properties (for a proof, see \iffull{\Cref{sec:others}}{the full version of this paper \cite{Borassi2016d}}). 

\begin{lemma}
Assume that $\beta>2$, and let $T$ be the set of vertices with degree at least $n^x$. Then, $\dist(s,T):=\min_{t \in T} \dist(s,t) \leq \timp s{n^{x(\beta-2)+\epsilon}}+1$ \whp.
\end{lemma}

In other words, this lemma says that, for each vertex $v$, $\min_{i=1,\dots,k} \dist(v,s_i) \leq \timp v{n^{x(\beta-2)+\epsilon}}$: hence, after a BFS from each vertex $s_i$ has been performed, the upper bound of $v$ is at most $D$ if $\timp v{n^{x(\beta-2)+\epsilon}}\leq (1+4\epsilon)\left(\frac{x\log n}{\log M_1(\mu)}+\frac{\log n}{-\log \C}\right)$. We conclude that, after $n^{1+\epsilon-x(\beta-1)}$ BFSes, only vertices verifying $\timp v{n^{x(\beta-2)+\epsilon}}> (1+4\epsilon)\left(\frac{x\log n}{\log M_1(\mu)}+\frac{\log n}{-\log \C}\right)$ have upper bounds bigger than $D$.

By \Cref{ax:dev}, the number of vertices that verify the latter inequality is at most $\O\left(n\C^{(1-\O(\epsilon))\left(\frac{x\log n}{\log M_1(\mu)}+\frac{\log n}{-\log \C}-\frac{x(\beta-2)\log n}{\log M_1(\mu)}\right)}\right)=n^{1-\frac{-\log \C}{\log n}\left(\frac{\log n}{-\log \C}-\frac{x(\beta-3)\log n}{\log M_1(\mu)}\right)+\O(\epsilon)}=n^{\frac{x(\beta-3)(-\log \C)}{\log M_1(\mu)}+\O(\epsilon)}$. Hence, by performing $\O\left(n^{\frac{x(\beta-3)(-\log \C)}{\log M_1(\mu)}+\O(\epsilon)}\right)$ more BFSes, the algorithm terminates.

We conclude that the total number of BFSes is at most $\max\left(n^{\frac{x(\beta-3)(-\log \C)}{\log M_1(\mu)}+\O(\epsilon)}, n^{1-x(\beta-1)}\right)$: if we substitute $x=\frac{1}{\beta-1+(\beta-3)\frac{-\log \C}{\log M_1(\mu)}}$, we obtain $n^{\frac{1}{1+\frac{\beta-1}{\beta-3}\frac{\log M_1(\mu)}{\log \C}}+\O(\epsilon)}$. Then, the running time is at most $n^{1+\frac{1}{1+\frac{\beta-1}{\beta-3}\frac{\log M_1(\mu)}{\log \C}}+\O(\epsilon)}$.

\section{Conclusions and Open Problems.} \label{sec:conc}

In this paper, we have developed an axiomatic framework to evaluate heuristics and algorithms that compute metric properties of real-world graphs. The analyses performed in this framework motivate the empirical results obtained by previous works, they let us compare the different algorithms, and they provide more insight into their behavior. Furthermore, these results can be turned into average-case analyses in realistic models of random graphs.

This work leaves several open problems. First of all, it would be useful to improve the analysis with respect to the parameter $\epsilon$, by computing the exact constants instead of simply writing $\O(\epsilon)$.

Furthermore, in some cases, if we ignore $\epsilon$, we have exactly computed the constants appearing in the exponent. However, in other cases, we just proved upper bounds: it would be interesting to understand if these bounds are tight. We conjecture that the bounds for the algorithms to compute the diameter are tight, or almost tight, but the bounds for other algorithms might be improved (for example, to perform the analysis of the \oracle\ algorithm, we used estimates that are probably not optimal).

Finally, it could be interesting to generalize these results to other models: for instance, a possible generalization is to include all Inhomogeneous Random Graphs \cite{Bollobas2007} (while in this paper we only considered only Rank-1 Inhomogeneous Random Graphs). We conjecture that also these graphs satisfy the four properties, because the known asymptotics for diameter and average distance are very similar to the asymptotics obtained in this paper. 

Another possible generalization is to consider directed graphs: most of the algorithms we analyze in this paper can be generalized to the directed case, and the four properties can be generalized, as well. However, in the literature, there are no well-established models of power law random directed graphs: for this reason, it would be interesting to develop natural generalizations of the models considered, and prove that these generalizations satisfy the (generalized) properties.

\bibliographystyle{plain}
\bibliography{library}

\iffull{}{\end{document}}

\newpage

\appendix

\onecolumn

\section{The Model.} \label{sec:model}

We assume the reader to be familiar with the basic notions of graph theory (see, for example,~\cite{Cormen2009}), and we assume that all the graphs we consider are undirected and unweighted. Furthermore, we say that an event $E$ holds asymptotically almost surely or \aas\ if, when $n$ tends to infinity, $\P(E)=1-\o(1)$; it holds with high probability or \whp\ if $\P(E)=1-\o\left(n^{-k}\right)$ for each $k \in \NN$.

In this paper, we consider different models of random graphs: the Configuration Model (CM, \cite{Bollobas1980}), and Rank-1 Inhomogeneous Random Graph models (IRG, \cite{Hofstad2014}, Chapter 3), such as the Chung-Lu model \cite{Chung2006}, and the Norros-Reittu model \cite{Norros2006}. All these models are defined by fixing in advance the number $n$ of nodes, and $n$ weights $\w v$, one for each vertex in the graph. Then, edges are created at random, trying to give $\w v$ outgoing edges to each vertex $v$. We assume that the weights $\w v$ are chosen according to a power law distribution $\lambda$, which is the degree distribution of many real-world graphs \cite{Newman2003}: more specifically, we assume that, for each $d$, the number of vertices with weight bigger than $d$ is $\Theta\left(\frac{n}{d^{\beta-1}}\right)$, for some constant $\beta$.\footnote{In some cases, a stronger definition of power law is used, that is, it is assumed that there are $\Theta\left(\frac{n}{d^{\beta}}\right)$ vertices with degree $d$, for each $d$. However, our proofs still work with the weaker definition.}

After defining the weights, we have to define how we generate the edges:
\begin{itemize}
\item in the CM, we give $\w v$ half-edges, or stubs to a vertex $v$; edges are created by pairing these $M=\sum_{v \in V} \w v$ stubs at random (we assume the number of stubs to be even, by adding a stub to a random vertex if necessary).
\item in IRG, an edge between a vertex $v$ and vertex $w$ is created independently with probability $f(\frac{\w v\w w}{M})$, where $M=\sum_{v \in V} \w v$, and
\begin{itemize}
\item in general, we assume the following:
\begin{itemize}
\item $f$ is derivable at least twice in $0$;
\item $f$ is increasing;
\item $f'(0)=1$;
\item $f(x)=1-\o(x^k)$ for each $k$, when $x$ tends to infinity.
\end{itemize}
\item in the Chung-Lu model, $f(x)=\min(x,1)$;
\item in the Norros-Reittu model, $f(x)=1-e^{-x}$.
\end{itemize}
\end{itemize}

\begin{Remark} 
The first two assumptions in IRG are needed to exclude pathological cases. The third assumption is just needed to simplify notation, but it can be easily lifted by modifying the weights $\w v$: for instance, if $f'(0)=c$, we may multiply all $\w v$s by $\sqrt{c}$, and redefine $f_1(x)=f\left(\frac{x}{c}\right)$, obtaining the same graph with a function verifying $f_1'(0)=1$. The fourth assumption is less natural, and there are models where it is not satisfied, like the Generalized Random Graph model (\cite{Hofstad2014a}, Chapter 6). However, if the average degree is finite (that is, $\beta>2$), the proofs do not need this assumption (in this work, we have chosen to use this assumption in order to simplify the statements).
\end{Remark}

In order to prove our results, we further need some technical assumptions, to avoid pathological cases. In particular, we exclude from our analysis the values of $\beta$ corresponding to the phase transitions: $\beta=2$, and $\beta=3$. Furthermore, in the regime $1<\beta<2$, we have other phase transitions related to the diameter of the graph, which is $\left\lfloor 3+\frac{\beta-1}{2-\beta}\right\rfloor$: we assume that $\frac{\beta-1}{2-\beta}$ is not an integer, and, with abuse of notation, we write $\left\lfloor \frac{\beta-1}{2-\beta}-\epsilon\right\rfloor=\left\lfloor \frac{\beta-1}{2-\beta}\right\rfloor=\left\lceil \frac{\beta-1}{2-\beta}\right\rceil-1=\left\lceil \frac{\beta-1}{2-\beta}+\epsilon\right\rceil-1$. 

Finally, we need a last assumption on the degree distribution $\lambda$: all our metric quantities make sense only if the graph is connected. Hence, we need to assume that $\lambda$ does not contain ``too many vertices'' of small degree, so that \aas\ there is a unique connected component of size $\Theta(n)$, named giant component. All our results hold in the giant component of the graph considered.

In the remainder of this section, we define precisely this assumption, and we further define some more constants that appear in the main theorems. A reader who is not interested in these technicalities might just skip this part, assuming that the graph is connected, and that the main theorems hold (our probabilistic analyses do not depend on the definition of these constants). 

The first definition is the \emph{residual} distribution $\mu$ \cite{Fernholz2007,Hofstad2014,Hofstad2014a}: intuitively, if we choose a random node $v$, and we choose a random neighbor $w$ of $v$, the degree of $w$ is $\mu$-distributed (see \Cref{thm:branchprocess}, in the case $\ell =1$). This distribution is defined as follows.

\begin{Definition} \label{def:residual}
Given a distribution $\lambda$, its first moment $M_1(\lambda)$ is the expected value of a $\lambda$-distributed random variable. The residual distribution $\mu$ of the distribution $\lambda$ is:
\begin{itemize}
\item in the CM, $\mu(i)=\frac{(i+1)\lambda(i+1)}{M_1(\lambda)}$;
\item in IRG, let $\mu'(i)=\frac{i\lambda(i)}{M_1(\lambda)}$: $\mu(i)$ is a Poisson distribution with random parameter $\mu'$.
\end{itemize}
\end{Definition}

In \Cref{sec:small}, we show that the number of vertices at distance $\ell$ from a given vertex $v$ is very close to a $\mu$-distributed branching process $\z \ell$ (for more background on branching processes, we refer to \cite{Athreya1972}). If $M_1(\mu)<1$, this branching process dies \aas: in terms of graphs, it means that the biggest component has size $\O(\log n)$, and there is no giant component. Conversely, if $M_1(\mu)$ is bigger than $1$, then the branching process has an infinite number of descendants with positive probability $p$: in terms of graphs, it means that there is a connected component of size close to $pn$ (see \cite{Hofstad2014a} for a proof). Hence, we assume $M_1(\mu)>1$ and we ignore all vertices that are not in the giant component.

Finally, given a $\mu$-distributed branching process, we may consider only the branches that have an infinite number of descendant (see \cite{Athreya1972}, I.D.12): we obtain another branching process with offspring distribution $\eta$ depending on $\mu$. In particular, our results depend on $\eta(1)$, that is, the probability that an $\eta$-distributed random variable has value $1$. For more information on the value $\eta(1)$, we refer to \cite{Fernholz2007}. In the following, we also assume that $\eta(1)>0$: this is true if and only if $\mu(0) \neq 0$ or $\mu(1)\neq 0$. In IRG, this is automatically implied by the definition of $\mu$, while in the CM this is an additional technical assumption.

\section{Validity of the Properties in Random Graphs: Formal Proof.} \label{sec:proof}

In this section, we prove \Cref{thm:main}, that states that the four properties are \aas\ verified if a graph is generated with the Configuration Model, or with Rank-1 Inhomogeneous Random Graph models. We follow the sketch in \Cref{sec:overviewmain}. In \Cref{sec:probprel} we state some basic lemmas that are used throughout this section, while in \Cref{sec:big} we analyze the size of $\g \ell s$ when $\g \ell s$ is ``big'' (at least $n^\epsilon$). Then, \Cref{sec:small} completes \Cref{sec:big} by analyzing the size of $\g \ell s$ when $\g \ell s$ is small, using branching process approximation. Then, in \Cref{sec:1beta2} we analyze separately the case $1<\beta<2$, which has a different behavior. \Cref{sec:probnum} develop tools to convert probabilistic results into results on the number of vertices satisfying a certain property. Finally, \Cref{sec:distfromsize} proves \Cref{thm:main}, relying on the results of all previous sections, and \Cref{sec:others} proves other results that were used in some analyses (these analyses are marked with $(*)$ in \Cref{tab:summary}).

\subsection{Probabilistic Preliminaries.} \label{sec:probprel}

In this section, we state some basic probabilistic theorems that are used in the proof of our main results. For a more thorough discussion and for their proof, we refer to \cite{Chung2006}.

\begin{lemma}[Multiplicative form of Chernoff bound] \label{lem:chernoffmul}
Let $\X_1,\dots,\X_k$ be independent Bernoulli random variables, and let $\S=\sum_{i=1}^k \X_i$. Then,
\[
\P\left(\S<(1-\epsilon)\E[\S]\right) \leq \left(\frac{e^{-\epsilon}}{(1-\epsilon)^{1-\epsilon}}\right)^{\E[\S]}
\quad\quad\quad
\P\left(\S>(1+\epsilon)\E[\S]\right) \leq \left(\frac{e^{\epsilon}}{(1+\epsilon)^{1+\epsilon}}\right)^{\E[\S]}
\]
\end{lemma}

\begin{lemma}[Hoeffding's inequality] \label{lem:hoeffding}
Let $\X_1,\dots,\X_k$ be independent random variables such that $a_i<\X_i<b_i$ almost surely, and let $\S=\sum_{i=1}^k \X_i$. Then,
\[
\P\left(|\S-\E[\S]| > \lambda\right) \leq 2e^{-\frac{2\lambda^2}{\sum_{i=1}^k|b_i-a_i|^2}}
\]
\end{lemma}

The next lemmas deal with supermartingales and submartingales, which are defined as follows.
\begin{Definition}
Let $\X_1,\dots,\X_k$ be a sequence of random variables, let $\mathcal{F}_1,\dots,\mathcal{F}_k$ be a sequence of $\sigma$-fields such that $\X_1,\dots,\X_i$ are $\mathcal{F}_i$-measurable. The sequence is a martingale if the conditional expectation $\E[\X_{i+1}|\mathcal{F}_i]$ is equal to $\X_i$, it is a supermartingale if $\E[\X_{i+1}|\mathcal{F}_i] \leq \X_i$, and it is a submartingale if $\E[\X_{i+1}|\mathcal{F}_i] \geq \X_i$.
\end{Definition}

The terms ``submartingale'' and ``supermartingale'' have not been used consistently in the literature, since in some works a supermartingale verifies $\E[\X_{i+1}|\mathcal{F}_i] \geq \X_i$, and a submartingale verifies $\E[\X_{i+1}|\mathcal{F}_i] \leq \X_i$ \cite{Chung2006}. In this work, we use the most common definition.

\begin{lemma}[Azuma inequality for supermartingales] \label{lem:azumasuper}
Let $\X_k$ be a supermartingale, and let us assume that $|\X_k-\X_{k+1}|<M$ almost surely. Then, $\P(\X_n-\X_0 \geq t) \leq e^{-\frac{t^2}{2nM^2}}$.
\end{lemma}

\begin{lemma}[Azuma inequality for submartingales] \label{lem:azumasub}
Let $\X_k$ be a submartingale, and let us assume that $|\X_k-\X_{k+1}|<M$ almost surely. Then, $\P(\X_n-\X_0 \leq -t) \leq e^{-\frac{t^2}{2nM^2}}$.
\end{lemma}

\begin{lemma}[strengthened version of Azuma inequality] \label{lem:azumapp}
Let $\X_k$ be a supermartingale associated with a filter $\mathcal{F}$, and assume that $\var(\X_k|\mathcal{F}_{k-1}) \leq  \sigma^2$, and $\X_k-\E(\X_k|\mathcal{F}_{k-1}) \leq M$. Then, 
\[
\P\left(\X_k \geq \X_0+\lambda\right) \leq e^{\frac{-\lambda^2}{2k\sigma^2+\frac{M\lambda}{3}}}.
\]
\end{lemma}

Finally, we need a technical lemma on the sum of power law random variables.

\begin{lemma} \label{lem:sumpowerlaw}
Let $\X=\sum_{i=1}^k \X_i$, where $k$ tends to infinity and the $\X_i$s are power law random variables with exponent $1<\beta<2$. Then, for each $c>0$, $\P\left(\X>k^{\frac{1+c}{\beta-1}}\right)=\O(k^{-c})$.
\end{lemma}
\begin{proof}
For each $i$, $\P\left(\X_i>k^{\frac{1+c}{\beta-1}}\right) =\O\left(\left(k^{-\frac{1+c}{\beta-1}}\right)^{\beta-1}\right)=\O\left(\frac{1}{k^{1+c}}\right)$, and consequently the probability that there exists $i$ such that $X_i>k^{\frac{1+c}{\beta-1}}$ is $\O\left(k^{-c}\right)$.

Conditioned on $\X_i \leq k^{\frac{1+c}{\beta-1}}$ for each $i$,
\begin{align*}
\E\left[\X \right]&=\E\left[\sum_{\ell =1}^\infty |\{i:\X_i>\ell \}|\right] 
=\E\left[\sum_{\ell =1}^{k^{\frac{1+c}{\beta-1}}} |\{i:\X_i>\ell \}|\right] 
=\sum_{\ell =1}^{k^{\frac{1+c}{\beta-1}}} \E\left[|\{i:\X_i>\ell \}|\right]  \\
&=\sum_{\ell =1}^{k^{\frac{1+c}{\beta-1}}} \O(k\ell ^{-\beta+1})
=\O\left(k^{1+\frac{(1+c)(2-\beta)}{\beta-1}}\right)
= k^{\frac{1+c(2-\beta)}{\beta-1}}.
\end{align*}
We conclude that $\P\left(\X>k^{\frac{1+c}{\beta-1}}\right) = \P\left(\X>k^{\frac{1+c}{\beta-1}}\middle| \exists i, \X_i>k^{\frac{1+c}{\beta-1}} \right)\P\left(\exists i, \X_i>k^{\frac{1+c}{\beta-1}} \right)+ \P\left(\X>k^{\frac{1+c}{\beta-1}}\middle| \forall i, \X_i<k^{\frac{1+c}{\beta-1}} \right)\P\left(\forall i, \X_i<k^{\frac{1+c}{\beta-1}} \right) \leq \O\left(k^{-c}\right)+ \frac{1}{k^{\frac{1+c}{\beta-1}}}\E\left(\X>k^{\frac{1+c}{\beta-1}}\middle| \forall i, \X_i<k^{\frac{1+c}{\beta-1}} \right)=\O\left(k^{-c}+ k^{\frac{1+c(2-\beta)}{\beta-1}-\frac{1+c}{\beta-1}}\right)=\O\left(k^{-c}+ k^{-c}\right)=\O\left(k^{-c}\right)$ by Markov inequality.
\end{proof}

\subsection{Big Neighborhoods.} \label{sec:big}

First of all, let us define precisely the typical time needed by a vertex of degree $d$ to reach size $n^x$. In \Cref{sec:axioms}, we defined $\Td d{n^x}$ as the smallest $\ell $ such that $\g \ell s>n^x$, and then we stated which are the typical values of $\Td d{n^x}$ in different regimes. In this section, we do the converse: we define $\F d{S}$ as a function of the degree distributions, and we show that  there is a high chance that $\g{(1-\epsilon)\F dS}s<S<\g{(1+\epsilon)\F dS}s$, if $s$ is a vertex of degree $d$ in the giant component.

\begin{Definition}
In the following for any $0<d<S$, we denote by
\[
\F d{S}=\begin{cases}\log_{M_1(\mu)}\left(\frac{S}{d}\right) & \text{ if $M_1(\mu)$ is finite.} \\
\log_{\frac{1}{\beta-2}}\left(\frac{\log S}{\log d}\right)& \text{ if $\lambda$ is a power law distribution with $2<\beta<3$.}
\end{cases}
\]
\end{Definition}

Following the intuitive proof, in this section we fix $x,y$ bigger than $\epsilon$, and we bound $\timp s{n^y}-\timp s{n^x}$. The main technique used is to prove that, \whp, each neighbor which is big enough satisfies some constraints, and these constraints imply bounds on $\timp s{n^y}-\timp s{n^x}$. More formally, we prove the following theorem.

\begin{theorem} \label{thm:bigneighbors}
For each $0<x<y<1$, $\timp{s}{n^y}-\timp{s}{n^x} \geq (1-\epsilon)\F{n^x}{n^y}$ \aas, and $\timp{s}{n^y}-\timp{s}{n^x} \leq (1+\epsilon)\F{n^x}{n^y}$ \whp.
\end{theorem}

The proof of this theorem is different for the CM and for IRG. In particular, the main tool used to prove this theorem is an estimate on $\g{\ell +1}s$ knowing $\g{\ell }s$: intuitively, in the CM, for each vertex in $\G{\ell }s$ we count how many neighbors it has in $\G{\ell +1}s$, while in IRG we count how many vertices outside $\N \ell s$ have a neighbor in $\G \ell s$.

\subsubsection{Configuration Model.}

Let us assume that we know the structure of $\N \ell s$ (that is, we consider all possible events $\boldsymbol{E}_i$ that describe the structure of $\N \ell s$, and we prove bounds conditioned on $\boldsymbol{E}_i$; finally, though a union bound, we remove the conditioning). Let us define a random variable $\D \ell s$, which measures ``how big a neighbor is''.
\begin{Definition}
Given a graph $G=(V,E)$ generated through the CM, we denote by $\D \ell s$ the set of stubs of vertices in $\G \ell s$, not paired with stubs of vertices in $\G{\ell -1}s$. We denote $\d \ell s=|\D \ell s|$. 
\end{Definition}
In order to make this analysis work, we need to assume that $\w{\N \ell s}<n^{1-\epsilon}$ and $\D \ell s > n^\epsilon$.

Let us consider the following process: we sort all the stubs in $\D \ell s$, obtaining $a_1,\dots,a_{\d \ell s}$, and, starting from $a_1$, we choose uniformly at random the ``companion'' of $a_i$ among all free stubs (if $a_i$ is already paired with a stub $a_j$ for some $j<i$, we do not do anything). The companion of $a_i$ can be one of the following:
\begin{enumerate}
\item a stub of a vertex $v$ that already belongs to $\G{\ell +1}s$ (because another stub of $v$ was already paired with a stub in $\g \ell s$); \label{item:diag}
\item a stub of a ``new'' vertex; \label{item:new}
\item another unpaired stub in $\D \ell s$. \label{item:horiz}
\end{enumerate}

Let us prove that the number of stubs in $\D \ell s$ that are paired with other stubs in $\D \ell s$ is small (\Cref{item:horiz}): at each step, the probability that we choose one of these stubs is the ratio between the number of unpaired stubs in $\D \ell s$ with respect to the total number of unpaired stubs. Since $\w {\N \ell s} <n^{1-\epsilon}$, the number of unpaired stubs in $\D \ell s$ is at most $n^{1-\epsilon}$, and the total number of unpaired stubs is at least $M-n^{1-\epsilon}=M(1-\o(1))$. Hence, the probability that we choose one of these stubs is at most $\frac{n^{1-\epsilon}}{M}<n^{-\epsilon}$. Let $\X_a$ be a Bernoulli random variable which is $1$ if we pair $a$ with another stub inside $\D \ell s$, $0$ if $a$ is already paired when we process it, or if it is paired outside $\D \ell s$ (observe that the number of vertices paired inside $\D \ell s$ is $2\sum_{a \in \D \ell s}\X_a$). We want to apply Azuma's inequality: first, we sort the stubs in $\D \ell s$, obtaining $a_1,\dots,a_{\d \ell s}$. By the previous argument, $\boldsymbol{S}_k=\sum_{i=1}^k \X_{a_i}-kn^{-\epsilon}$ is a supermartingale, and hence $\P(\X_k>\epsilon k) \leq e^{-\frac{\epsilon^2 k^2}{2k}}$: for $k=\D \ell s$, this probability is at most $e^{-\epsilon^3n^\epsilon}$. In conclusion, \whp, at most $2\epsilon\D \ell s$ stubs in $\D \ell s$ are paired to other stubs in $\D \ell s$.

Let us consider a stub $a$ paired outside $\D \ell s$ with a random stub $\boldsymbol{a}'$: if the number of stubs that are already in $\D{\ell +1}s$ is at most $n^{1-\epsilon^2}$, then the probability that $\boldsymbol{a}'$ is already in $\D{\ell +1}s$ is at most $n^{-\epsilon^2}$. In order to solve the case where $\w{\G{\ell +1}s}>n^{1-\epsilon^2}$, let us assume that $\D \ell s<n^{1-\epsilon}$: in this case, since the number of elements in $\D{\ell +1}s$ decreases at most by $1$ at each step, $\D{\ell +1}s \geq n^{1-\epsilon^2}-n^{1-\epsilon} \geq n^{1-\epsilon}$.

Hence, the case that ``almost always'' occurs is that the new stub $\boldsymbol{a}'$ belongs to a ``new'' vertex. Relying on this, we can lower bound $\g{\ell +1}s$: by definition, $\g{\ell +1}s \leq \d \ell s$, and we want to prove that $\g{\ell +1}s \geq (1-\epsilon)\d \ell s$. Since the number of stubs in $\D \ell s$ paired with other stubs in $\D \ell s$ is negligible \whp, we can write $\g{\ell +1}s=\sum_{i=1}^{(1-\epsilon)\d \ell s} \X_{i}$, where $\X_{i}=1$ with probability at least $1-n^{-\epsilon^2}$, $0$ otherwise (note that the $\X_{i}$s are not independent, but if $\D{\ell +1}s<n^{1-\epsilon}$, then $\P(\X_{i}=1) \geq 1-n^{-\epsilon^2}$, as before). We want to apply Azuma's inequality: $\boldsymbol{S}_k=\sum_{i=1}^k \X_{i}-k(1-n^{-\epsilon^2})$ is a submartingale, and hence $\P(\boldsymbol{S}_k<-\epsilon k) \leq e^{-\frac{\epsilon^2 k^2}{2k}}$: for $k=\d \ell s$, this probability is at most $e^{-\epsilon^3n^\epsilon}$. Hence, \whp, $\X_{i} \geq k(1-n^{-\epsilon^2})-\epsilon k \geq (1-2\epsilon)k$, and for $k=\d \ell s$ we have proved the following lemma.

\begin{lemma}  \label{lem:gammafromweightcm}
Given a random graph $G=(V,E)$ generated through the CM and a vertex $s \in V$, if $\d \ell s>n^\epsilon$ and $\w{\N \ell s}<n^{1-\epsilon}$, then $(1-2\epsilon)\d \ell s \leq \g{\ell +1}s \leq \d \ell s$ \whp.
\end{lemma}

\begin{corollary}\label{cor:gammafromweightcm}
For each vertex $s$, let $\timpd{s}{S}$ be the smallest integer such that $\d \ell s>S$. Then, for each $0<x<1$, $\timpd{s}{n^x}+1 \leq \timp{s}{n^x} \leq \timpd{s}{\frac{n^x}{1-\epsilon}}+1$ \whp.
\end{corollary}
\begin{proof}
For the first inequality, if $\ell =\timp s{n^x}$, $\d{\ell -1}x \geq \g \ell s \geq n^x$. For the second inequality, for each $i < \ell -1$, $n^x>\g{i+1}s \geq (1-\epsilon)\d i s$ by the previous lemma. Hence, $\timpd s{\frac{n^x}{1-\epsilon}}$ cannot be smaller than $\ell -1$.
\end{proof}

Hence, in order to understand $\timp s{n^x}-\timp s{n^y}$, we may as well understand $\timpd s{n^x}-\timpd s{n^y}$, and we do it by estimating $\d{\ell +1}s$ from $\d \ell s$. As before, $\d{\ell +1}s = \sum_{a \in \D \ell s} \Y_{a}$, where $\Y_a$ is $0$ if the stub $\boldsymbol{a}$ paired with $a$ is in $\D \ell s$, $-1$ if $\boldsymbol{a}$ is in $\G{\ell +1}s$, otherwise it the number of stubs of the vertex of $\boldsymbol{a}$, minus one (because $\boldsymbol{a}$ is not in $\D{\ell +1}s$). By definition, the distribution of $\Y_{a}$ is very close to $\mu$ (more specifically, $\sum_{k=0}^\infty |\mu(k)-\P(\Y_{a}=k)|<\frac{1}{n^{\epsilon}}$).

It remains to estimate this sum: we need to do it differently for upper and lower bounds, and for different regimes of $\beta$.

\paragraph{Lower bound, $2<\beta<3$.} The probability that at least one of the $\Y_{a}$ is at least $\d \ell s^{\frac{1-\epsilon}{\beta-2}}$ is close to $\P\left(\mu>\d \ell s^{\frac{1-\epsilon}{\beta-2}}\right)$, because, \whp, no visited vertex can have weight bigger than $\d \ell s^{\frac{1-\epsilon}{\beta-2}}$ (otherwise, there would be a $\ell '<\ell $ such that $\d{\ell '}s \geq \d \ell s^{\frac{1-\epsilon}{\beta-2}}$). Hence, the probability that one of the $\Y_{a}$s is at least $\d \ell s^{\frac{1-\epsilon}{\beta-2}}$ is $\Theta\left(\frac{1}{\d \ell s^{\frac{1-\epsilon}{\beta-2}(\beta-2)}}\right)=\Theta\left(\d \ell s^{-1+\epsilon}\right)$. We want to apply Azuma's inequality to prove that at least one of $\Y_{a}$s is bigger than $\d \ell s^{\frac{1-\epsilon}{\beta-2}}$. Let us number the stubs in $\D \ell s$, obtaining $a_1,\dots,a_{\d \ell s}$, and let $\S_k=\sum_{i=0}^k \Y_{a_i}'-ck\d \ell s^{-1+\epsilon}$, where $\Y_{a_i}'=1$ if $\Y_{a_i}>\d \ell s^{\frac{1-\epsilon}{\beta-2}}$, $0$ otherwise, and $c$ is a small enough constant, so that $\S_k$ is a submartingale. Furthermore, $\var(\Y_{a_i}') \leq \E[(\Y_{a_i}')^2] = \E[\Y_{a_i}'])= \O\left(\d \ell s^{\frac{1-\epsilon}{\beta-2}}\right)$. Then, by the strengthened version of Azuma's inequality (\Cref{lem:azumapp}), if $k=\d \ell s$, $\P\left(\S_k \leq \frac{c}{2}k\d \ell s^{-1+\epsilon}\right) \leq e^{-\Omega\left(\frac{k^2 \d \ell s^2}{2k \d \ell s+k\d \ell s}\right)} \leq e^{-\Omega(\d \ell s^\epsilon)} \leq e^{-n^{\epsilon^3}}$. Hence, \whp, $\S_{\d \ell s} \geq \frac{c}{2}k\d \ell s^{-1+\epsilon}>0$, and consequently there is $i$ such that $\Y_{a_i}' \neq 0$. This means that, for each $i$, $\d{\ell +i}s \geq \d \ell s^{\left(\frac{1-\epsilon}{\beta-2}\right)^i}$.

\paragraph{Upper bound, $2<\beta<3$.} By \Cref{lem:sumpowerlaw}, since $\mu$ is a power law with exponent $\beta-1$, and $2<\beta<3$, the probability that $\sum_{a \in \D \ell s} \Y_a$ is bigger than $k^{\frac{1+\epsilon}{\beta-2}}$ is at most $\O(k^{-\epsilon})=\O\left(n^{-\epsilon^2}\right)$. Consequently, by a union bound, $\d{\ell +i}s \leq \d{\ell }s^{\left(\frac{1+\epsilon}{\beta-2}\right)^i}$ for each $i<n^{\epsilon^3}$, with probability $1-\o(1)$.

\paragraph{Lower bound, $\beta>3$.} We cannot apply directly Azuma's inequality to say that $\d{\ell +1}s$ is close to $\E[\d{\ell +1}s] = (1+\o(1))M_1(\mu)\d \ell s$, because $\Y_{a}$ can assume very large values. However, we can ``cut the distribution'', by defining $\Y_{a}'=\Y_{a}$ if $\Y_{a}<N$, $0$ otherwise. If $N$ is big enough, $\E[\Y_{a}']>M_1(\mu)-\epsilon$. By a straightforward application of Azuma's inequality (\Cref{lem:azumasub}), $\d{\ell +1}s \geq \sum_{a \in \D \ell s} \Y_{a} \geq (1-\epsilon)(M_1(\mu)-\epsilon) \d \ell s$ \whp. Consequently, $\d{\ell +i}s \geq (M_1(\mu)-\O(\epsilon))^i \d \ell s$, \whp.

\paragraph{Upper bound, $\beta>3$.} The expected value of $\d{\ell +i}s$ is at most $(M_1(\mu)+\epsilon)^i\d \ell s$. A straightforward application of Markov inequality lets us conclude that $\P\left(\d{\ell +i}s>(M_1(\mu)+\epsilon)^i\d \ell sn^{\epsilon}\right) \leq n^{-\epsilon}$.

\begin{proof}[Proof of \Cref{thm:bigneighbors}, CM]
By \Cref{cor:gammafromweightcm}, $\timpd{s}{n^x}+1 \leq \timp{s}{n^x} \leq \timpd{s}{(1+\epsilon)n^x}+1$. Hence, $\timpd{s}{n^y} -\timpd{s}{(1+\epsilon)n^x}\leq \timp{s}{n^y}-\timp{s}{n^x} \leq \timpd{s}{n^y(1+\epsilon)}-\timpd{s}{n^x}$.

If we apply the lower bounds with $i=\F{Z^0}S(1+\epsilon')$, $\ell =\timpd{s}{n^x}$, we obtain the following.
\begin{itemize}
\item If $2<\beta<3$, either $\n{\ell +j}s>n^{1-\epsilon}$ for some $j<i$, or, \whp, $\d{\ell +i}s \geq \d \ell s^{\left(\frac{1-\epsilon}{\beta-2}\right)^i} \geq n^{x\left(\frac{1-\epsilon}{\beta-2}\right)^{(1+\epsilon')\log_{\frac{1}{\beta-2}}\frac{y}{x}}} 
= n^{xe^{\log\left(\frac{y}{x}\right)(1+\epsilon')\frac{\log \frac{1-\epsilon}{\beta-2}}{\log \frac{1}{\beta-2}}}} \geq (1+\epsilon)n^{y}$ if $\epsilon$ is small enough with respect to $\epsilon'$. In both cases, $\timp{s}{n^y}-\timp{s}{n^x} \leq \timpd{s}{(1+\epsilon)n^y} -\timpd{x}{n^x} \leq \F{Z^0}S(1+\epsilon')$. With a very similar computation, one can conclude that $\timp{s}{n^y}-\timp{s}{n^x} \geq \F{Z^0}S(1-\epsilon')$ \aas\ (the only difference is how to handle the case where $\n{\ell +j}s>n^{1-\epsilon}$: to this purpose, it is enough to observe that if $n^x<n^y<n^{1-\epsilon}$, for the whole process $\n{\ell +j}s<n^y<n^{1-\epsilon}$).
\item If $\beta>3$, as before, either $\n{\ell +j}s>n^{1-\epsilon}$ for some $j<i$, or, \whp, $\d{\ell +i}s \geq \d \ell s(M_1(\mu)-\epsilon)^i \geq \d \ell s(M_1(\mu)-\epsilon)^{(1+\epsilon) \log_{M_1(\mu)}n^{y-x}}=n^xe^{\log(n^{y-x})(1+\epsilon')\frac{M_1(\mu)-\epsilon}{M_1(\mu)}} \geq n^y(1+\epsilon)$ if $\epsilon$ is small enough with respect to $\epsilon'$. We conclude that $\timp{s}{n^y}-\timp{s}{n^x} \leq \timpd{s}{(1+\epsilon)n^y} -\timpd{x}{n^x} \leq \F{Z^0}S(1+\epsilon')$ \whp. A similar computation yields $\timp{s}{n^y}-\timp{s}{n^x} \geq \timpd{s}{n^y} -\timpd{x}{(1+\epsilon)n^x} \leq \F{Z^0}S(1-\epsilon')$ \aas.
\end{itemize} 
\end{proof}

To conclude this section, we prove a stronger upper bound in the case $\beta>3$, which is used in two of our probabilistic analyses.

\begin{lemma} \label{lem:volnextcm}
Assume that $\d \ell s>d_{\max}n^{\epsilon}$, where $d_{\max}$ is the maximum degree in the graph, and that $M_1(\mu)$ is finite. Then, \whp, $\d{\ell +1}s \leq \d \ell s (M_1(\mu)+\epsilon)$.
\end{lemma}
\begin{proof}
We want to apply Azuma's inequality as in the lower bound. More precisely, $\d{\ell +1}s \leq \sum_{i=1}^{\d \ell s} \Y_{a_i}$, where $\E\left[\Y_{a_i}\right]$ is at most $M_1(\mu)+\epsilon$, conditioned on the values of $\Y_{a_j}$ for each $j<i$. Hence, $\sum_{i=1}^{k} \Y_{a_i}-k(M_1(\mu)+\epsilon)$ is a supermartingale, and $\Y_{a_i}<n^{\frac{1}{\beta-1}}<n^{\frac{1}{2}-\epsilon}$ for $\epsilon$ small enough. By Azuma's inequality (\Cref{lem:azumasuper}), $\P\left(\Y_k \leq \epsilon k\right)\leq e^{-\frac{\epsilon^2k^2}{2kd_{\max}}}\leq e^{-\epsilon^3 n^{\epsilon}}$ for $k=\d \ell s>d_{\max}n^{\epsilon}$. We proved that, \whp, $\d{\ell +1}s-\d \ell s(M_1(\mu)+\epsilon)=\sum_{i=1}^{\d \ell s} \Y_{a_i}-\d \ell s(M_1(\mu)+\epsilon) \leq \epsilon \d \ell s$, and consequently $\d{\ell +1}s \leq \d \ell s(M_1(\mu)+2\epsilon)$.
\end{proof}
Combining this lemma with \Cref{lem:gammafromweightcm}, we obtain the following corollary.

\begin{corollary} \label{cor:verybigcm1}
Assume that $d_{\max}n^{\epsilon}<\g \ell s<n^{1-\epsilon}$, where $d_{\max}$ is the maximum degree in the graph, and that $M_1(\mu)$ is finite. Then, \whp, $\g{\ell+1}s \leq \g \ell s (M_1(\mu)+\epsilon)$.
\end{corollary}

\begin{corollary} \label{cor:verybigcm}
For each vertex $v$, and for each $0<x<y<1$ such that $d_{\max}<n^{x-\epsilon}$, $\timp v{n^y}-\timp v{n^x} \geq (1-\epsilon)\log_{M_1(\mu)}n^{y-x}$, \whp.
\end{corollary}

\subsubsection{Inhomogeneous Random Graphs.}

Let us assume that we know the structure of $\N \ell s$. Following the proof for the CM, we define the auxiliary quantity $\d \ell s=\w{\g \ell s}$.
Again, we need to assume that $\w{\N \ell s}<n^{1-\epsilon}$ and that $\d \ell s > n^\epsilon$.

Let $w$ be a vertex with weight at most $\frac{n^{1-\epsilon}}{\d \ell s}$, outside $\N \ell s$: $\P\left(w \notin \g {\ell +1}s\right)=\prod_{v \in \g \ell s} \left(1-f\left(\frac{\w v \w w}{M}\right)\right)=\prod_{v \in \g \ell s} \left(1-(1+\o(1))\left(\frac{\w v \w w}{M}\right)\right) = e^{-\sum_{v \in \g \ell s}(1+\o(1))\left(\frac{\w v \w w}{M}\right)}=e^{-(1+\o(1))\left(\frac{\d \ell s \w w}{M}\right)} = 1-(1+\o(1))\left(\frac{\d \ell s \w w}{M}\right)$. Hence, $\P\left(w \in \g {\ell +1}s\right) = (1+\o(1))\left(\frac{\d \ell s \w w}{M}\right)$, and $\g{\ell +1}s=\sum_{w \notin \N \ell s} \X_w$, where the $\X_w$s are independent Bernoulli random variables with success probability $(1+\o(1))\left(\frac{\d \ell s \w w}{M}\right)$ if this quantity is much smaller than $1$, otherwise $\O(1)$. We want to compute the number of vertices in $\G{\ell +1}s$, knowing $\d \ell s$: first, we observe that the number of vertices with weight at least $\frac{n^{1-\epsilon}}{\d \ell s}$ is $\O\left(n\left(\frac{\d \ell s}{n^{1-\epsilon}}\right)^{\beta-1}\right)=\O\left(\d \ell s \frac{n^{\epsilon(\beta-1)}\d \ell s^{\beta-2}}{n^{\beta-2}}\right)=\O\left(\d \ell s n^{\epsilon(\beta-1)-\sqrt{\epsilon}(\beta-2)}\right)=\o(\d \ell s)$, assuming $\d \ell s<n^{1-\sqrt{\epsilon}}$, and we can safely ignore these vertices. By the multiplicative form of Chernoff bound (\Cref{lem:chernoffmul}), if $\S=\sum_{w \notin \N \ell s,\w w<\frac{n^{1-\epsilon}}{\d \ell s}} \X_w$:
\[
\P\left(\S<(1-\epsilon)\E[\S]\right) \leq \left(\frac{e^{-\epsilon}}{(1-\epsilon)^{1-\epsilon}}\right)^{\E[\S]} \leq e^{(-\epsilon-(1-\epsilon)\log(1-\epsilon))n^\epsilon} \leq e^{-\epsilon^3n^\epsilon}
\]  
\[
\P\left(\S>(1+\epsilon)\E[\S]\right) \leq \left(\frac{e^{\epsilon}}{(1+\epsilon)^{1+\epsilon}}\right)^{\E[\S]} \leq e^{(\epsilon+(1+\epsilon)\log(1+\epsilon))n^\epsilon} \leq e^{-\epsilon^3n^\epsilon}
\]
if $\epsilon$ is small enough. By changing the value of $\epsilon$ with $\sqrt{\epsilon}$, we have proved the following lemma.

\begin{lemma} \label{lem:gammafromweightirg}
Assume that $\w{\N \ell s}<n^{1-\epsilon}$ and that $\d \ell s > n^\epsilon$. Then, $(1-\epsilon)\d \ell s \leq \g{\ell +1}s \leq (1+\epsilon)\d \ell s$ \whp.
\end{lemma}
\begin{corollary} \label{cor:gammafromweightirg}
For each vertex $s$, let $\timpd{s}{S}$ be the smallest integer such that $\d \ell s>S$. Then, for each $0<x<1$, $\timpd{s}{(1-\epsilon)n^x}+1 \leq \timp{s}{n^x} \leq \timpd{s}{(1+\epsilon)n^x}+1$.
\end{corollary}

As in the CM, we need to estimate $\timpd{s}{n^x}-\timpd{s}{n^y}$. To this purpose, we compute $\w{\G{\ell +1}s}$ knowing $\d \ell s$. Using the previous notations, $\w{\G{\ell +1}s}=\sum_{w \notin \N \ell s} \w w\X_w$, and we estimate this sum by considering separately upper and lower bounds, and the different possible values of $\beta$.

\paragraph{Lower bound, $2<\beta<3$.} Let us consider all vertices with weight at least $\d \ell s^{\frac{1-\epsilon}{\beta-2}}$. The probability that one of these vertices is connected to a vertex in $\G \ell s$ is $\Theta\left({\d \ell s^{\frac{1-\epsilon}{\beta-2}+1}}{n^{-1}}\right)=\Theta\left({\d \ell s^{\frac{\beta-1-\epsilon}{\beta-2}}}{n^{-1}}\right)$. Through a straigthforward application of the multiplicative for of Chernoff bound (\Cref{lem:chernoffmul}), since there are $\Theta\left({n}{\d \ell s^{-\frac{1-\epsilon}{\beta-2}}}\right)$ such vertices, we can prove that there is at least one node with weight at least $\d \ell s^{\frac{1-\epsilon}{\beta-2}}$ which is connected to $\G \ell s$, \whp. This means that $\w{\g{\ell +1}s} \geq \d \ell s^{\frac{1-\epsilon}{\beta-2}}$, and consequently, by a union bound, $\w{\G{\ell +i}s} \geq \w{\G{\ell }s}^{\left(\frac{1-\epsilon}{\beta-2}\right)^i}$ for each $i$ such that $\w{\G{\ell +i}s}<n^{1-\epsilon}$.

\paragraph{Upper bound, $2<\beta<3$.} Let us consider all vertices with weight at least $\d \ell s^{\frac{1+\epsilon}{\beta-2}}$. The probability that one of these vertices is connected to a vertex in $\G \ell s$ is $\Theta\left({\d \ell s^{\frac{1+\epsilon}{\beta-2}+1}}{n^{-1}}\right)=\Theta\left({\d \ell s^{\frac{\beta-1+\epsilon}{\beta-2}}}{n^{-1}}\right)$. Since there are $\Theta\left({n}{\d \ell s^{-\frac{(1+\epsilon)(\beta-1)}{\beta-2}}}\right)$ such vertices, by a union bound, the probability that none of these vertices is connected to a vertex in $\G \ell s$ is $1-\Theta\left(\d \ell s^{-\frac{\beta-1+\epsilon-(1+\epsilon)(\beta-1)}{\beta-2}}\right)=1-\Theta\left(n^{-\epsilon^2}\right)$. Conditioned on this event, the expected value of $\w{\G {\ell +1}s}$ is at most $\d \ell s\sum_{\w v < \d \ell s^{\frac{1+\epsilon}{\beta-2}}} \frac{\w v^2}{n}=\Theta\left(\d \ell s^{1+\frac{(1+\epsilon)(3-\beta)}{\beta-2}}\right)=\Theta\left(\d \ell s^{\frac{1+\epsilon(3-\beta)}{\beta-2}}\right)$. By Markov inequality, with probability at least $1-n^{-\epsilon^3}$, $\w{\G{\ell +1}s} \leq \d \ell s^{\frac{1+\epsilon(3-\beta)}{\beta-2}+\epsilon}= \d \ell s^{\frac{1}{\beta-2}+\epsilon}$. As a consequence, $\w{\G{\ell +i}s} \leq \d \ell s^{\left(\frac{1}{\beta-2}\right)^i}$ for each $i<n^{\epsilon^4}$, \aas.

\paragraph{Lower bound, $\beta>3$.} We want to apply Hoeffding's inequality to prove that, if $\d \ell s>n^\epsilon$, $\d{\ell +1}s \geq (1-\epsilon)\E[\d{\ell +1}s]\geq (1-2\epsilon)M_1(\mu)\d \ell s$. To this purpose, let $N$ be a big constant (to be chosen later): $\d{\ell +1}s=(1+\o(1))\sum_{w \in V} \w w \X_w \geq (1+\o(1))\sum_{\w w < N} \w w \X_w$. By Hoeffding's inequality (\Cref{lem:hoeffding}), \whp, $\sum_{\w w < N} \w w \X_w$ is at least $(1-\epsilon)\E\left[\sum_{\w w < N} \w w \X_w\right]$: if we choose $N$ big enough, the latter value is at least $(1-2\epsilon)M_1(\mu)$. This means that $\d{\ell +1}s \geq (1-2\epsilon)M_1(\mu)\d{\ell }s$ \whp, and by a union bound $\d{\ell +i}s \geq (1-2\epsilon)^iM_1(\mu)^i\d{\ell }s$.

\paragraph{Upper bound, $\beta>3$.} Conditioned on $\d{\ell }s$, the expected value of $\d{\ell +1}s$ is at most $\sum_{w \notin \N \ell s} \w w\E[\X_w] = (1+\o(1))\sum_{w \in V} \w w \frac{\d \ell s\w w}{M}=(1+\o(1))\d \ell s \frac{M_2(\lambda)}{M_1(\lambda)}=(1+\o(1))\d \ell s M_1(\mu) \leq (M_1(\mu)+\epsilon)\d \ell s$. By Markov inequality, $\P\left(\d{\ell +i}s>(M_1(\mu)+\epsilon)^i\d \ell sn^{\epsilon}\right) \leq n^{-\epsilon}$.

\begin{proof}[Proof of \Cref{thm:bigneighbors}, IRG]
By \Cref{cor:gammafromweightirg}, $\timpd{s}{(1-\epsilon)n^x}+1 \leq \timp{s}{n^x} \leq \timpd{s}{(1+\epsilon)n^x}+1$. Hence, $\timpd{s}{(1-\epsilon)n^y} -\timpd{s}{(1+\epsilon)n^x}\leq \timp{s}{n^y}-\timp{s}{n^x} \leq \timpd{s}{n^y(1+\epsilon)}-\timpd{s}{n^x(1-\epsilon)}$. By the aforementioned lower bounds on $\d{d+i}s$, the conclusion follows.
\end{proof}

As before, we conclude this section by proving a stronger upper bound in the case $\beta>3$, which is used in two of our probabilistic analyses.

\begin{lemma} \label{lem:volnextirg}
Assume that $\d \ell s>d_{\max}n^{\epsilon}$, where $d_{\max}$ is the maximum degree in the graph, and assume that $\beta>3$. Then, \whp, $\d{\ell +1}s \leq \d \ell s (M_1(\mu)+\epsilon)$.
\end{lemma}
\begin{proof}
As in the lower bound, we write $\d {\ell +1}s=\w{\G{\ell +1}s} \leq \sum_{w \in V} \w w\X_w$, where $\X_w$ is a Bernoulli random variable with success probability $1-\prod_{v \in \G \ell s} \left(1-f\left(\frac{\w v\w w}{n}\right)\right)=1-\prod_{v \in \G \ell s} e^{-(1+\o(1))\frac{\w v\w w}{n}}=1-e^{-(1+\o(1))\frac{\w{\G \ell s}\w w}{n}} = (1+\o(1))\left(\frac{\w{\G \ell s}\w w}{n}\right)$. Hence, $\E\left[\sum_{w \in V} \w w\X_w\right] = (1+\o(1))\left(\sum_{w \in V} \frac{\w w^2\w{\G \ell s}}{n}\right)\leq \left(M_1(\mu)+\epsilon\right)\d \ell s$. A simple application of Hoeffding's inequality (\Cref{lem:hoeffding}) lets us conclude, since $\w w<d_{\max}$ for each $w$.
\end{proof}
Combining this lemma with \Cref{lem:gammafromweightcm}, we obtain the following corollary.

\begin{corollary} \label{cor:verybigirg1}
Assume that $d_{\max}n^{\epsilon}<\G \ell s<n^{1-\epsilon}$, where $d_{\max}$ is the maximum degree in the graph, and that $M_1(\mu)$ is finite. Then, \whp, $\g{\ell +1}s \leq \g \ell s (M_1(\mu)+\epsilon)$.
\end{corollary}

\begin{corollary} \label{cor:verybigirg}
For each vertex $v$, and for each $0<x<y<1$ such that $d_{\max}<n^{x-\epsilon}$, $\timp v{n^y}-\timp v{n^x} \geq (1-\epsilon)\log_{M_1(\mu)}n^{y-x}$, \whp.
\end{corollary}

\subsection{Small Neighborhoods.} \label{sec:small}

Using \Cref{thm:bigneighbors}, we reduced ourselves to prove \Cref{ax:dev} and the bounds in \Cref{tab:tc} for some small values of $x$. We start the proof by formalizing \Cref{item:smallneigh} in \Cref{sec:overviewmain}: the main tool is the relationship between the size $\g \ell s$ of a neighbor of a vertex $s$ and a $\mu$-distributed branching process $\gd \ell s$.

\begin{theorem}\label{thm:branchprocess}
Let $\boldsymbol{G}=(V,\boldsymbol{E})$ be a random graph with degree distribution $\lambda$, let $\mu$ be the corresponding residual distribution, and let $s \in V$. There are multisets $\Gd \ell s$ of vertices such that:
\begin{enumerate}
\item the cardinality $\gd \ell s$ of $\Gd \ell s$ is a $\mu$-distributed branching process;
\item if $\Nd \ell s=\cup_{i=0}^\ell  \Gd is$, $\P\left(\G{\ell +1}s=\Gd{\ell +1}s\middle|\Nd \ell s=\N \ell s\right)=\O\left(\w{\Nd{\ell }s}^2\frac{M_2(\lambda)}{n}\right)$.
\end{enumerate}
\end{theorem}

\subsubsection{Proof for the Configuration Model.}
For each vertex $v$, let us fix a set of stubs $a_{v,1},a_{v,\w v}$ attached to $v$ (let $A$ the set of all stubs). 

We define a procedure that generates a random pairing of stubs (and, hence, a graph), by fixing a vertex $s$ and pairing stubs ``in increasing order of distance from $s$'', obtaining something similar to a breadth-first search (BFS). This way, in order to understand the structure of $\n \ell s$, we only have to consider the first $\ell $ levels of this BFS, and we may ignore how all other stubs are paired.

The procedure keeps the following information:
\begin{itemize}
\item a partial function $\boldsymbol{\alpha}:A \rightarrow A$, that represent a partial pairing of stubs;
\item for each $\ell $, a set $\boldsymbol{I}^\ell $ of all stubs at distance $\ell $ from $s$;
\item for each $\ell $, a set $\Gd \ell s$ of all vertices at distance $\ell $ from $s$.
\end{itemize}

The random part of our procedure is given by a set of random variables $\{\b_{a,i}\}_{a \in A, i \in \NN}$, whose range is the set of stubs. Informally, stub $a$ ``wants to be paired'' with stub $\b_{a,0}$, if available, otherwise $\b_{a,1}$, and so on. We assume that, for each $a$, $A=\{b \in A:\b_{a,i}=b \text{ for some }i\}$ is infinite (this event occurs with probability $1$).

\begin{Definition}
The procedure \pu\ starts with $\au$ as the empty function, $\Gu 0s=\{s\}$, $\Gu \ell s=\emptyset$, $\Iu 0=\{a_{s,1}, \dots, a_{s,\w s}\}, \Iu \ell =\emptyset$ for each $\ell >0$. Then, for increasing values of $\ell $, for each stub $a$ in $\Iu \ell $ (any order is fine):
\begin{enumerate}
\item it sets $\b$ as the first $\b_{a,i}$ such that $\au(\b_{a,i})$ is undefined;
\item it defines $\au(a)=\b$, $\au(\b)=a$;
\item if $\b$ is not in $\Iu \ell$ for any $\ell$:
\begin{enumerate}
\item it adds to $\Iu{\ell +1}$ all stubs of $V(\b)$ except $\b$; \label{item:adds}
\item it adds $V(\b)$ to $\Gu{\ell +1}s$, and it sets $V(a)$ as the father of $V(\b)$; \label{item:addv}
\end{enumerate}
\item else:
\begin{enumerate}
\item it removes $\b$ from $\Iu \ell $. \label{item:rems}
\end{enumerate}
\end{enumerate}
The procedure ends when $\Iu \ell $ is empty, and all remaining stubs are paired uniformly at random (so that $\au$ becomes a total function). 
\end{Definition}

At the end, the pairing $\au$ is uniformly distributed, because, at each step, we choose the ``companion'' of a stub uniformly among all unpaired stubs. Furthermore, if we consider the graph obtained with the pairing $\au$, the set $\Gu \ell s$ is the set of vertices at distance $\ell $ from $s$.

Now, we define another similar, simplified procedure. This time, we let $\b$ be any stub, and we do not test if $\b$ is not in $\Id \ell $ for some $\ell $, and we add it anyway to $\Id {\ell +1}$. This way, $\Id {\ell }$ and $\Gd{\ell }s$ become multisets (that is, repetitions are allowed).

\begin{Definition}
The procedure \pd\ starts with $\Gd 0s=\{s\}$, $\Gd \ell s=\emptyset$ for each $d>0$, $\Id 0=\{a_{s,1}, \dots, a_{s,\w s}\}, \Id \ell =\emptyset$ for each $d>0$. Then, for increasing values of $\ell $, for each stub $a$ in $\Id \ell $ (any order is fine):
\begin{enumerate}
\item it sets $\b=\b_{a,0}$, and it shifts the $\b_{a,i}$s by one;
\item it defines $\ad(a)=\b$, $\ad(\b)=a$ (in case, it replaces its value);
\item in any case:
\begin{enumerate}
\item it adds to $\Id{\ell +1}$ all stubs of $V(\b)$ except $\b$;
\item it adds $V(\b)$ to $\Gd{\ell +1}s$, and it sets $V(a)$ as the father of $V(\b)$;
\end{enumerate}
\end{enumerate}
The procedure ends when $\Id \ell $ is empty, or continues indefinitely.
\end{Definition}

Thanks to these simplifications, we are able to prove that, in procedure \pd, the cardinality $\gd \ell s$ of $\Gd \ell s$ is a branching process, starting from $\gd 1s=\w s$.

\begin{lemma}
In procedure \pd, the stochastic process $\gd \ell s$ is a $\mu$-distributed branching process, starting from $\gd 1s=\w s$.
\end{lemma}
\begin{proof}
First of all, we observe that $\gd{\ell +1}s=|\Id \ell |$, so it is enough to prove that $\Id \ell $ is a branching process. It is clear that $\Id 0s=\w s$. Moreover, $\Id{\ell +1}=\sum_{a \in \Id{\ell }} \w{V(\b_{a,0})}-1$. Let $\X_i:\w{V(\b_{a,0})}-1$: since the $\b_{a,0}$s are independent, also the $\X_i$s are independent, and $\P(\X_i=k)=\P(\w{V(\b_i)})=k+1=\frac{(k+1)n\lambda(k+1)}{\sum_{j=0}^{\infty}jn\lambda(j)}=\mu(k)$.
\end{proof}

With this choice of $\gd \ell s$, we have proved the first part of \Cref{thm:branchprocess}. Now, we have to bound the probability that $\Gd{\ell +1}s \neq \G{\ell +1}s$, assuming $\Nd{\ell }s = \N \ell s$: the next lemma gives a bound which is stronger than the bound in \Cref{thm:branchprocess}, and it concludes the proof.

\begin{lemma}
Let $\N \ell s=\bigsqcup_{i=0}^\ell  \G \ell s$, $\Nd \ell s=\bigsqcup_{i=0}^\ell  \Gd \ell s$, where $\bigsqcup$ denotes the disjoint union, and let $\n \ell s=|\N \ell s|, \nd \ell s=|\Nd \ell s|$. Assuming $\Nd \ell s=\N \ell s$, the probability that $\Gd{\ell +1}s \neq \G{\ell +1}s$ is at most $\frac{1}{n}2\w{\Nd \ell s}^2$.
\end{lemma}
\begin{proof}
Let us pair the stubs in $\Gd \ell s$ one by one, and try to bound the probability that a stub is paired ``differently''. More formally, for each stub $a$ paired by the procedures \pu, \pd, we bound the probability that $a$ is the first stub that was paired differently (so that we can assume that the pairing of all other stubs was the same). In particular, both procedures choose the companion $\boldsymbol{b}$ of $a$ as $\boldsymbol{b}_{a,0}$, if $\boldsymbol{b}_{a,0}$ is not already in $\Nd \ell s$, and it is not already paired with a stub in $\Nd \ell s$. Hence, the probability that the companion of $a$ is the same in the two procedures is at most the probability that $\boldsymbol{b}_{a,0}$ is not in $\Nd \ell s$, and it is not already paired with a stub in $\Nd \ell s$. This probability is at most $\frac{2\w{\Nd \ell s}}{M}$. By a union bound, we can estimate that the probability that at least a stub is paired differently in the two procedures is at most $\w{\Gd \ell s}\frac{2\w{\Nd \ell s}}{M} \leq \frac{2\w{\Nd \ell s}^2}{M} \leq \frac{2\w{\Nd \ell s}^2}{n}$.
\end{proof}

\subsubsection{Proof for Rank-1 Inhomogeneous Random Graphs.}

In this section, we prove \Cref{thm:branchprocess} in IRG. Let us fix a set $V$ of vertices, let us fix the expected degree $\w v$ of each vertex $v\in V$, and let us fix $M=\sum_{v \in V}\w v$. 

Let $s$ be any vertex, and let us define a procedure that considers edges ``in increasing order of distance from $s$'', obtaining something similar to a BFS. This way, in order to understand the structure of $\n \ell s$, we only have to consider the first $d$ levels of this BFS, and we may ignore all other edges.

We denote by $\{\X_{v,w}\}$ a random variable that has value $1$ if the edge $(v,w)$ exists, $0$ otherwise. Note that the $\X_{v,w}$s are independent Bernoulli random variables with success probability $f\left(\frac{\w v\w w}{M}\right)$.

\begin{Definition}
The procedure \pu\ starts with $\Gu 0s=\{s\}$, $\Gu \ell s=\emptyset$. Then, for increasing values of $d$, for each vertex $v \in \Gu \ell s$:
\begin{enumerate}
\item for each vertex $w$ such that $\X_{v,w}=1$:
\begin{enumerate}
\item if $w$ is not in $\bigcup_{i=0}^{\infty}\Gu is$:
\begin{enumerate}
\item add $w$ to $\Gu{\ell +1}s$.
\end{enumerate}
\end{enumerate}
\end{enumerate}
The procedure ends when $\Gu \ell s$ is empty.
\end{Definition}

There are two reasons why this procedure is not a branching process. The first and simplest problem, that occurred also in the CM, is that we need to check that $w$ is ``a new vertex'', and hence there is dependance between the number of children of different vertices. However, there is also a more subtle problem: if we assume that there is no dependency, we can informally write $\gu{\ell +1}s=\sum_{v \in \Gu \ell s} \sum_{w \in V} \X_{v,w}$. To turn this into a branching process, we have to link $\gu \ell s$ with $\gu{\ell +1}s$, but the previous formula also depends on which vertices are in $\Gu \ell s$. If we condition on which vertices we find in $\Gu \ell s$, then the random variables $\sum_{w \in V} \X_{v,w}$ are not identically distributed. So, we have to fix $\gu \ell s$, write $\Gu \ell s=\{\boldsymbol{v}_1,\dots,\boldsymbol{v}_k\}$, where $\boldsymbol{v}_i$ is a random variable taking values in $V$, and then set $\gu{\ell +1}s=\sum_{i=1}^{\Gu \ell s} \sum_{w \in V} \X_{\boldsymbol{v}_i,w}$. Now, the random variables $\sum_{w \in V} \X_{\boldsymbol{v}_i,w}$ are \iid, but the distribution of the weight of $\boldsymbol{v}_i$ (and hence the distribution of the sum) depends on $\g \ell s$ in general, so we do not obtain a branching process. Summarizing, the second problem is that we need somehow to choose $\gd \ell s$ before choosing which vertices are in $\Gd \ell s$, then choose $\boldsymbol{v}_i$ in a way that is independent from $\gd \ell s$. It turns out that, if the random variables $\X_{v,w}$ are Poisson-distributed, actually the two choices can be made independent. So, in the second procedure we define, we do not only ignore already visited vertices, but we also define new Poisson random variables $\Y_{v,w}$ such that $\Y_{v,w}=\X_{v,w}$ with probability  $1-\O\left(\left(\frac{\w v \w w}{M}\right)^2\right)$, and we work with $\Y_{v,w}$.

\begin{lemma}
Given a Bernoulli random variable $\X$ with success probability $f(p)$, it is possible to define a random variable $\Y=\poi(p)$ such that $\X=\Y$ with probability $1-\O\left(p^2\right)$.
\end{lemma}
\begin{proof}
Let $\boldsymbol{E}_0, \boldsymbol{E}_1$ be the events $\X=0$, $\X=1$. Let $\boldsymbol{E}_0'$ be an event such that $\boldsymbol{E}_0' \subseteq \boldsymbol{E}_0$ or $\boldsymbol{E}_0' \supseteq \boldsymbol{E}_0$, and $\P(\boldsymbol{E}_0')=\P(\poi(p)=0)$: we define $\Y=0$ in $\boldsymbol{E}_0'$. Similarly, let $\boldsymbol{E}_1'$ be an event such that $\boldsymbol{E}_1' \subseteq \boldsymbol{E}_1$ or $\boldsymbol{E}_1' \supseteq \boldsymbol{E}_1$, $\boldsymbol{E}_1' \cap \boldsymbol{E}_0'=\emptyset$, and $\P(\boldsymbol{E}_1')=\P(\poi(p)=1)$. We define $\Y=1$ in $\boldsymbol{E}_1'$. Then, we cover the rest of the space as we wish.

We know that $\Y=\X$ on $\boldsymbol{E}_0 \cap \boldsymbol{E}_0'$ and on $\boldsymbol{E}_1 \cap \boldsymbol{E}_1'$: let us prove that the probability of these events is $1-\O\left(p^2\right)$. Indeed, the probability of $\boldsymbol{E}_0$ is $1-f(p)=1-p+\O(p^2)$, the probability of $\boldsymbol{E}_1$ is $p+\O(p^2)$, the probability of $\boldsymbol{E}_0'$ is $e^{-p}=1-p+\O(p^2)$, and the probability of $\boldsymbol{E}_1'$ is $pe^{-p}=p+\O(p^2)$. In any case, 
$\P((\boldsymbol{E}_0 \cap \boldsymbol{E}_0') \cup (\boldsymbol{E}_1 \cap \boldsymbol{E}_1'))=\min(\P(\boldsymbol{E}_0), \P(\boldsymbol{E}_0'))+\min(\P(\boldsymbol{E}_1), \P(\boldsymbol{E}_1'))=1-p+\O(p^2)+p+\O(p^2)=1+\O(p^2)$.
\end{proof}

\begin{Definition}
The procedure \pd\ starts with $\Gd 0s=\{s\}$, $\Gd \ell s=\emptyset$. Then, for increasing values of $\ell $, for each vertex $v \in \Gd \ell s$:
\begin{enumerate}
\item for each vertex $w$:
\begin{enumerate}
\item in any case:
\begin{enumerate}
\item add $\Y_{v,w}$ times $w$ to $\Gd{\ell +1}s$;
\item replace $\Y_{v,w}$ with another $\poi\left(\frac{\w v\w w}{M}\right)$ random variable, independent from all previous events.
\end{enumerate}
\end{enumerate}
\end{enumerate}
The procedure ends when $\Gd \ell s$ is empty, or it continues forever.
\end{Definition}

\begin{theorem}
The cardinality $\gd \ell s$ of $\Gd \ell s$ is a $\mu$-distributed branching process, starting from $\Gd 1s=\deg(s)$.
\end{theorem}
\begin{proof}
In this procedure, we got rid of the dependencies between different zones of the branching tree. Hence, if $\Gd \ell s=\{v_1,\dots,v_{\gd \ell s}\}$, we formalize the previous computation by saying that $\gd{\ell +1}s=\sum_{i=1}^{\gd \ell s} \sum_{w \in V} \Y_{\boldsymbol{v}_i,w}=\sum_{i=1}^{\gd \ell s} \sum_{w \in V} \poi\left(\frac{\w{\boldsymbol{v}_i} \w w}{M}\right)=\sum_{i=1}^{\gd \ell s}\poi\left(\w{\boldsymbol{v}_i}\right)$. 

It only remains to prove that the probability that $\P\left(\w{\boldsymbol{v}_i}=k\right)=\frac{k\lambda(k)}{M_1(\lambda)}$, independently from $\Gd \ell s$. We need the following facts:
\begin{itemize}
\item $\gd \ell s=\sum_{v \in \gd{\ell -1}s}\sum_{w \in V} \poi\left(\frac{\w v \w w}{M}\right)=\poi\left(\w{\gd{\ell -1}s}\right)=\poi(\eta)$ if $\eta=\w{\gd{\ell -1}s}$;
\item if $\boldsymbol{T}_u$ is the number of times that vertex $u$ appears in $\Gd \ell s$, $\boldsymbol{T}_u=\sum_{v \in \gd{\ell -1}s} \poi\left(\frac{\w u \w v}{M}\right)=\poi\left(\frac{\w u \w {\Gd{\ell -1}s}}{M}\right)=\poi(\theta)$ if $\theta=\frac{\w u \w {\Gd{\ell -1}s}}{M}$;
\item conditioned on $\boldsymbol{T}_u=k$, $\gd \ell s-k=\sum_{v \in \Gd{\ell -1}s}\sum_{w \in V-\{u\}} \poi\left(\frac{\w v \w w}{M}\right)=\poi\left(\w{\Gd{\ell -1}s}\frac{M-\w u}{M}\right)=\poi(\eta-\theta)$.
\end{itemize}
Using these three results, we can prove that:
\begin{align*}
\P\left(\boldsymbol{T}_u=k| \gd \ell s=h\right)&=\frac{\P(\gd \ell s=h|\boldsymbol{T}_u=k)\P(\boldsymbol{T}_u=k)}{\P(\gd \ell s=h)} \\
&=\frac{\P(\poi(\eta-\theta)=h-k)\P(\poi(\theta)=k)}{\poi(\eta)=h} \\
&=\frac{e^{-(\eta-\theta)}\frac{(\eta-\theta)^{h-k}}{(h-k)!}e^{-\theta}\frac{\theta^k}{k!}}{e^{-\eta}\frac{\eta^h}{h!}} \\
&=\frac{h!}{k!(h-k)!} \left(\frac{\theta}{\eta}\right)^{k}\left(1-\frac{\theta}{\eta}\right)^{h-k}
=\binom hk \left(\frac{\w u}{M}\right)^{k}\left(1-\frac{\w u}{M}\right)^{h-k}
\end{align*}
Hence, the probability that $u$ appears $k$ times in our process is exactly the probability that $u$ appears $k$ times if we select $\gd{\ell }s$ vertices, by picking $u$ with probability $\frac{\w u}{M}$. Summing over all vertices $u$ with weight $k$, $\P\left(\w{\boldsymbol{v}_i}=k\right)=n\lambda(k)\frac{k}{M} = \frac{k\lambda(k)}{M_1(\lambda)}$. This concludes the proof: indeed, $\gd{\ell +1}s=\sum_{i=1}^{\gd{\ell }s} \poi(\w{\boldsymbol{v}_i})$, and $\P\left(\w{\boldsymbol{v}_i}=k\right)=\frac{k\lambda(k)}{M_1(\lambda)}$: hence, $\poi(\w{\boldsymbol{v}_i})$ is $\mu$-distributed.
\end{proof}
With this choice of $\Gd \ell s$, we have proved the first part of \Cref{thm:branchprocess}. Now, we have to bound the probability that $\Gd{\ell +1}s \neq \G{\ell +1}s$, assuming $\Gd{\ell }s = \G{\ell }s$: the next lemma gives a bound which is stronger than the bound in \Cref{thm:branchprocess}, and it concludes the proof.

\begin{lemma}
Let $\Nd \ell s=\bigsqcup_{i=0}^\ell  \Gd is$, where $\bigsqcup$ denotes the disjoint union, and let $\nd \ell s=|\Nd \ell s|$. Assuming $\Nd \ell s=\N \ell s$, $\P\left(\Gd{\ell +1}s \neq \G{\ell +1}s\right)\leq \frac{1}{n}2\w{\Nd \ell s}^2$.
\end{lemma}
\begin{proof}
The procedures \pu\ and \pd\ behave differently only if one of the following holds:
\begin{enumerate}
\item $\Y_{v,w}>0$ for some $v \in \Gd \ell s, w \in \Nd \ell s$;
\item $\Y_{v,w} \neq \X_{v,w}$ for some $v \in \Gd \ell s, w \in V$
\end{enumerate}
The probability that the first case occurs is $\sum_{v \in \Gd \ell s} \sum_{w \in \Nd \ell s} \frac{\w v\w w}{M} \leq \frac{\w{\Nd \ell s}^2}{M} \leq \frac{\w{\Nd \ell s}^2}{n}$. The probability that the second case occurs is $\sum_{v \in \Gd \ell s} \sum_{w \in V} \O\left(\frac{\w v^2\w w^2}{M^2}\right)=\O\left(\w{\Gd \ell s}^2 \frac{nM_2(\lambda)}{n^2M_1(\lambda)^2}\right)=\O\left(\w{\Gd \ell s}^2 \frac{M_2(\lambda)}{n}\right)$.
\end{proof}

\subsubsection{Bounds for Branching Processes.}

In order to analyze the neighborhood sizes, we need to better understand the behavior of branching processes. For this reason, we need the following lemmas.

\begin{lemma}\label{lem:branchgrowthlow}
Let $\z{}$ be a $\mu$-distributed branching process, let $\ell ,S$ be integers such that $S \leq \log^2 \ell $. Then, for $\ell $ tending to infinity, $\P\left(0<\z \ell <S\right) \leq (\eta(1)+\o(1))^{\ell }$.
\end{lemma}
\begin{proof}
We divide the proof in two different cases: in the first case, we condition on the fact that $\z \ell $ eventually dies (for more background on branching processes conditioned on death/survival, we refer to \cite{Athreya1972}). Conditioned on death, the expected number of descendants after $\ell $ steps is $Z^0\eta(1)^\ell \leq e^{-\ell (-\log \eta(1))}$: by Markov inequality, the probability that $\z \ell \geq 1$ is at most $\E[\z \ell ]=e^{-\ell (-\log \eta(1))}$.

In the second case, let $\zc \ell $ be the process $\z \ell $ conditioned on survival: since $\z \ell  \geq \zc \ell $, it is enough to prove the claim for $\Zc$. We name ``bad'' a step of this process in which $\zc{\ell +1} = \zc \ell $ (note that $\zc{\ell +1} \geq \zc \ell $): let us perform $h=\ell -S$ steps, trying to find $S$ good steps. A step is bad with probability $\eta(1)$, and the probability that at least $\ell -S$ steps are bad is
\[
\sum_{i=\ell -S}^{\ell } \binom{\ell }{i} \eta(1)^i(1-\eta(1))^{\ell -i}\leq S\ell ^{S}\eta(1)^{\ell -S}=e^{\log S+S\log \ell -(-\log(\eta(1)))\left(\ell -S\right)}=e^{-(1+\o(1))\ell (-\log \eta(1))}.
\]
If $i$ is a good step, $\zc i \geq \zc{i-1}+1$, otherwise $\zc i \geq \z{i-1}$: hence, if there are at least $S$ good steps, $\zc \ell  \geq S$.
\end{proof}

\begin{lemma} \label{lem:branchgrowthlow2}
Let $\z{}$ be a $\mu$-distributed branching process with $Z^0=\omega(1)$, and let $S>Z^0$. Then, for each $\ell >(1+\epsilon)\F{Z^0}S$, $\P\left(\z \ell <S\right) \leq e^{-\Omega(Z^0)}+\o(1)^{\ell -\F{Z^0}S}$. If $\mu(0)=0$, $\P\left(\z \ell <S\right) \leq \o(1)^{\ell -\F{Z^0}S}$.
\end{lemma}
\begin{proof}
We can view the branching process as the sum of $Z^0$ different branching processes. A standard theorem in the theory of branching processes \cite[1.A.5, Theorem 1]{Athreya1972} says that the probability that one of this branching processes dies is $z_\mu$, where $z_\mu$ is the only integer between $0$ and $1$ such that $z_\mu=\sum_{i \in \NN} \mu(i)z_{\mu}^i$. Since the different processes are independent, by Chernoff bound, the probability that at least $\frac{Z^0{\mu}}{2}$ processes survive is at least $e^{-\frac{Z^0z_\mu}{8}}=e^{-\Omega(Z^0)}$. Hence, if $\zc{}$ is the process $\z{}$ conditioned on survival, $\zc 0=\Omega(Z^0)$ with probability $e^{-\Omega(Z^0)}$. Furthermore, if $\mu(0)=0$, $\zc 0=Z^0$ by definition.

Then, let us perform $\ell $ steps, and let us estimate $\zc \ell $: a step is ``bad'' if $\zc{i+1} \leq \zc i M_1(\mu)^{1-\epsilon}$, if $M_1(\mu)$ is finite, or $\zc{i+1} \leq \left(\zc{i}\right)^{\frac{1-\epsilon}{\beta-1}}$ if $\mu$ is a power law with exponent $\beta$: it is simple to prove that a step is bad with probability at most $\o(1)$, if $\zc \ell  \geq \zc 0=\Omega(Z^0)$ tends to infinity. If the number of good steps is at least $(1+3\epsilon)\F{Z^0}S$:
\begin{itemize}
\item if $M_1(\mu)$ is finite, $\zc {\ell } \geq \zc 0 M_1(\mu)^{(1-\epsilon)(1+3\epsilon)\F{Z^0}S} \geq \zc 0 M_1(\mu)^{(1+\epsilon)\F{Z^0}S} \geq \zc 0 \frac{S}{Z^0}\omega(1) \geq S$;
\item if $\mu$ is power law with $1<\beta<2$, $\zc {\ell } \geq \left(\zc 0\right)^{\left(\frac{1-\epsilon}{\beta-1}\right)^{(1+3\epsilon)\F{Z^0}S}} \geq \left(\zc 0\right)^{\left(\frac{1}{\beta-1}\right)^{(1+\epsilon)\F{Z^0}S}} \geq e^{\log\left(\zc 0\right)\frac{\log(S)}{\log\left(Z^0\right)}\Omega(1)} \geq S$.
\end{itemize} 

Let $\ell '=\ell -(1+3\epsilon)\F{Z^0}S$ be the maximum number of bad steps: by changing the value of $\epsilon$ in the statement, we can assume that $\ell (1-\epsilon) \geq (1+3\epsilon)\F{Z^0}S$, and hence $\ell '=\ell -(1+3\epsilon)\F{Z^0}S \geq \epsilon \ell $. We need to bind the probability that at least $\ell '$ steps are bad: this is equal to the probability that the sum of $\ell $ Bernoulli variables with success probability $\o(1)$ is at least $\ell '$. This probability is 
\[
\sum_{i=\ell '}^\ell \binom{\ell }{i}(\o(1))^i(1-\o(1))^{n-i} \leq \ell 2^\ell (\o(1))^{\ell '} \leq 2^{\O(\ell ')}(\o(1))^{\ell '}=\o(1)^{\ell '}.
\]
\end{proof}

\begin{corollary} \label{cor:branchgrowthlow}
Let $\Z$ be a $\mu$-distributed branching process, and let $S$ be an integer. If $\ell =\omega(1)$, and $\ell >(1+\epsilon)\F{Z^0}S$, then $\P\left(0<\z {(1+\epsilon)\ell }<S\right) \leq \eta(1)^{\ell -\F{Z^0}S}$.
\end{corollary}
\begin{proof}
If the process dies, by \Cref{lem:branchgrowthlow} it dies before performing $\ell $ steps with probability smaller than $\eta(1)^\ell $. Otherwise, by \Cref{lem:branchgrowthlow}, $\zc {(1+\epsilon)\ell -\F{Z^0}S}\geq \log \ell =\omega(1)$ with probability $1-\frac{\eta(1)^{\ell -\F{Z^0}S}}{2}$, if $\ell $ is big enough. Conditioned on this event, by \Cref{lem:branchgrowthlow2}, $\zc{(1+\epsilon)\ell -\F{Z^0}S+(1+2\epsilon)\F{Z^0}S+\epsilon \ell } \geq S$ with probability $1-\o(1)^{\epsilon \ell + 2\epsilon \F{\log \ell }S)} \geq 1-\o(1)^{\epsilon \ell } \geq 1-\frac{\eta(1)^{\ell }}{2}$. Summing the two probabilities, $\P\left(0<\z{(1+2\epsilon)\ell }<S\right) \leq \eta(1)^{\ell -\F{Z^0}S}$.
\end{proof}

Now, let us prove upper bounds on neighborhood sizes, that correspond to lower bounds on $\timp s{n^x}$.

\begin{lemma}\label{lem:branchgrowthup}
Let us fix $\epsilon>0$, and a $\mu$-distributed branching process $\z \ell $. Given a value $S=\omega(1)$, $\P\left(\forall \ell <(1-\epsilon)\F{Z^0}S,\z {\ell } < S\right) \geq 1-\o(1)$.
\end{lemma}
\begin{proof}
Assume that $M_1(\mu)$ is finite: $\E\left[\z \ell \right] \leq \E\left[\z {(1-\epsilon)\F{Z^0}S}\right]=Z^0M_1(\mu)^{(1-\epsilon)\F{Z^0}S}=Z^0M_1(\mu)^{(1-\epsilon)\log_{M_1(\mu)}\frac{S}{Z^0}}=Z^0\left(\frac{S}{Z^0}\right)^{1-\epsilon}=S^{1-\epsilon}Z_0^\epsilon=S\left(\frac{Z^0}{S}\right)^{\epsilon}$. We conclude by Markov inequality.

Let us consider the case where $\mu$ is a power law distribution with $1<\beta<2$. By \Cref{lem:sumpowerlaw} applied with $k=\z i$, with probability at least $\left(\frac{1}{\z i}\right)^\epsilon$, $\z {i+1}<\left(\z i\right)^{\frac{1+\epsilon}{\beta-1}}$. Let us assume $\z i>\log S$ for each $i$ (increasing the number of elements, we can only increase the number of descendants). Consequently, the probability that $\z \ell $ is bigger than $\left(Z^0\right)^{\left(\frac{1+\epsilon}{\beta-1}\right)^\ell }$ is at most $\sum_{i=1}^\ell  \left(\frac{1}{\z i}\right)^\epsilon \geq \ell  \left(\frac{1}{\log S}\right)^\epsilon=\o(1)$ if $\ell =(1-\epsilon)\F{Z^0}S$. With probability $1-\o(1)$, $\z k <\left(Z^0\right)^{\left(\frac{1+\epsilon}{\beta-1}\right)^\ell }$: since $\ell <(1-\epsilon')\F{Z^0}S$, the claim follows.
\end{proof}

Now, we need to prove a corresponding bound for tail probabilities.
\begin{lemma}\label{lem:branchgrowthuptail}
Let us fix $\epsilon>0$, and a $\mu$-distributed branching process $\z \ell $ such that $\z 0=1$. Given integers $\ell =\omega(1)$, $S$ such that $\F{1}{\ell ^2} \leq \epsilon \F{\ell ^2}{S}$, $\P\left(\forall i<(1-\epsilon)(\F{Z^0}S+\ell ),0< \z {i} < S\right) \geq \eta(1)^\ell $.
\end{lemma}
\begin{proof}
First of all, let $\zc \ell $ be the corresponding branching process conditioned on survival ($\zc 0=1$ with probability $\Omega(1)$). Assuming $\zc 0=1$, the probability that $\zc \ell $ starts with a path of length $\ell $ is $\eta(1)^\ell $. Now, let us estimate $\P\left(\z \ell =k \wedge \zc \ell =1\right) \leq \P\left(\zc \ell =1\middle|\z \ell =k\right) \leq k z_{\mu}^k$, where $z_{\mu}$ is the probability that a $\mu$-distributed branching process has an infinite number of descendants. Hence, $\P(\zc \ell =1 \wedge \z \ell <k) \geq \eta(1)^\ell -\sum_{i=k}^{\infty} iz_{\mu}^{i-1}=\eta(1)^\ell -\frac{kz_{\mu}^{k-1}(1-z_{\mu})+z_{\mu}^k}{(1-z_{\mu})^2}=\eta(1)^\ell -\O\left(kz_\mu^{k-1}\right)$.

For $k=\ell ^2$, $\P(\z \ell <\ell ^2) \geq \P\left(\z \ell <\ell ^2 \wedge \zc \ell =1\right)=\eta(1)^\ell -\O\left(\ell ^2 z_\mu^{\ell ^2-1}\right)=\eta(1)^\ell (1-\o(1))$. Then, let us consider the process $\Z_1$ defined by $\Z_1^k=\z{\ell +k}$: since $\Z_1^0<\ell ^2$, we know by \Cref{lem:branchgrowthup} that $\P\left(\forall k<(1-\epsilon)\F{\ell ^2}S+\ell ,\Z_1^{k} < S\right) \geq 1-\o(1)$. Since the behavior of $\Z_1$ is independent from the behavior of $\Z$, and since $\F{\ell ^2}S =\F{1}{S}-\F{1}{\ell ^2} \geq (1-\epsilon)\F{1}{S}$, 
\[
\P\left(\forall k<(1-2\epsilon)\F{1}S+\ell ,\z{k} < S\right) =\Omega\left(1\cdot\eta(1)^\ell (1-\o(1)) \cdot 1\right)=\Omega\left(\eta(1)^\ell \right).
\]
We conclude by replacing $\ell $ with $(1-\epsilon)\ell $.
\end{proof}

\subsubsection{Bounds on Neighborhood Sizes.}

Now, we need to translate the results in the previous section from the realm of branching processes to the realm of random graphs, using \Cref{thm:branchprocess}. First of all, the following corollary of \Cref{thm:branchprocess} gives us a simpler bound to decide when the branching process approximation works.

\begin{corollary} \label{cor:differ}
Let $G$ be a random graph, $s \in G$. There exists a constant $c_\lambda$ only depending on $\lambda$ such that, for each $\ell <n^{c_\lambda}$ $\P\left(\Nd \ell v \neq \N \ell v\right)=\O\left(n^{-c_\lambda}\right)$, assuming $\w{\Gd is}<n^{c_\lambda}$ for each $i<\ell $. The same is true if we condition on the size of $\Gd \ell v$.
\end{corollary}
\begin{proof}
By \Cref{thm:branchprocess},
\begin{align*}
\P\left(\Nd \ell v \neq \N \ell v\right) &= \sum_{i=1}^\ell  \P\left(\Nd iv \neq \N iv \wedge \Nd {i-1}v = \N {i-1}v\right) \\
&= \sum_{i=1}^\ell  \P\left(\Gd iv \neq \G iv \wedge \Nd {i-1}v = \N {i-1}v\right) \\
&\leq \sum_{i=1}^\ell  \P\left(\Gd iv \neq \G iv \middle| \Nd {i-1}v = \N {i-1}v\right)  \\
&=\O\left(\sum_{i=1}^\ell  \w{\Nd{i}s}^2\frac{M_2(\lambda)}{n}\right)\\
&=\O\left(\ell ^3 n^{2c_\lambda}\frac{M_2(\lambda)}{n}\right) \\
&=\O\left(\frac{M_2(\lambda)}{n^{1-5c_\lambda}}\right).
\end{align*}
We conclude because $M_2(\lambda)=\O(1)$ if $\lambda$ has finite variance, $M_2(\lambda)=n^{3-\beta}$ if $\lambda$ is power law with exponent $2<\beta<3$: this means that it is enough to choose $c_\lambda$ such that $n^{\max(0,3-\beta)-1+5c_\lambda}<n^{-c_\lambda}$, that is, $c_\lambda<\frac{1-\max(0,3-\beta)}{6}$.
\end{proof}

From this corollary, it is easy to translate \Cref{lem:branchgrowthup,lem:branchgrowthuptail} in terms of random graphs, at least for values of $x$ smaller than $c_\lambda$. 
\begin{corollary}
Let $G=(V,E)$ be a random graph, let $s \in V$ be in the giant component, and let $x<c_\lambda$ be a fixed, small enough constant. Then, $\P\left(\timp s{n^x}>(1-\epsilon)\F{\deg(s)}{n^x}\right) \geq 1-\o(1)$. Furthermore, if $s$ has degree $1$, $\P\left(\timp s{n^x}>(1-\epsilon)(\F{1}{n^x}+\alpha)\right) \geq \eta(1)^\alpha$ for $\alpha=\omega(1)$.
\end{corollary}
\begin{proof}
By \Cref{cor:differ}, assuming $x<c_\lambda$, $\Gd \ell s=\g \ell s$ with probability $1-\o(1)$; furthermore, by \Cref{lem:branchgrowthup}, $\Gd \ell s<n^x$ for each $\ell  < (1-\epsilon)\F{\deg(s)}{n^x}$ with probability $1-\o(1)$.

Similarly, by \Cref{lem:branchgrowthuptail}, $\P(\forall i<(1-\epsilon)(\F {\deg(s)}{n^x}+\alpha), \gd is<n^x) \geq \eta(1)^\alpha$ , and since $\g is=\gd is$ with probability $1-\o(1)$ for each $i<\timp s{n^x}$, we conclude that $\P\left(\timp s{n^x}>(1-\epsilon)(\F{1}{n^x}+\alpha)\right) \geq (1-\o(1))\eta(1)^\alpha$.
\end{proof}

The translation of the lower bounds is more complicated: the main problem is that, when $\Nd \ell s \neq \N \ell s$, we know very little on the size of $\Nd \ell s$. In order to deal also with this case, as soon as $\Nd \ell s \neq \N \ell s$, we remove the whole $\N \ell s$ from the graph, and we consider the neighborhood growth of a new vertex $s'$ which was in $\G{\ell +1}s$ in the previous graph. We prove that the behavior of the neighbors of $s'$ in the new graph is ``very similar'' to the behavior of the neighbors of $s$ in the old graph: basically, the only difference is that we re-start from size $1$ instead of $\gd \ell s$. However, this difference is compensated by the fact that the probability that $\Nd \ell s \neq \N \ell s$ is small. In other words, it is more likely that $\gd \ell s$ remains $1$ for $\ell $ steps, rather than that $\Nd \ell s \neq \N \ell s$.

Let us formalize this intuitive proof. First of all, we need to understand what happens when we remove a neighbor from the graph.

\begin{lemma} \label{lem:subgraph}
Let $G=(V,E)$ be a random graph, let $s \in V$, let $\ell  \in \NN$, and let us assume that $\w {\N \ell s}<n^{1-\epsilon}$. Then, conditioned on the structure of $\N \ell s$, the subgraph induced by $V-\N \ell s$ is again a random graph, and the values of $\eta(1)$, $M_1(\mu)$ change by $\O\left(\frac{1}{n^\epsilon}\right)$.
\end{lemma}
\begin{proof}[Proof for the CM]
Let us consider the graph obtained from $G$ by removing all the stubs in $\N \ell s$, and all the stubs paired with stubs in $\N \ell s$. The pairing on the remaining stubs is clearly a random pairing, and the number of stubs removed is at most $n^{1-\epsilon}$, and if $\lambda'$ is the degree distribution of $G-\N \ell s$, $\sum_{i \in \NN} |\lambda(i)-\lambda'(i)|<\frac{1}{n^\epsilon}$. From this condition, it is easy to prove that $\eta(1)$ and $M_1(\mu)$ cannot change by more than $\O(n^\epsilon)$, if $n$ is big enough.
\end{proof}
\begin{proof}[Proof for IRG]
In this case, let us remove $\N \ell s$, and let us consider the probability that two vertices outside $\N \ell s$ are connected: $\P(E(v,w))=f\left(\frac{\w v \w w}{M}\right)=f\left(\frac{\w v \w w}{M-\w {\N \ell s}}\frac{M-\w {\N \ell s}}{M}\right)$. Let $\w v'=\w v\sqrt{\frac{M-\w {\N \ell s}}{M}}$: clearly, $\w v'=\w v(1+\o(1))$, and $G - \N \ell s$ is a random graph with weights $\w v'$. Furthermore, if $\lambda'$ is the degree distribution of $G - \N \ell s$, it is clear that the required conditions are satisfied, because the dependency between $\lambda$, $\mu$, and $\eta$ is continuous.
\end{proof}

Using this lemma, we may translate \Cref{cor:branchgrowthlow} to the context of random graphs.

\begin{lemma}\label{lem:neighgrowthlow}
Let $G$ be a graph with a power law degree distribution $\lambda$ with exponent $\beta$, let $\mu$, $\eta$ be as before. There exists a positive constant $c_\lambda$ only depending on $\lambda$ such that, for each $\ell ,S$ such that $\ell  = \O(\log n)$, $n^\epsilon < S < n^{c_\lambda}$, $\P\left(\forall \ell '<\ell (1+\epsilon), 0<\g {\ell '}s<S\right) =\O\left(\eta(1)^{\ell -\F{Z^0}S}\right)$.
\end{lemma}
\begin{proof}
First of all, we may assume that $\ell -\F{Z^0}S=\omega(1)$, otherwise the probabilistic bound is trivial. By \Cref{cor:differ}, the three following cases are possible:
\begin{itemize}
\item $\N \ell s=\Nd \ell s$;
\item $\w{\G is} \geq 4S n^\epsilon$ for some $i<\ell $;
\item none of the two cases above applies.
\end{itemize}
In the first case, the result follows directly by \Cref{cor:branchgrowthlow}. In the second case, let $i$ be the smallest integer such that $\w{\G is}\geq 4Sn^\epsilon$: in IRG, by \Cref{lem:gammafromweightirg}, $\G{i+1}s \geq (1-\epsilon)\w{\G is}$ \whp, and $\timp s{n^x} < i+1<\ell $. In the CM, $\d{i-1}s+\d is\geq 4S n^\epsilon$, and as a consequence either $\d{i-1}s \geq 2Sn^\epsilon$ or $\d is\geq 2S n^\epsilon$: by \Cref{lem:gammafromweightcm}, $\timp s{n^x} < i+1<\ell $.

It only remains to solve the third case. The probability that this case occurs is $\O\left(n^{-c_\lambda}\right)$ by \Cref{cor:differ}. However, $n^{-c_{\lambda}}$ is not sufficient for our purposes, because $\eta(1)^{\ell -\F{Z^0}S}$ can be much smaller. Let us consider the following process: we explore neighbors of $v$ of increasing size, until we hit a neighbor $i$ verifying $\Gd is \neq \G is$. If $\Gd is \neq \G is$, either all vertices in $\G{i}s$ have all edges directed inside $\G{i}s$, and $\G{i+1}s$ is empty, or there is at least a vertex $v$ with an edge directed outside $\G{i}s$. In the former case, we know that $\g {i+1}s=0$, and the conclusion follows. In the latter case, we remove $\G{i}s$ from the graph: the size of $\G {i+j}s$ is at least the size of $\G j{v'}$, where $v'$ is the neighbor of $v$ outside $\G is$. Furthermore, by \Cref{lem:subgraph}, $G-\G is$ is a random graph, with degree distribution very similar to the degree distribution of $G$: indeed, $i<\ell =\O(\log n)$, and the volume of vertices removed is at most $S \log n \leq n^{1-\epsilon}$. Moreover, the size of the neighbors of $v'$ is independent from all previous events, because all we knew about $v'$ has been removed from the graph. Then, we can restart the exploration from $v'$, in the new graph: if $\Gd j{v'} \neq \G j{v'}$, we proceed again as before. 

More formally, let us fix $\ell $, and let $P(\ell ,h)$ be the probability that $\G \ell s<S$, and that $\Gd js \neq \G js$ happened $h$ times in the aforementioned process. We prove by induction on $h$ that $P(\ell ,h) \leq e^{-(1+\epsilon)(\ell -\F{Z^0}S)(-\log \eta(1))}$. The base case follows by our initial argument. For inductive step, let $\boldsymbol{\ell }'$ be the smallest integer such that $\G {\boldsymbol{\ell }'}s \neq \Gd {\boldsymbol{\ell }'}s$: note that $\P\left(\boldsymbol{\ell }'=i\right) \leq \P\left(\boldsymbol{\ell }'<\ell \right) \leq \frac{n^{-k}}{\ell } \leq n^{-k+\epsilon}$, and that, by inductive hypothesis, $P(\ell -\boldsymbol{\ell }',h)\leq e^{-(-\log\eta(1)+\epsilon)(\ell -\boldsymbol{\ell }')}$ if $\ell -\boldsymbol{\ell }' \geq \log S$, and consequently $P(\ell -\boldsymbol{\ell }',h)\leq e^{-(-\log\eta(1)+\epsilon)(\ell -\boldsymbol{\ell }'-\log S)}$.
\begin{align*}
\P(\ell ,h+1) &\leq \sum_{i=0}^\ell  \P(\boldsymbol{\ell }'=i) P(i,0)P(\ell -i,S,h) \\
&\leq \sum_{i=0}^\ell  n^{-k+\epsilon}e^{-(-\log \eta(1)+\epsilon)(i-\F{Z^0}S)}e^{-(-\log \eta(1)+\epsilon)(\ell -i-\log S-\F{Z^0}S)} \\
&\leq n^{-k+\epsilon}e^{-(-\log \eta(1)+\epsilon)(\ell -2\F{Z^0}S-\log S)}\\
&\leq e^{-(-\log \eta(1)+\epsilon)(\ell -\F{Z^0}S)}e^{-(k-\epsilon)\log n+\F{Z^0}S+\log S}.
\end{align*}
The inductive step is proved, if $e^{-(k-\epsilon)\log n+\F{Z^0}S+\log S}<1$, that is, $(\min(1,\beta-2)-2\log_n S-2\epsilon)\log n>\F{Z^0}S+\log S$, which is implied by $(\min(1,\beta-2)-2\epsilon)\log n>\log_{M_1(\mu)}S+3\log S$, that is, $S\left(3+\frac{1}{\log M_1(\mu)}\right)<n^{\min(1,\beta-2)-2\epsilon}$. The lemma follows by choosing the right value of $c_\lambda$.
\end{proof}

By combining this lemma with \Cref{thm:bigneighbors}, we have proved the following theorem.

\begin{theorem} \label{thm:probgrowth}
Let $G=(V,E)$ be a random graph, let $\lambda$ be the degree distribution of $G$, let $\mu$, $\eta$ be as before, and let $0<x<1$. Then, if $s \in V$, $\deg(v)=d$, the following hold:
\begin{itemize}
\item $\timp s{n^x} \geq (1-\epsilon)\F{d}{n^x}$ \aas;
\item $\P\left(\timp s{n^x} \geq (1+\epsilon)\left(\alpha+\F{d}{n^x}\right)\right)=\O\left(\eta(1)^{\alpha}\right)$;
\item $\P\left(\timp s{n^x} \geq (1-\epsilon)\left(\alpha+\F{d}{n^x}\right)\right)=\Omega\left(\eta(1)^{\alpha}\right)$.
\end{itemize}
\end{theorem}

\subsection{The Case $1<\beta<2$.} \label{sec:1beta2}

In the case $1<\beta<2$, the branching process approximation does not work: indeed, the distribution $\mu$ is not even defined, because $M_1(\lambda)$ is infinite. For this reason, we need a different analysis, which looks similar to the ``big neighbors'' analysis in the case $\beta>2$. We prove that the graph which is generated from this distribution has a very dense core, which is made by all vertices whose degree is big enough: almost all the other vertices are either connected to the core, or isolated, so that the average distance between two nodes is $2$ or $3$. There are also some paths of length $\O(1)$ leaving the core, whose length depends on the value $\beta$ of the distribution.

In our analysis, in order to avoid pathological cases, we have to assume that $\w v<(1-\epsilon)M$ for each $v$ in the Chung-Lu model (otherwise, all vertices with weight at least $1+\epsilon$ would be connected to the maximum degree vertex). Note that this event holds with probability $\O(1)$.

Before entering the details of our analysis, we need some probabilistic lemmas that describe the relationship between the weight and the degree of a vertex.

\begin{lemma}[\cite{Esker2005}, Equation A.1.7] \label{lem:gc}
For each $\epsilon>0$, there exists $N_\epsilon$ and $C_\epsilon$ not depending on $n$ such that the following hold a.a.s.:
\begin{itemize}
\item $M=(1+\epsilon)\sum_{i=1}^{N_\epsilon} \w i$;
\item $\w 1 \leq C_\epsilon\w {N_\epsilon}$.
\end{itemize}
\end{lemma}

\begin{corollary}
The vertex with maximum weight has weight $\Theta\left(n^{\frac{1}{\beta-1}}\right)$ \aas, and $M=\Theta\left(n^{\frac{1}{\beta-1}}\right)$.
\end{corollary}

In this regime, we still need the definitions of $\D \ell s$ as in the case $\beta>2$, but in this case we will not prove that $\g{\ell +1}s$ is close to $\d \ell s$ \whp: for example, let $s$ be a vertex with weight $M=\Theta\left(n^{\frac{1}{\beta-1}}\right)$: clearly, $\g{\ell +1}s$ cannot be $M$, which is bigger than $n$. Indeed, we prove that $\g{\ell +1}s$ is close to $\D \ell s^{\beta-1}$: this way, the number of neighbors of a vertex $s$ with weight $\Theta(M)$ is close to $n$, which makes sense. In order to prove this result, we need a technical lemma on the volume of some vertices.

\begin{lemma} \label{lem:vol}
Given a random graph with a degree distribution $\lambda$ which is power law with exponent $1<\beta<2$, $\sum_{\w w \leq d} \w w=\O\left(nd^{2-\beta}\right)$.
\end{lemma}
\begin{proof}
This result is a simple application of Abel's trick to estimate a sum: 
$\sum_{\w w \leq d} \w w = \sum_{i=1}^{d} i(|\{w:\w w \geq i\}|-|\{w:\w w \geq i+1\}|) =\sum_{i=1}^{d} i|\{w:\w w \geq i\}|-\sum_{i=2}^{d+1}|(i-1)\{w:\w w \geq i\}|) \leq \sum_{i=1}^{d} |\{w:\w w \geq i\}|=\sum_{i=1}^{d} \O\left(\frac{n}{i^{\beta-1}}\right)=\O\left(n\int_{\frac{1}{2}}^{d} x^{1-\beta}dx\right)=\O\left(nd^{2-\beta}\right)$.
\end{proof}

Using this lemma, we can formally prove the relation between $\d \ell s$ and $\g{\ell +1}s$.

\begin{lemma} \label{lem:gammafromweight1beta2}
For each $\epsilon>0$, and for each $\ell $ such that $\n{\ell }s<\d \ell s^{\beta-1}n^{-\epsilon}$, and $\d \ell s>n^\epsilon$, $\g {\ell+1}s = \Theta\left(\d \ell s^{\beta-1}\right)$. 
\end{lemma}
\begin{proof}
Let us fix $\epsilon>0$, and let us prove that $\d{\ell +1}v \geq \g \ell v^{\beta-1}$. Let us consider the set $W$ made by all vertices with weight at least $\frac{M}{\d \ell s}$, not in $\N \ell s$: there are $\Theta\left(n\left(\frac{\d \ell s}{M}\right)^{\beta-1}\right)=\Theta\left(\d \ell s^{\beta-1}\right)$ such vertices, because $\n \ell s<\d \ell s^{\beta-1}$. We want to apply concentration inequalities to prove that there are $\O(|W|)$ vertices in $W$ that are in $\D{\ell +1}s$. First of all, let us assume without loss of generality that $\d \ell s>n^\epsilon$, otherwise this inequality is empty. In the Configuration Model, let us sort the vertices in $W$, obtaining $w_1,\dots,w_k$, and let us consider a procedure where we pair stubs of $w_i$ until we find a connection to $\D \ell s$. Since, at each step, the number of stubs in $\D \ell s$ that are not paired with a vertex in $W$ is $\O(\D \ell s)$, and $w_i$ has $\frac{Mn^\epsilon}{\D \ell s}$ stubs, at each step there is probability $\O(1)$ that $w_i$ is connected to a vertex in $\G \ell s$. A simple application of Azuma's inequality lets us conclude. In IRG, the probability that a vertex $w \in W$ is linked to a vertex in $\G{\ell +1}s$ is at least $\sum_{v \in \G \ell s} f\left(\frac{\w v}{\d \ell s}\right) = \O(1)$: a simple application of the multiplicative form of Chernoff bound (\Cref{lem:chernoffmul}) lets us conclude.

For an upper bound, we can divide the vertices in $\G{\ell +1}s$ in two sets $W,W'$, where $W$ is the set of vertices with weight at most $\frac{M}{\w v}$, $W'=W^C$. For $W'$, the number of vertices with weight at least $\frac{M}{\w v}$ is $\O\left(n\left(\frac{\w v}{M}\right)^{\beta-1}\right)=\O\left(\w v^{\beta-1}\right)$, and hence the number of neighbors of $v$ in $W'$ is at most $\w v^{\beta-1}$. For the set $W$, we have to consider separately IRG and the CM. In the first case, let $\X_{w}=1$ if $w \in \G{\ell +1}s$, $0$ otherwise: we want to estimate $\sum_{w \in W} \X_{w}$. Through the previous lemma, the expected value of this sum is:
\begin{align*}
\E\left[\sum_{w \in W} \X_{w}\right] &= \sum_{w\in W} \sum_{v \in \G \ell s}f\left(\frac{\w v \w w}{M}\right) \\
&= \sum_{w\in W} \sum_{v \in \G \ell s} (1+\o(1))\frac{\w v \w w}{M} \\
&= (1+\o(1))\frac{\d \ell s}{M}\sum_{w\in W} \w w \\
&= \O\left(\frac{\d \ell s}{M}n\left(\frac{M}{\d \ell s}\right)^{2-\beta}\right) \\
& \leq \d \ell s^{\beta-1}.
\end{align*}
Since these random variables are independent, we can apply Chernoff bound to prove that $\sum_{w \in W} \X_w \leq \E\left[\sum_{w \in W} \X_{w}\right]$.

In the CM, let $a_1,\dots,a_{\d \ell s}$ be the stubs in $\D \ell s$. By the previous lemma, the number of stubs in $W$ is $\sum_{w \in W} \w w=\O\left(\frac{nM^{2-\beta}}{\d \ell s^{2-\beta}}\right)=\O\left(\frac{M}{\d \ell s^{2-\beta}}\right)$. The number of vertices in $W \cap \g {\ell +1}s$ is at most the number of stubs in $\d \ell s$ which are paired with stubs in $W$: let us pair stubs in $\d \ell s$ in order: at each step, the probability that we hit a stub in $W$ is $\O\left(\frac{1}{M}\frac{M}{\d \ell s^{2-\beta}}\right)$, because there are still $\O(M)$ stubs outside $W$. A simple application of Azuma's inequality proves that $\g{\ell +1}s \leq \O\left(\frac{\d \ell s}{\d \ell s^{2-\beta}}\right) =\O\left(\d \ell s^{\beta-1}\right)$.
\end{proof}

\begin{corollary} \label{cor:weightdeg}
For each vertex $s$ with degree at least $n^\epsilon$, the number of neighbors of $s$ is $\Theta\left(\w s^{\beta-1}\right)$.
\end{corollary}
\begin{proof}
Apply the previous lemma with $\ell =0$.
\end{proof}

Using the last two results, we can transform statements dealing with the number of vertices to statements dealing with weights. For this reason, we can analyze the weights, which are much simpler.

\begin{lemma} \label{lem:probcon}
The probability that a vertex $v$ with weight $\w v$ is connected to a vertex with weight at most $\rho$ is $\O\left(\frac{n\w v \rho^{2-\beta}}{M}\right)$.
\end{lemma}
\begin{proof}
First, we can assume that $\w v \ll \frac{M}{n\rho^{2-\beta}}$, otherwise the thesis of the lemma is trivially true.

In the CM, let us pair all the stubs of $v$ in order. At each step, the probability that we hit a stub whose vertex has weight at most $\rho$ is $\O\left(\frac{1}{M}\sum_{w:\w w<\rho}{\w w}\right)$, because we have paired at most $\w v \ll M$ vertices. Summing over all stubs of $v$, we obtain $\frac{\w v}{M}\sum_{w:\w w<\rho}{\w w}=\O\left(\frac{n\w v \rho^{2-\beta}}{M}\right)$ by \Cref{lem:vol}.

In IRG, this probability is $\sum_{w:\w w < \rho} f\left(\frac{\w v \w w}{M}\right)=(1+\o(1))\frac{\w v}{M}\sum_{w:\w w<\rho}{\w w}=\O\left(\frac{n\w v \rho^{2-\beta}}{M}\right)$ by \Cref{lem:vol}.
\end{proof}

\begin{lemma} \label{lem:highdeg}
A vertex $v$ with degree at least $\log^2n$ is \whp\ connected to all vertices with weight at least $\epsilon M$.
\end{lemma}
\begin{proof}
In IRG, this lemma follows from our assumptions on $f$. In the CM, let $v$ be a vertex with degree at least $\log^2n$, let $w$ be a vertex with degree at least $\epsilon M$, and let $a_1,\dots,a_k$ be the stubs of $v$. Let us pair the stubs $a_i$ in order: at each step, the probability that $a_i$ is connected to a stub in $w$ is at least $\epsilon$. Hence, by Azuma's inequality (\Cref{lem:azumasub}), at least one of the stubs $a_i$ is connected to a stub in $W$.
\end{proof}

\begin{lemma}
For each vertex $s$ with degree at most $n^{1-\epsilon}$, $\P\left(\timp s{n^{1-\epsilon}} =2\right)=1-\frac{1}{n^{\O(\epsilon)}}$.
\end{lemma}
\begin{proof}
Since $\deg(s)<n^{1-\epsilon}$, $\timp s{n^{1-\epsilon}} \geq 2$. For the lower bound, if $s$ is connected to a vertex with weight $M^{1-\frac{\epsilon}{2}}$, then $\timp s{n^{1-\epsilon}} \leq 2$ by \Cref{cor:weightdeg}. By \Cref{lem:highdeg}, this happens \whp\ if $\deg(v)>\log^2 n$: for this reason, the only remaining case is when $\deg(v)<\log^2 n$, and $v$ is not connected to any vertex with weight $n^{1-\frac{\epsilon}{2}}$. In this case, we prove that $v$ is likely to be isolated: indeed, let us bind the probability that $v$ is connected to a vertex $w$ with degree at most $n^{1-\frac{\epsilon}{2}}$ (hence, with weight at most $M^{1-\epsilon'}$). By \Cref{lem:probcon}, this probability is $\O\left(\frac{n \w v M^{(2-\beta)\left(1-\epsilon'\right)}}{M}\right)=\O\left(n\log n M^{1-\beta}M^{-\epsilon'(2-\beta)}\right)=\O\left(n^{-\epsilon''}\right)$. This means that, by Markov inequality, the number of vertices that are not isolated and not connected to a vertex with degree $n^{1-\epsilon}$ is at most $n^{1-\epsilon''}$, \aas. We conclude that $\Td d{n^{1-\epsilon}} \geq 2+\O\left(n^{-\epsilon''}\right)$.
\end{proof}

Let us now estimate the deviations from this probability.

\begin{lemma} \label{lem:1beta2up}
For each vertex $s$, $\P\left(\timp s{n^{x}}=\ell \right) \leq n^{1-\frac{2-\beta}{\beta-1}(\ell -2-x)+\o(1)}$.
\end{lemma}
\begin{proof}
If $\deg(s)>\log^2 n$, $\timp s{n^x} \leq 2$ \whp. Otherwise, since all vertices with degree at least $\log^2n$ are connected to the vertex with maximum degree, $\timp s{n^x}=\ell $ implies that all vertices at distance at most $\ell -3$ from $s$ have degree at most $\log^2n$. Hence, $\g iv \leq \log^{2\ell }n =n^{\o(1)}$ for each $i \leq \ell -3$. This means that, for each $i \leq \ell -4$, there is a vertex in $\g iv$ with weight $n^{\o(1)}$ connected to another vertex with weight $n^{\o(1)}$. The probability that this happens is at most $n^{-\frac{2-\beta}{\beta-1}+\o(1)}$, because there are $n^{\o(1)}$ such vertices, and we may apply \Cref{lem:probcon} to each of them.

Since these events are independent, if we multiply the probabilities for each $i$ between $0$ and $\ell -4$, the probability becomes $n^{-(\ell -3)\frac{2-\beta}{\beta-1}+\o(1)}$. Finally, all vertices in $\g {\ell -3}s$ should be connected to vertices with degree at most $n^x$, and hence to vertices with weight at most $\O\left(n^{\frac{x}{\beta-1}}\right)$. Again by \Cref{lem:probcon}, the probability that this event happens is $n^{\o(1)}\frac{n n^{\o(1)} n^{\frac{x(2-\beta)}{\beta-1}}}{M}=n^{1-\frac{1}{\beta-1}+\frac{x(2-\beta)}{\beta-1}+\o(1)}=n^{-\frac{2-\beta}{\beta-1}(1-x)+\o(1)}$. Overall, the probability that $\timp s{n^x}=\ell $ is at most $n^{-\frac{2-\beta}{\beta-1}(\ell -2-x)+\o(1)}$.

\end{proof}

\begin{lemma} \label{lem:1beta2low}
For each vertex $s$ with degree $1$, $\P\left(\timp s{n^{x}}=\ell \right) \geq n^{1-\frac{2-\beta}{\beta-1}(\ell -2-x)+\o(1)}$.
\end{lemma}
\begin{proof}
Let $s$ be a vertex of weight $1$: we want to estimate the probability that $s$ is connected to a vertex of weight $2$ in the CM, with weight $1$ in IRG. This probability is $\frac{2n\lambda(2)}{M}=n^{-\frac{2-\beta}{\beta-1}+\o(1)}$ in the CM, $1-\left(1-f\left(\frac{1}{M}\right)\right)^n=1-\left(1-\frac{1+\o(1)}{M}\right)^n=1-e^{\frac{-n(1+\o(1))}{M}}=\frac{n(1+\o(1))}{M}=n^{-\frac{2-\beta}{\beta-1}+\o(1)}$. Assuming this event holds, the probability that $s$ is not connected to any other vertex is $1$ in the CM, and it is $\O(1)$ in IRG, assuming the maximum weight is smaller than $(1-\epsilon)M$. This means that, with probability $\O(1)n^{-\frac{2-\beta}{\beta-1}+\o(1)}=n^{-\frac{2-\beta}{\beta-1}+\o(1)}$, $s$ is connected to a single vertex $s_1$ with weight $1$ in IRG, $2$ in the CM. We may re-iterate the process with $s_1$, finding a new vertex $s_2$, and so on, for $\ell -3$ steps. The probability that we find a path of length $\ell -3$ is $n^{-(\ell -3)\frac{2-\beta}{\beta-1}+\o(1)}$. Then, let us estimate the probability that $s_{\ell -3}$ is connected only to a vertex with degree at most $n^x$. In the CM, the number of stubs of vertices with degree at most $n^x$ is $\sum_{\w w<n^{\frac{x}{\beta-1}}} \w w=\O\left(n^{1+x\frac{2-\beta}{\beta-1}}\right)$, and hence the probability that we hit a stub of a vertex with degree at most
 $n^x$ is $\O\left(\frac{n^{1+x\frac{2-\beta}{\beta-1}}}{M}\right)=n^{-(1-x)\frac{2-\beta}{\beta-1}+\o(1)}$. 
 In IRG, the probability is $1-\prod_{\w w<n^{\frac{x+\o(1)}{\beta-1}}}\left(1-f\left(\frac{\w w}{M}\right)\right)=1-\prod_{\w w<n^{\frac{x+\o(1)}{\beta-1}}}\left(1-\frac{\w w(1+\o(1))}{M}\right)=1-\prod_{\w w<n^{\frac{x+\o(1)}{\beta-1}}}e^{-\frac{\w w(1+\o(1))}{M}}=1-e^{-\frac{n^{1+x\frac{2-\beta}{\beta-1}+\o(1)}}{M}}=1-e^{-n^{-(1-x)\frac{2-\beta}{\beta-1}+\o(1)}}=n^{-(1-x)\frac{2-\beta}{\beta-1}+\o(1)}$.
 
In both cases, we proved that the probability of having a path of length $\ell -2$ followed by a vertex with degree at most $n^x$ is at most $n^{-(\ell -2-x)\frac{2-\beta}{\beta-1}+\o(1)}$. It is clear that in this case $\timp s{n^x} \geq \ell $.

\end{proof}

Summarizing the results obtained in this section, we have proved the following theorem.

\begin{theorem} \label{thm:probgrowth1beta2}
Let $G=(V,E)$ be a random graph with degree distribution $\lambda$, which is power law with $1<\beta<2$. Then, if $s \in V$, $\deg(v)=d$, for each $x$ between $0$ and $1$, the following hold:
\begin{itemize}
\item $\timp s{n^x} \leq 2$ \aas;
\item $\P\left(\timp s{n^x} \geq \alpha+2\right)\leq n\C^{\alpha-x+\o(1)}$;
\item $\P\left(\timp s{n^x} \geq \alpha+2\right) \geq n\C^{\alpha+1-x+\o(1)}$.
\end{itemize}
\end{theorem}

\subsection{Applying the Probabilistic Bounds.} \label{sec:probnum}

Until now, we have proved bounds on the probability that $\timp s{n^x}$ has certain values. In this section, we turn these probabilistic bounds into bounds on the number of vertices that satisfy a given constraint, concluding the proof of the main theorems, and of the values in \Cref{tab:tc}. The main tool used in the following lemma.

\begin{lemma} \label{lem:indep}
For each vertex $t$, let $\boldsymbol{E}_{\ell }(t)$ be an event that only depends on the structure of $\N{\ell }t$. Then, for each set $\boldsymbol{T} \subseteq V$, $0<x<1$, if $\boldsymbol{E}(t)$ is the event $\forall \ell <\timp t{n^x}-1,\boldsymbol{E}_\ell (t)$
\[
|\{t \in \boldsymbol{T}: \boldsymbol{E}(t)\}| =\left(1\pm\o(1)\right)\sum_{t \in \boldsymbol{T}}\P\left(\boldsymbol{E}(t)\middle|t \in \boldsymbol{T}\right)\pm|\boldsymbol{T}|\frac{M^{2x}}{M^{1-\o(1)}}.
\]
If we condition on the structure of a neighbor with volume at most $n^y$, a very similar result holds:
\[
|\{t \in T: \boldsymbol{E}(t)\}| =\left(1\pm\o(1)\right)\sum_{t \in T}\P\left(\boldsymbol{E}(t)\middle|t \in T\right)\pm|T|\left(\frac{M^{x+y}+M^{2x}}{M^{1-\o(1)}}\right).
\]
\end{lemma}
\begin{proof}
First of all, we assume without loss of generality that $|T|<n^{2\epsilon}$, by dividing $T$ in several sets if this is not the case. Let us sort the vertices in $T$, obtaining $t_1,\dots,t_k$, let $\X_i$ be $1$ if $\boldsymbol{E}(t_i)$ holds, $0$ otherwise, and let us assume that we know the structure of $\N{\timp {t_j}{n^x}-2}{t_j}$ for each $j<i$ (in other words, let $\mathfrak{A}_i$ be the $\sigma$-field generated by all possible structures of $\N{\timp {t_j}{n^x}-2}{t_j}$ for each $j<i$, and of the neighbor with volume at most $n^y$). Then, the probability that $\N{\timp {t_i}{n^x}-2}{t_i}$ touches $\N{\timp {t_j}{n^x}-2}{t_j}$ is at most $\sum_{\ell ,\ell ' < \O(\log n)} \P\left(\ell \leq \timp {t_j}{n^x}-2 \wedge \ell '\leq\timp {t_i}{n^x}-2 \wedge \G \ell {t_i} \cap \G{\ell '}{t_j} \neq \emptyset\right) \leq \sum_{\ell ,\ell ' < \O(\log n)}\P\left(\G \ell {t_i} \cap \G{\ell '}{t_j} \neq \emptyset \middle| \ell \leq\timp {t_j}{n^x}-2 \wedge \ell '\leq\timp {t_i}{n^x}-2\right) \leq \O(\log^2 n)\frac{M^{2x}+M^{x+y}}{M^{1-\epsilon}}$, because $\g{\timp {t_i}{n^x}-1}{t_i}<n^x$ for each $i$, and consequently $\d{\timp {t_i}{n^x}-2}{t_i}<M^{x+\epsilon}$ \whp, by \Cref{lem:gammafromweightcm,lem:gammafromweightirg,lem:gammafromweight1beta2}. As a consequence $p_i=\P\left(\boldsymbol{E}(t_i)\right)-\frac{M^{2x+2\epsilon}}{M} \leq \P(\boldsymbol{E}(t_i) | \mathfrak{A}_i) \leq \P\left(\boldsymbol{E}(t_i)\right)+\frac{M^{2x+2\epsilon}}{M}=q_i$.

We have proved that $\boldsymbol{S}_k=\sum_{i=1}^k \X_i-p_i,\boldsymbol{S'}_k= \sum_{i=1}^k q_i-\X_i$ are submartingales. If $p=\sum_{i=1}^k p_i$, by the strengthened version of Azuma's inequality (\Cref{lem:azumapp}), $\P\left(\boldsymbol{S}_k>\epsilon kp\right) \leq e^{-\O\left(\frac{\epsilon^2k^2p^2}{kp+\epsilon kp}\right)} \leq e^{-\epsilon^3kp} \leq e^{-\epsilon^3n^\epsilon}$. This proves that $|\{t \in T: \boldsymbol{E}(t)\}| \geq \left(1-\epsilon\right)\sum_{t \in V}\P\left(\boldsymbol{E}_\ell (t)\middle| \ell <\timp t{n^x}-1,t \in T\right)+|T|\frac{M^{2x}+M^{x+y}}{M^{1-\epsilon}}$, \whp. The other inequality follows from a very similar argument applied to $\boldsymbol{S'}_k$.
\end{proof}

\begin{corollary} \label{cor:numvert}
Let $p=\P(\timp t{n^x})\leq \ell |\deg(t)=d)$, and let us assume that $p>M^{2x+\epsilon-1}$. Then, $(1-\epsilon)p|T| \leq |\{t \in T: \timp s{n^x})\leq \ell \} \leq (1+\epsilon)p|T|$.
\end{corollary}
\begin{proof}
We apply \cref{lem:indep} with $T$ as the set of vertices of degree $d$, $\boldsymbol{E}_\ell (t)$ as the event that $\ell  \leq 2+(1-\epsilon)\F k{n^x}$. We obtain that $|\{t \in T:\timp s{n^y} \leq (1-\epsilon) \F d{n^x}\}|=|\{t \in T:\forall \ell  < \timp s{n^y}-1, \ell  \leq (1-\epsilon) \F d{n^x}-2\}|=(1\pm \o(1))p|T|\pm |T|\frac{M^{2x+\o(1)}}{M}= (1\pm \o(1))p|T|$.
\end{proof}

\subsection{Proof of \Cref{thm:main}.} \label{sec:distfromsize}

\begin{proof}[Proof that \Cref{ax:dev} holds]
For the first statement, if $\deg(s)=n^{\alpha}$ with $\alpha>\epsilon$, in the case $\beta>2$, we know by \Cref{thm:bigneighbors} that $\timp s{n^y} \leq \timp s{n^{\alpha}} + (1+\epsilon)\F{n^{\alpha}}{n^y} \leq 1+(1+\epsilon)\Td{n^{\alpha}}{n^y} \leq (1+2\epsilon)\Td{n^{\alpha}}{n^y}$. In the case $1<\beta<2$, we know by \ref{lem:highdeg} that $s$ is connected to the maximum degree vertex, which has degree $\Theta(n)$: hence, $\timp s{n^x} \leq 2=\Td{n^\alpha}{n^x}$.

For the other statements, if $x$ is small enough, this result follows by \Cref{cor:numvert,thm:probgrowth,thm:probgrowth1beta2}. For bigger values of $x$, we can extend it with \Cref{thm:bigneighbors}.
\end{proof}

\begin{proof}[Proof that \Cref{ax:touch} holds, CM]
Let us recall the definition of $\D \ell {s}$ as the set of stubs of vertices at distance $\ell $ from $s$, not paired with stubs at distance $\ell -1$. We know that $\D \ell s \geq \G{\ell +1}s^{\max\left(\frac{1}{\beta-1}\right)}$ by \cref{lem:gammafromweightcm,lem:gammafromweight1beta2}. For $\ell _s=\timp s{n^x}-1$, $\ell _t=\timp t{n^y}-1$, $\d {\ell _s}s \geq \g{\ell _s+1}s^{\max\left(1,\frac{1}{\beta-1}\right)}n^{-\epsilon} \geq n^{x\max\left(1,\frac{1}{\beta-1}\right)-\epsilon} \geq M^{x-\epsilon}$, and similarly $\d {\ell _t} t \geq M^{y-\epsilon}$. Consequently, $\D {\ell _s} s \D {\ell _t}t \geq M^{x+y-2\epsilon} \geq M^{1+\epsilon'}$. We claim that, \whp, a stub in $\D {\ell _s} s$ is paired with a stub in $\D {\ell _t}t$, and consequently $\dist(s,t) \leq \ell _s+\ell _t+1 =\timp s{n^x}+ \timp t{n^y}-1$, proving the theorem. To prove this claim, let us first observe that if $\N{\ell _s}s$ and $\N {\ell _t} t$ touch each other, then $\dist(s,t) \leq \ell _s+\ell _t < \timp s{n^x}+ \timp t{n^y}-1$, and the result follows. Otherwise, let us assume without loss of generality that $x<y$ (if $x>y$, we swap the roles of $s$ and $t$, if $x=y$, we can decrease $x$ by a small amount, and we change the value of $\epsilon$). Let us consider the $M^{x-\epsilon}$ unpaired stubs $a_1,\dots,a_{M^{x-\epsilon}}$ in $\D {\ell _t}s$, and let us pair these stubs one by one, by defining $\X_i=1$ if the stub is paired to a stub in $\D {\ell _t}t$, $0$ otherwise. Note that, conditioned on all possible pairings of $a_j$ with $j<i$, $\E[\X_i] \geq \frac{M^{y-\epsilon}-M^{x-\epsilon}}{M} \geq \frac{M^{y-2\epsilon}}{M}$. Hence, $\S_k=\frac{kM^{y-2\epsilon}}{M}-\sum_{i=1}^k \X_i$ is a supermartingale, and $\var\left[\X_i\right] \leq \E\left[\X_i^2\right] \leq \E\left[\X_i\right] \leq \frac{M^{y-2\epsilon}}{M}$. By a strengthened version of Azuma's inequality (\Cref{lem:azumapp}), $\P\left(\sum_{i=1}^k \X_i=0\right) \leq \P\left(k\frac{M^{y-2\epsilon}}{M}-\sum_{i=1}^k \X_i<\epsilon i\frac{M^{y-2\epsilon}}{M}\right) \leq e^{\frac{-\epsilon^2k^2M^{2(y-2\epsilon)}}{\Omega\left(kM^{y-2\epsilon}M\right)}} = e^{-\Omega\left(\frac{\epsilon^2 k M^{y-2\epsilon}}{M}\right)}$. For $k=M^{x-\epsilon}$, we have proved that, \whp, the number of stubs in $\D{\ell _s}s$ that are paired with stubs in $\D{\ell _t}t$ is at least $(1-\epsilon)\frac{M^{x+y-3\epsilon}}{M}\geq 1$, and consequently $\dist(s,t)\leq \ell _s+\ell _t < \timp s{n^x}+ \timp t{n^y}-1$.
\end{proof}

\begin{proof}[Proof that \Cref{ax:touch} holds, IRG]
As in \Cref{sec:big}, let $\D \ell s$ be the volume of $\G \ell s$, and let $\ell _s=\timp s{n^x}-1$, $\ell _t=\timp t{n^y}-1$. If $\beta>2$, by \Cref{lem:gammafromweightirg}, $\d{\ell _s}s>(1-\epsilon)M^x$, and $\d{\ell _t}t>(1-\epsilon)M^y$. The probability that a vertex $v \in \G{\ell _s}s$ is not connected to any vertex $w \in \G{\ell _t}t$ is $\prod_{w \in \G{\ell _t}t} 1-f\left(\frac{\w v \w w}{M}\right)$. We have to consider different cases separately.
\begin{itemize}
\item If $\sum_{v \in \G{\ell _s}s,\w v<\frac{M}{M^y}} \w v>\frac{Mn^\epsilon}{M^y}$, by removing some vertices we can assume that all vertices in $\G{\ell _s}s$ have weight at most $\frac{M}{M^y}$. In this case, the number of vertices $v \in \G{\ell _s}s$ having a connection to $\G{\ell _t}{t}$ is $\sum_{v \in \G{\ell _s}s} \X_v$, where the $\X_v$s are independent random variables with success probability $1-\prod_{w \in \G{\ell _t}t} 1-f\left(\frac{\w v \w w}{M}\right) = 1-\prod_{w \in \G{\ell _t}t} e^{-\Omega\left(\frac{\w v \w w}{M}\right)}=1-e^{-\Omega\left(\frac{\w v M^y}{M}\right)}=\Omega\left(\frac{\w v M^y}{M}\right)$. We conclude by a straightforward application of the multiplicative form of Chernoff bound (\Cref{lem:chernoffmul}). 
\item If we do not fall into the previous case, $\sum_{v \in \G{\ell _s}s,\w v<\frac{M}{M^y}} \w v<\frac{Mn^{\epsilon}}{M^y}$, and by slightly decreasing $x$ we can assume without loss of generality that all vertices in $\G{\ell _s}s$ have weight at least $\frac{M}{M^y}$. By changing the roles of $s$ and $t$, we can also assume that all vertices in $\G{\ell _t}t$ have weight at least $\frac{M}{M^x}$. Assuming this, we still have to divide the analysis in two possible cases.
\begin{itemize}
\item if $\g{\ell _s}s\g{\ell _t}t>n^\epsilon$, the number of connections between $\G{\ell _s}s$ and $\G{\ell _t}t$ is at least $\sum_{v \in \G{\ell_s}s,w \in \G{\ell_t}t} \X_{v,w}$, where the $\X_{v,w}$s are independent random variables with success probability $f\left(\frac{\w v \w w}{M}\right) = \Theta(1)$. Since the sum is made by at least $n^\epsilon$ terms, we can conclude by a straightforward application of the multiplicative form of Chernoff bound (\Cref{lem:chernoffmul}). 
\item If $\g{\ell _s}s\g{\ell _t}t<n^\epsilon$, there is at least a vertex $v\in \G{\ell _s}s$ with weight $n^{x-\epsilon}$, and a vertex $w \in \G{\ell _t}t$ with weight $n^{y-\epsilon}$. Then, $\P(E(v,w))=f\left(\frac{\w v \w w}{M}\right)=f\left(\frac{M^{x+y-2\epsilon}}{M}\right) \geq f\left(M^\epsilon\right) \geq 1-\o(M^{\epsilon k})$ for each $k$ (we recall that, in our assumptions, $f(x)=1-\o(x^k)$ for each $k$, if $x$ tends to infinity). We conclude because this means that $v$ is connected to $w$ \whp.
\end{itemize}
\end{itemize}
\end{proof}

\begin{proof}[Proof that \Cref{ax:untouch} holds]
Let us fix $x \geq \frac{1}{2}$, let $s$ be any vertex, and let us fix an integer $\ell _s$ such that $\d {\ell _s}s<M^{x+\epsilon}$. Let us consider a vertex $t \in W$, and let $\ell _t$ be an integer such that $\d{\ell _t}t<M^{y+\epsilon}$: if $\boldsymbol{E}\left(\d {\ell _s}s,\d {\ell _t}t\right)$ is the event that there is an edge between $\D {\ell _s}s$ and $\D {\ell _t}t$, $\P\left(\boldsymbol{E}\left(\d {\ell _s}s,\d {\ell _t}t\right)\middle|\d {\ell _s}s<M^{x+\epsilon},\d {\ell }t<M^{y+\epsilon}\right)<\frac{M^{x+y+3\epsilon}}{M}$. Hence,
\begin{align*}
&P\left(\exists \ell _s,\ell _t: \d {\ell _s}s<M^{x+\epsilon} \wedge \d {\ell _t}t<M^{y+\epsilon} \wedge \boldsymbol{E}\left(\d {\ell _s}s,\d {\ell _t}t\right)\right) \\
& \leq \sum_{\ell _s,\ell _t=0}^{\O(\log n)}\P\left(\d {\ell _s}s<M^{x+\epsilon} \wedge \d {\ell _t}t<M^{y+\epsilon} \wedge \boldsymbol{E}\left(\d {\ell _s}s,\d {\ell _t}t\right)\right) \\
& \leq \sum_{\ell _s,\ell _t=0}^{\O(\log n)}\P\left(\boldsymbol{E}\left(\d {\ell _s}s,\d {\ell _t}t\right)\middle|\d {\ell _s}s<M^{x+\epsilon},\d {\ell _t}t<M^{y+\epsilon}\right) \\
& \leq \frac{M^{x+y+4\epsilon}}{M}.
\end{align*}This means that, with probability $1-\frac{M^{x+y+4\epsilon}}{M}$, $\dist(s,t) \geq \ell _s+\ell _t+2$, where $\ell _s$ (resp., $\ell _t$) is the maximum integer such that $\d{\ell _s}s<M^{x+\epsilon}$ (resp., $\d{\ell _t}t<M^{y+\epsilon}$). By definition of $\ell _s,\ell _t$, $\d{\ell _s+1}s>n^{x+\epsilon}$, and by \Cref{lem:gammafromweightcm,lem:gammafromweightirg,lem:gammafromweight1beta2}, $\g{\ell _s+2}s>n^x$, meaning that $\timp s{n^x} \leq \ell _s+2$, \whp. Since the same holds for $t$, $\dist(s,t) \geq \ell _s+\ell _t+2 \geq \timp s{n^x} +\timp t{n^y} -2$, with probability $1-\frac{M^{x+y+4\epsilon}}{M}$.

We have to translate this probabilistic result into a result on the number of vertices $t$ such that $\dist(s,t) < \timp s{n^x} +\timp t{n^y}-2$. To this purpose, we apply \Cref{lem:indep}, by fixing $s$, conditioning on $\N {\timp s{n^x}-2}s$ (which has volume at most $n^x$), and defining $\boldsymbol{E}(t)$ as $\dist(s,t) < \timp s{n^x} +\timp t{n^y}-2$. Since $y<x$, and $x+y<1$, $|\{t \in T:\dist(s,t)<\timp s{n^x}\} \leq (1+\o(1))|T|\frac{M^{x+y+4\epsilon-1}}{M}\pm |T|\frac{M^{x+y}+M^{2x}}{M^{1-\o(1)}} \leq |T|M^{x+y+5\epsilon-1}$.
\end{proof}

\begin{proof}[Proof that \Cref{ax:deg} holds]
For values of $d$ bigger than $n^\epsilon$, by \Cref{lem:gammafromweightcm,lem:gammafromweightirg,lem:gammafromweight1beta2}, a vertex with weight $d$ has degree $\Theta\left(d^{\max(1,n^{\frac{1}{\beta-1}})}\right)$. Hence, since the number of vertices with weight at least $d$ is $\Theta\left(\frac{n}{d^{\beta-1}}\right)$, the conclusion follows.

For smaller values of $d$, a vertex with weight $d$ has degree bigger than $\frac{1}{2}d^{\max\left(1,\frac{1}{\beta-1}\right)}$ with probability $p=\O(1)$: through simple concentration inequalities it is possible to prove that the degree number of vertices with degree at least $\frac{1}{2}d^{\max\left(1,\frac{1}{\beta-1}\right)}$ is $\O(|\{v \in V:\w v \geq d\}|)=\O(\frac{n}{d^{\beta-1}})$. By defining $d'=d^{\max\left(1,\frac{1}{\beta-1}\right)}$, we conclude.

\end{proof}

\subsection{Other Results.} \label{sec:others}

Before concluding, we need to prove some lemmas that are used in some probabilistic analyses, even if they do not follow from the main theorems.

\begin{lemma} \label{lem:distfromset}
Assume that $\beta>2$, and let $T$ be the set of vertices with degree at least $n^x$. Then, $\dist(s,T):=\min_{t \in T} \dist(s,t) \leq \timp s{n^{x(\beta-2)+\epsilon}}+1$ \whp.
\end{lemma}
\begin{proof}
By removing some vertices from $T$, we can redefine $T$ as the set of vertices with weight at least $n^{x+\epsilon}$ (because each vertex with weight at least $n^{x+\epsilon}$ has degree at least $n^x$ by \Cref{lem:gammafromweightcm,lem:gammafromweightirg,lem:gammafromweight1beta2}). After this modification, the number of vertices in $T$ is $\Theta\left(\frac{n}{n^{(x+\epsilon)(\beta-1)}}\right)=\Theta\left(n^{1-(x+\epsilon)(\beta-1)}\right)$, and the volume of $T$ is $\Omega\left(n^{1-(x+\epsilon)(\beta-1)+x+\epsilon}\right)=\Omega\left(n^{1-(x+\epsilon)(\beta-2)}\right)$. We recall the definition of $\d \ell s$: in the CM, it is the number of stubs at distance $\ell $ from $s$, not paired with stubs at distance $\ell -1$, while in IRG it is the volume of the set of vertices at distance $\ell $ from $s$. By \cref{lem:gammafromweightcm,lem:gammafromweightirg,lem:gammafromweight1beta2}, if $\ell =\timp s{n^{(x+3\epsilon)(\beta-2)}}-1$, $\d{\ell }s \geq n^{(x+2\epsilon)(\beta-2)}$. In the CM, since the pairing of stubs is random, there is \whp\ a stub in $\D \ell s$ which is paired with a stub of a vertex in $T$. In IRG, the probability that a vertex in $\G \ell s$ is paired with a vertex in $T$ is at least $\sum_{v \in \G \ell s} \sum_{t \in T} \X_{vt}$, where the $\X_{vt}s$ are Bernoulli random variables with success probability $f\left(\frac{\w v\w t}{M}\right)$. We conclude by a straightforward application of Chernoff bound (\cref{lem:chernoffmul}).
\end{proof}

\begin{lemma} \label{lem:subsetdegneigh}
Given a vertex $v$ and an integer $\ell $, assume that $n^\epsilon<\g \ell v<n^{1-\epsilon}$, and let $S=\{s \in V:n^\alpha<\deg(s)<n^{\alpha+\epsilon}\}$, for some $\alpha>0$. Then, $|S \cap \G \ell v| \leq \g \ell v |S|n^{-1+\alpha+\epsilon}$ \whp.
\end{lemma}
\begin{proof}[Proof for the CM]
By \Cref{lem:gammafromweightcm}, we can assume that $n^\epsilon \leq \d {\ell -1}v \leq n^{1-\epsilon}$. Let us sort the stubs in $\D{\ell -1}v$, and let $\X_i$ be $1$ if the $i$-th stub is paired with a stub of a vertex in $S$, $0$ otherwise. Clearly, $|S \cap \G \ell v| \leq \sum_{i=1}^{\d {\ell -1}v}\X_i$. Since $\d{\ell -1}v<n^{1-\epsilon}$, conditioned on the outcome of the previous variables $\X_j$, $\P\left(\X_i=1\right)=\O\left(\frac{1}{n}\sum_{v \in S} \w v\right) \leq |S|n^{-1+\alpha+\epsilon}$ (because we have already paired at most $\o(n)$ stubs). Hence, $\mathbf{S}_k=k|S|n^{-1+\alpha+\epsilon}-\sum_{i=0}^k \X_i$ is a submartingale, and if $k=\d{\ell -1}v$, by the strenghtened version of Azuma's inequality (\Cref{lem:azumapp}), \whp, $\S_k \geq -k|S|n^{-1+\alpha+\epsilon}$, that is, $k|S|n^{-1+\alpha+\epsilon}-\sum_{i=0}^k \X_i\geq -k|S|n^{-1+\alpha+\epsilon}$, and $|S \cap \G \ell v| \leq \sum_{i=0}^k \X_i \leq 2k|S|n^{-1+\alpha+\epsilon}$. The result follows.
\end{proof}
\begin{proof}[Proof for IRG]
By \Cref{lem:gammafromweightirg}, we can assume that $n^\epsilon \leq \d {\ell -1}v \leq n^{1-\epsilon}$. The probability that a vertex $s \in S$ is not linked to any vertex in $\G{\ell -1}v$ is $\prod_{w \in \G{\ell -1}v}\left(1-f\left(\frac{\w w\w s}{n}\right)\right)=\prod_{w \in \G{\ell -1}v}\left(1-\O\left(\frac{\w w\w s}{n}\right)\right) \leq \prod_{w \in \G {\ell -1}v}e^{-\O\left(\frac{\w w\w s}{n}\right)}=e^{-\O\left(\frac{\d {\ell -1}v\w s}{n}\right)}=e^{-\O\left(\frac{\d{\ell -1}vn^{\alpha+\epsilon}}{n}\right)}$. If $\d{\ell -1}v>n^{1-\alpha-2\epsilon}$, the result of the lemma is trivial, if we change the value of $\epsilon$. If $\d{\ell -1}v<n^{1-\alpha-2\epsilon}$, the probability that a vertex in $S$ is not linked to any vertex in $\d{\ell -1}s$ is $e^{-\O\left(\frac{\d{\ell -1}sn^{\alpha+\epsilon}}{n}\right)}=1-\O\left(\frac{\d{\ell -1}sn^{\alpha+\epsilon}}{n}\right)$, and hence the probability that it is connected to a vertex in $\d{\ell -1}s$ is $\O\left(\frac{\d{\ell -1}sn^{\alpha+\epsilon}}{n}\right)$. By a straightforward application of Chernoff bound, the number of vertices in $S$ that belong to $\G \ell s$ is $\O\left(\frac{|S|\d{\ell -1}sn^{\alpha+\epsilon}}{n}\right)$, \whp.
\end{proof}

\begin{lemma} \label{lem:cutvert}
Assume that $\beta>3$, and let $v$ a vertex with degree $\omega(1)$. Let $S$ be the set of vertices with degree between $n^\alpha$ and $n^{\alpha+\epsilon}$. Then, the number of pairs of vertices $s,t \in S$ such that $\dist(s,v)+\dist(v,t) \leq c\log_{M_1(\mu)}n$, and $\dist(s,w)+\dist(w,t) > c\log_{M_1(\mu)}n$ for each $w$ such that $\deg(w) > \deg(v)$ is at most $\deg(v)^2|S|^2n^{-2+c+2\alpha+\epsilon}$.
\end{lemma}
\begin{proof}
First of all, we want to assume without loss of generality that $v$ is the vertex with maximum degree. To this purpose, we remove from the graph all vertices with degree bigger than $\deg(v)$: by \Cref{lem:subgraph}, we obtain a new random graph $G'$, and the value of $M_1(\mu)$ changes by $\o(1)$. Furthermore, $v$ is the vertex with maximum degree in the new graph, and the shortest paths not passing from vertices with degree bigger than $\deg(v)$ are conserved.

Then, let us assume that $v$ is the vertex with maximum degree. By \cref{cor:verybigcm1,cor:verybigirg1}, $\g iv \leq n^\epsilon \deg(v)(M_1(\mu)+\epsilon)^i$, and by \cref{lem:subsetdegneigh}, $|S \cap \G iv| \leq \g iv |S|n^{-1+\alpha+\epsilon} \leq \deg(v)|S|(M_1(\mu)+\epsilon)^in^{-1+\alpha+2\epsilon}$. We conclude that the number of pairs $(s,t) \in S^2$ such that $\dist(s,v)+\dist(v,t) \leq c\log_{M_1(\mu)}n$ is at most:
\begin{align*}
&\sum_{i+j=c\log_{M_1(\mu)}n} |\{s:\dist(s,v)\leq i\}||\{t:\dist(t,v) \leq c\log_{M_1(\mu)}n-i\}| \\
&\leq \sum_{i+j=c\log_{M_1(\mu)}n} i\deg(v)|S|(M_1(\mu)+\epsilon)^in^{-1+\alpha+2\epsilon} \cdot j\deg(v)|S|(M_1(\mu)+\epsilon)^jn^{-1+\alpha+2\epsilon} \\
&\leq \sum_{i+j=c\log_{M_1(\mu)}n} \deg(v)^2|S|^2(M_1(\mu)+\epsilon)^{i+j}n^{-2+2\alpha+5\epsilon} \\
&\leq \deg(v)^2|S|^2n^{-2+c+2\alpha+\epsilon'}. \\
\end{align*} 
\end{proof}

\section{The BCM Algorithm.} \label{sec:bcm}

The \bcm\ algorithm \cite{Borassi2015b,Bergamini2015} exactly computes the $k$ most central vertices according to closeness centrality. We recall that the farness $\farn s$ of a vertex $s$ is defined as $\sum_{t \in V}\dist(s,t)$; the closeness centrality of $s$ is $\clos s=\frac{1}{\farn s}=\frac{1}{\sum_{t \in V}\dist(s,t)}$. Intuitively, a vertex with high closeness centrality needs a few step to ``talk'' to all other vertices, and consequently it is considered central. 

The textbook algorithm that computes the $k$ most central vertices simply computes the farness $\farn s$ of each vertex $s$ through $n$ BFSes, and it returns the $k$ vertices with smallest $\farn s$ values. The running time is $\O(mn)$. The improvement proposed by the \bcm\ algorithm stops the BFS from a vertex $s$ as soon as we can guarantee that $s$ is not in the top-$k$. To this purpose, assume that the $k$-th smallest farness found until now is $f_k$, and we have visited all vertices up to distance $\ell $. Then, we can lower bound the farness of $s$ by setting distance $\ell +1$ to all unvisited vertices, and if this lower bound is bigger than $f_k$, then we can safely interrupt the BFS.

A further speed-up can be obtained by computing a better bound: we set distance $\ell +1$ to a number of vertices equal to $\sum_{t \in \G \ell s} \deg(t)-1$ (which is the number of edges exiting from level $\ell $ of the BFS tree), and distance $\ell +2$ to all other vertices. In formula, we stop a visit from a vertex $s$ after $\ell $ steps if:
\[
\lfarn \ell s:=\sum_{t \in \N \ell s} \dist(s,t) + (\ell +1)\gamma^\ell _U(s) + (\ell +2)\left(n-|\N \ell s|-\gamma^\ell _U(s)\right) \geq f_k
\]
where $\N \ell s$ is the set of vertices at distance at most $\ell $ from $s$, $\G \ell s$ is the set of vertices at distance exactly $\ell $ from $s$, $\gamma^\ell _U(s)=\sum_{v \in \G \ell s} \deg(v)-1$.

It remains to choose the order in which vertices are processed. In order to speed-up the computation as much as possible, the algorithm processes vertices in decreasing order of degree, so that we obtain quite soon high values of $f_k$.

In order to analyze this algorithm, we first provide a deterministic bound on the running time. The idea behind these bounds is that a BFS from $s$ visits all vertices at distance at most $f_k-2$, then it might find a lower bound bigger than $f_k$. In particular, if the number of vertices visited at distance at most $f_k-2$ is much smaller than $n$, the bound is likely to be sufficient to stop the visit. Otherwise, the BFS from $s$ has already visited $\Theta(n)$ vertices: in both cases, the number of visited vertices is close to the number of vertices at distance at most $f_k-2$ from $s$.

\begin{lemma} \label{lem:upper}
Let $f_k$ be the $k$-th smallest farness among the $k$ vertices with highest degree, and let $\ell  \geq (1+\alpha)\frac{f_k}{n}-2$ be an integer. Then, the running time of the algorithm is at most $\O\left(m+\frac{1}{\alpha}\sum_{s \in V} \sum_{v \in \N{\ell }s} \deg(v)\right)$.
\end{lemma}
\begin{proof}
The first $k$ BFSes need time $\O(m)$, because they cannot be cut. For all subsequent BFSes, the $k$-th smallest farness found is at least $f_k$. Let us consider a BFS from $s$, and let us assume that we have visited all vertices up to distance $\ell $: our lower bound on the farness of $s$ is $\sum_{v \in \N \ell s} \dist(s,v) + (\ell +1)\gamma^\ell _U(s) + (\ell +2)\left(n-|\N \ell s|-\gamma^\ell _U(s)\right) \geq (\ell +2)\left(n-|\N \ell s|-\gamma^\ell _U(s)\right)\geq (\ell +2)\left(n-\sum_{v \in \N{\ell }s} \deg(v)\right)\geq \left(1+\alpha\right)\frac{f_k}{n}\left(n-\sum_{v \in \N{\ell }s} \deg(v)\right)$. We claim that the BFS from $s$ visits at most $\frac{m}{n\alpha} \sum_{v \in \N{\ell }s} \deg(v)$ edges: this is trivially true if $\sum_{v \in \N{\ell }s} \deg(v)>\frac{\alpha}{2} n$, because the value becomes $\O(m)$, while if $\sum_{v \in \N{\ell }s} \deg(v)<\frac{\alpha}{2} n$, the lower bound is at least $\left(1+\alpha\right)\frac{f_k}{n}\left(1-\frac{\alpha}{2}\right)n\geq f_k$, and the BFS is stopped after $\ell $ steps.
\end{proof}

\begin{lemma} \label{lem:lower}
Let $f_k$ be the $k$-th smallest farness, and let $\ell  \leq \frac{f_k}{n}-2$. Then, the running time of the algorithm is $\Omega\left(\sum_{s \in V} \n \ell s\right)$.
\end{lemma}
\begin{proof}
Clearly, at any step, the $k$-th minimum farness found is at least $f_k$. We want to prove that the BFS from $s$ reaches level $\ell $: indeed, the lower bound on the farness of $s$ is $\sum_{v \in \N \ell s} \dist(s,v) + (\ell +1)\gamma^\ell _U(v) + (\ell +2)\left(n-|\N \ell s|-\gamma^\ell _U(s)\right) \leq n(\ell +2) \leq f_k$. Hence, the BFS is not cut until level $\ell $, and at least $\Omega\left(\sum_{s \in V} \n \ell s\right)$ vertices are visited.
\end{proof}

We now need to compute these values in graphs in our framework. We analyze separately the running time in the case $1<\beta<2$, in the case $2<\beta<3$, and in the case $\beta>3$.

\subsection{The Case $1<\beta<2$.}

By \Cref{ax:deg}, if $s$ is one of the $k$ vertices with maximum degree, then $\deg(s)>2cn$ for some $c$ only depending on $k$, and by \Cref{thm:farnessupper}, the farness of $s$ is at most $(2+\o(1))n-\deg(s)\leq (2-c)n$ if $n$ is big enough. By \Cref{lem:upper} applied with $\ell =0$ and $\alpha=c$, the running time is at most $\O\left(m+\frac{1}{c}\sum_{s \in V} \sum_{t \in \N{0}s} \deg(t)\right)=\O\left(m+\sum_{s \in V} \deg(s)\right)=\O(m)$.

\subsection{The Case $2<\beta<3$.}

In this case, we know by \Cref{thm:farnesslower} applied with $x=\frac{1}{2}$ that the minimum farness is at least $n(1+\o(1))\left(\timp s{n^x}-\Td 1{n^x}+\avedist n11-1\right)=n\left(\frac{1}{2}+\o(1)\right)\avedist n11=(1+\o(1))\log_{\frac{1}{\beta-2}} \log n=\Theta(\log\log n)$.

We claim that all vertices with degree at least $n^x$ are at distance $\O(1)$ from each other: indeed, if $\deg(s),\deg(t)>n^x$, by \Cref{ax:dev}, $\timp s{n^{\frac{2}{3}}},\timp t{n^{\frac{2}{3}}}=\O(1)$, and by \Cref{ax:touch}, $\dist(s,t) < \timp s{n^{\frac{2}{3}}}+\timp t{n^{\frac{2}{3}}} = \O(1)$.

By \Cref{lem:lower}, if $\ell +2$ is smaller than $f_k$, the running time is $\Omega\left(\sum_{s \in V} \n \ell v\right) = \Omega\left(\sum_{\deg(s)>n^x} \n {\Theta(\log \log n)}s\right) = \Omega\left(\sum_{\deg(s)>n^x} |\{t:\deg(t)>n^x\}|\right) \geq \Omega\left(\left(n^{1-x(\beta-1)}\right)^2\right)=\Omega\left(n^{2-2x(\beta-1)}\right)$. If we choose $x=\epsilon$, the running time is $\Omega\left(n^{2-\O(\epsilon)}\right)$.

\subsection{The Case $\beta>3$.}

Let us estimate the farness of the $k$ vertices with highest degree. By \Cref{ax:deg}, the $k$ maximum degrees are $\Theta\left(n^{\frac{1}{\beta-1}}\right)$, and their farness is $n(1+\O(\epsilon))\left(\timp s{n^{\frac{1}{\beta-1}}}-\Td 1{n^{\frac{1}{\beta-1}}}+\avedist n11\right) \leq n(1+\O(\epsilon)) \log_{M_1(\mu)}\left( n^{1-\frac{1}{\beta-1}}\right)$. Hence, by \Cref{lem:lower} applied with $\ell =(1+\O(\epsilon))\log_{M_1(\mu)} \left(n^{1-\frac{1}{\beta-1}}\right)$, we obtain that the running time is $\O\left(\sum_{s \in V} \sum_{v \in \N{\ell }s} \deg(v)\right)=\O\left(\sum_{s \in V} \sum_{v \in \N{\ell }s} \w v\right)=\O\left(n^\epsilon\sum_{s \in V} \n{\ell+1}v\right)$, because $\deg(v)\leq n^\epsilon\w v$, and $\g{\ell+1}v > \w{\G \ell v}$ in random graphs, as shown in \Cref{lem:gammafromweightcm,lem:gammafromweightirg}. For the lower bound, we use two partial results that we obtain in the proof of the main theorems: \Cref{cor:verybigcm,cor:verybigirg}. These corollaries say that $\timp s{n^y}-\timp s{n^x} \geq (1-\epsilon)(y-x)\log_{M_1(\mu)}n$, for each $\frac{1}{\beta-1}<x<y<1$. Hence, for each vertex $s$, $\farn s \geq \left(n-n^{1-\epsilon}\right)\timp s{n^{1-\epsilon}} \geq \left(1-\frac{1}{\beta-1}-\O(\epsilon)\right)\log_{M_1(\mu)}n \geq (1-\O(\epsilon))\log_{M_1(\mu)}\left(n^{1-\frac{1}{\beta-1}}\right)$.

We conclude that, by \Cref{lem:lower,lem:upper}, the running time of the algorithm is $\sum_{s \in V} \n{\ell }s$, where $\ell =(1\pm\O(\epsilon))\log_{M_1(\mu)}\left(n^{1-\frac{1}{\beta-1}}\right)$. To estimate this value, we use the following result, that we prove in \Cref{sec:proof} (\Cref{thm:bigneighbors}).

\begin{theorem}
For each $0<x<y<1$, $\timp{s}{n^y}-\timp{s}{n^x} \geq (1-\epsilon)\log_{M_1(\mu)}n^{y-x}$ \aas. Moreover, $\timp{s}{n^y}-\timp{s}{n^x} \leq (1+\epsilon)\log_{M_1(\mu)}n^{y-x}$ \whp.
\end{theorem}

Thanks to this result, for each vertex $s$, $\timp s{n^y} \leq \timp s{n^\epsilon\deg(s)}+(1+\epsilon)\log_{M_1(\mu)}n^{y-x} \leq (1+\O(\epsilon))\log_{M_1(\mu)}\frac{n^{y}}{\deg(s)} \leq \ell$ if $y=1-\frac{1}{\beta-1}+\frac{\log \deg(s)}{\log n}+\O(\epsilon)$, and consequently $\g is<n^{y}$ for each $i<\ell $. Hence, the running time is smaller than $n^{\O(\epsilon)}\sum_{s \in V} \n \ell s \leq n^{\O(\epsilon)}\sum_{s \in V} \ell n^y \leq n^{1-\frac{1}{\beta-1}+\O(\epsilon)} \sum_{s \in V}\deg(s) \leq n^{2-\frac{1}{\beta-1}+\O(\epsilon)}$. An analogous argument proves that the running time is at least $n^{2-\frac{1}{\beta-1}-\O(\epsilon)}$.

\section{Distance Oracle.} \label{sec:oracle}

In this section, we analyze the performances of the distance oracle in \cite{Akiba2013}. Basically, this distance oracle assigns a label $L(s)$ to each vertex $s$: each label is a set of pairs $(v,\dist(s,v))$, where $v$ is a vertex in the graph, and $\dist(s,v)$ is the distance between $s$ and $v$. The construction of these labels enforces the so-called 2-hop cover property: for each pair of vertices $s,t$, there is a vertex $v$ in a shortest path between $s$ and $t$ that belongs both to $L(s)$ and to $L(t)$. Using the 2-hop cover property, it is possible to compute $\dist(s,t)=\min_{v \in L(s) \cap L(t)} \dist(s,v)+\dist(v,t)$, in time $\O(|L(s)|+|L(t)|)$. The space required is $\Theta\left(\sum_{s \in V} |L(s)|\right)$.

In order to compute a set of labels that satisfies the 2-hop cover property, we sort all vertices $s$ obtaining $s_1,\dots,s_n$ (any order is fine), and we perform a BFS from each vertex, following this order. During the BFS from $s_i$, as soon as we visit a vertex $t$, we add $(s_i,\dist(s_i,t))$ to the label of $t$. Furthermore, we prune each BFS at each vertex $t$ such that $\dist(s,t)=\min_{v \in L(s) \cap L(t)} \dist(s,v)+\dist(v,t)$, where $L(s)$ and $L(t)$ are the current labels. It is proved in \cite{Akiba2013} that the labels generated by the algorithm satisfy the 2-hop cover property, and that $s_i$ is in the label of $t$ if and only if there is no vertex $s_j$ for some $j<i$ that belongs to an $(s,t)$-shortest path. It remains to define how vertices are sorted: in \cite{Akiba2013}, it is suggested to sort them in order of degree (tie-breaks are solved arbitrarily).

First, as we did in the previous analyses, we compute a deterministic bound on the expected time of a distance query between two random vertices.

\begin{lemma} \label{lem:oracle}
For each $s_i$, let $N_{s_i}$ be the number of vertices $t \in V$ such that no $(s,t)$-shortest path passes from a vertex $s_j$ with $j<i$. Then, the average query time is $\O\left(\frac{1}{n}\sum_{s \in V} N(s)\right)$, and the space used is $\Theta\left(\sum_{s \in V} N(s)\right)$.
\end{lemma}
\begin{proof}
If the labels are sorted, in order to intersect $L(s)$ and $L(t)$, we need time $\O\left(|L(s)|+|L(t)|\right)$. Hence, the expected time of a distance query between two random vertices is $\frac{1}{n^2}\sum_{s,t \in V}\O\left(|L(s)|+|L(t)|\right)=\O\left(\frac{1}{n}\sum_{t \in V} |L(t)|\right)=\O\left(\frac{1}{n}\sum_{s,t \in V} X_{st}\right)=\O\left(\frac{1}{n}\sum_{s \in V} N(s)\right)$, where $X_{st}=1$ if $s \in L(t)$, $0$ otherwise. Similarly, the space used is $\Theta\left(\frac{1}{n}\sum_{t \in V} |L(t)|\right)=\Theta\left(\sum_{s \in V} N(s)\right)$.
\end{proof}

\subsection{The Case $1<\beta<2$.}

Let us fix $\epsilon>0$, and let us consider vertices $s$ with small degree (at most, $n^{2\epsilon}$): the number of vertices reachable from $s$ at distance $k$ passing only through vertices of degree smaller than $\deg(s)$ is at most $\deg(s)^k$: hence, $N(s)\leq D \deg(s)^D \leq Dn^{2\epsilon D} = n^{\O(\epsilon)}$ (because the diameter $D$ is constant).

Let us consider vertices $s$ such that $\deg(s)>n^{2\epsilon}$: by \Cref{ax:touch}, all these vertices are connected to each vertex with degree at least $n^{1-\epsilon}$, and no vertex $t$ with degree bigger than $n^{\epsilon}$ can contain $s$ in their label, unless $t$ is a neighbor of $s$. Consequently, the vertices that contain $s$ in their label are at most $\deg(s)$ vertices at distance $1$ from $s$, $N_2(s)$ vertices at distance $2$ from $s$, and at most $\left(n^{2\epsilon}\right)^D N_2(s) = n^{\O(\epsilon)} N_2(s)$ vertices at a bigger distance. Summing these values, we obtain that $N(s) \leq \deg(s)+N_2(s)n^{\O(\epsilon)}$: summing over all vertices $s$, the average query time is $\frac{1}{n}\O\left(n^{1+\epsilon}+\sum_{s \in V,\deg(s)>n^{2\epsilon}} \deg(s)+N_2(s)n^{2\epsilon}\right) =n^{\O(\epsilon)}\left(1+\frac{1}{n}\sum_{s \in V}N_2(s)\right)$. Since $t \in N_2(s)$ implies that $\deg(t)<n^{2\epsilon}$, and since the number of vertices with degree bigger than $n^x$ is $n^{1-x+\o(1)}$ by \Cref{ax:deg},
\begin{align*}
\sum_{s \in V}N_2(s)&=\sum_{s \in V}\sum_{v \in \G 1s,\deg(v)<\deg(s)}\sum_{t \in \G 1v,\deg(t)<n^{2\epsilon}} 1 \\
&=\sum_{v \in V}\left|\left\{t \in \G 1v:\deg(t)<n^{2\epsilon}\right\}\right|\left|\left\{s \in \G 1v:\deg(s)>\deg(v)\right\}\right| \\
&\leq\sum_{v \in V} \deg(v)\max\left(\deg(v),\frac{n^{1+\epsilon}}{\deg(v)}\right) \\
&=\sum_{d=1}^n \left|\left\{v:\deg(v)=d\right\}\right|d\max\left(d,\frac{n^{1+\epsilon}}{d}\right) \\
&\leq\sum_{d=1}^{n^{\frac{1}{2}}} \left|\left\{v:\deg(v)=d\right\}\right|d^2 + \sum_{d=n^{\frac{1}{2}}}^n \left|\left\{v:\deg(v)=d\right\}\right|n^{1+\epsilon}\\
&\leq S_1+n^{\frac{3}{2}+\epsilon}.
\end{align*}
Let us estimate $S_1$ using Abel's summation technique:

\begin{align*}
S_1&=\sum_{d=1}^{n^{\frac{1}{2}}} \left|\left\{v:\deg(v)= d\right\}\right| d^2 \\
&= \sum_{d=1}^{n^{\frac{1}{2}}} \left|\left\{v:\deg(v)\geq d\right\}\right|d^2-\sum_{d=1}^{n^{\frac{1}{2}}}\left|\left\{v:\deg(v)\geq d+1\right\}\right|d^2 \\
&\leq n+\sum_{d=1}^{n^{\frac{1}{2}}} \left|\left\{v:\deg(v)\geq d\right\}\right|(d^2-(d-1)^2) \\
&\leq n+\sum_{d=1}^{n^{\frac{1}{2}}} \frac{n}{d}2d \\
&= \O\left(n^{\frac{3}{2}}\right).
\end{align*}
Then, the average query time is $n^{\O(\epsilon)}\left(1+\frac{1}{n}\sum_{s \in V}N_2(s)\right)=n^{\frac{1}{2}+\O(\epsilon)}$. The space occupied is $n$ multiplied by the average query time, that is, $n^{\frac{3}{2}+\O(\epsilon)}$, by \Cref{lem:oracle}.

\subsection{The Case $2<\beta<3$.}


Let us consider the number $N_\ell (s)$ of vertices $t$ such that $s \in L(t)$ and $\dist(s,t)=\ell $. For small values of $\ell $, we estimate $N_\ell (s) \leq N_{\ell -1}(s) \deg(s)$, while for bigger values of $\ell $, we prove that all vertices with high degree are at distance at most $\ell -2$ from $s$, obtaining that $N_\ell (s) \leq f(\ell ,\deg(s))N_{k-1}(s)$, for some function $f$.

More formally, let $s,t$ be two vertices with degree at least $\log^2n$: the distance between $s$ and $t$ is at most $\timp s{n^{\frac{1+\epsilon}{2}}}+\timp t{n^{\frac{1+\epsilon}{2}}}\leq (1+\epsilon)\left(\log_{\frac{1}{\beta-2}}\left(\frac{\log^2 n^{\frac{1+\epsilon}{2}}}{\log \deg(s)\log \deg(t)}\right)\right)$ (our properties say that this approximation holds for $\deg(s),\deg(t)>n^\epsilon$, but it is easy to extend the results in \Cref{sec:big} to the case $\deg(s)>\log^2 n$). Consequently there is no vertex at distance $k$ from $s$ with degree bigger than $\max\left(\log^2n, e^{(\beta-2)^{\frac{k}{1+\epsilon}} \frac{\log^2{n^{\frac{1+\epsilon}{2}}}}{\log(\deg(s))}}\right)$.

As we said before, for $k<k_0=(1+\epsilon)\left(\log_{\frac{1}{\beta-2}}\left(\frac{\log^2 n^{\frac{1+\epsilon}{2}}}{\log^2 \deg(s)}\right)\right)$, we estimate $N_{k+1}(s) \leq N_k(s)\deg(s)$, and consequently $N_k(s) \leq \deg(s)^k \leq \deg(s)^{(1+\epsilon)\left(\log_{\frac{1}{\beta-2}}\left(\frac{\log^2 n^{\frac{1+\epsilon}{2}}}{\log^2 \deg(s)}\right)\right)}$. For bigger values of $k$, since $\G {k+1}v$ does not contain any vertex with degree bigger than $\max\left(\log^2n,e^{(\beta-2)^{\frac{k}{1+\epsilon}} \frac{\log^2{n^{\frac{1+\epsilon}{2}}}}{\log(\deg(s))}}+\log^2n\right)$, $N_{k+1}(S) \leq N_{k}(S) \max\left(\log^2n, e^{(\beta-2)^{\frac{k-1}{1+\epsilon}} \frac{\log^2{n^{\frac{1+\epsilon}{2}}}}{\log(\deg(s))}}\right)$. As a consequence, we can prove by induction that, if $k=\O(\log\log n)$: 

\begin{align*}
N_k(s) &\leq N_{k_0}(s) \prod_{i = k_0}^k e^{(\beta-2)^{\frac{i-1}{1+\epsilon}} \frac{\log^2{n^{\frac{1+\epsilon}{2}}}}{\log^2 \deg(s)}}\log^{2k} n\\
&\leq N_{k_0}(s)n^\epsilon \prod_{i = k_0}^k e^{(\beta-2)^{\frac{i-k_0}{1+\epsilon}} (\beta-2)^{\frac{k_0-1}{1+\epsilon}} \frac{\log^2{n^{\frac{1+\epsilon}{2}}}}{\log(\deg(s))}} \\
&\leq N_{k_0}(s)n^\epsilon e^{\sum_{i = k_0}^k (\beta-2)^{\frac{i-k_0}{1+\epsilon}} \log \deg(s)} \\
&\leq N_{k_0}(s)n^\epsilon \deg(s)^{\frac{1}{1-(\beta-2)^{\frac{1}{1+\epsilon}}}}. \\
&\leq n^{2\epsilon} \deg(s)^{(2+2\epsilon)\left(\log_{\frac{1}{\beta-2}}\left(\frac{\log n^{\frac{1+\epsilon}{2}}}{\log \deg(s)}\right)\right)+\frac{1}{3-\beta}}=f(\deg(s)). \\
\end{align*}

For bigger values of $k$, there are very few vertices at distance $k$ from $s$, and their contribution is negligible. Hence, we know that $N(s) \leq \min(n,n^\epsilon f(\deg(s)))=g(\deg(s))$: we want to compute $\sum_{s \in V} g(\deg(s))$.

\begin{lemma}
If $G$ is a random graph with power law degree distribution, 
\[
\sum_{s \in V,\deg(s)>\log^2n} g(d)=\O\left(n\int_{\log^2n}^{n^{\frac{1}{\beta-1}+\epsilon}}\frac{g(x)}{x^\beta}dx+n^{1+\epsilon}\right).
\]
\end{lemma}
\begin{proof}
We use Abel's summation technique twice:
\begin{align*}
\sum_{s \in V,\deg(s)>\log^2n} g(d)&= \sum_{d=\log^2n}^{n^{\frac{1+\epsilon}{\beta-1}}} \left|\left\{s \in V: \deg(s)=d\right\}\right|g(d) \\
&= \sum_{d=\log^2n}^{n^{\frac{1+\epsilon}{\beta-1}}} \left(\left|\left\{s \in V: \deg(s) \geq d\right\}\right|-\left|\left\{s \in V: \deg(s) \geq d+1\right\}\right|\right)g(d) \\
&\leq \sum_{d=\log^2n}^{n^{\frac{1+\epsilon}{\beta-1}}} \left(\left|\left\{s \in V: \deg(s) \geq d\right\}\right|-\left|\left\{s \in V: \deg(s) \geq d+1\right\}\right|\right)g(d) \\
&\leq \sum_{d=\log^2n}^{n^{\frac{1+\epsilon}{\beta-1}}} \frac{n}{d^{\beta-1}}(g(d)-g(d-1)) +\O(ng(\log^2n))\\
&\leq \O(ng(\log^2n))+\sum_{d=\log^2n}^{n^{\frac{1+\epsilon}{\beta-1}}} ng(d)\left(\frac{1}{d^{\beta-1}}-\frac{1}{(d+1)^{\beta-1}}\right) \\
&\leq \O(ng(\log^2n))+\sum_{d=\log^2n}^{n^{\frac{1+\epsilon}{\beta-1}}} ng(d)\left(\frac{1}{d^{\beta-1}}-\frac{1}{(d+1)^{\beta-1}}\right) \\
&=\O\left(n^{1+\epsilon}+\sum_{d=\log^2n}^{n^{\frac{1+\epsilon}{\beta-1}}} \frac{ng(d)}{d^{\beta}}\right).
\end{align*}We have to transform this sum into an integral: to this purpose, we observe that $g(d+\epsilon)=\O(g(d))$ for each $d>\log^2n-1$, and for each $\epsilon \leq 1$. Hence, $\sum_{d=\log^2n}^{n^{\frac{1+\epsilon}{\beta-1}}} \frac{g(d)}{d^{\beta}}=\int_{\log^2n-1}^{n^{\frac{1}{\beta-1}+\epsilon}}\frac{g(\left\lceil x\right\rceil)}{\left\lceil x\right\rceil^\beta}dx=\O\left(\int_{\log^2n-1}^{n^{\frac{1}{\beta-1}+\epsilon}}\frac{g(x)}{x^\beta}dx\right)$.
\end{proof}
It remains to estimate this integral.

\begin{align*}
& n\int_{\log^2 n}^{n^{\frac{1}{\beta-1}}} \frac{1}{x^\beta} \min\left(n,x^{(2+2\epsilon)\log_{\frac{1}{\beta-2}}\left(\frac{\log n^{\frac{1+\epsilon}{2}}}{\log x}\right)+\frac{1}{3-\beta}}\right) dx = \left(\frac{\log x}{\log n} = t\right)\\
&\leq n^{1+\O(\epsilon)}\int_{0}^{\frac{1}{\beta-1}} n^{-\beta t}\min\left(n,n^{t\left(2\log_{\frac{1}{\beta-2}}\left(\frac{1}{2t}\right)+\frac{1}{3-\beta}\right)}\right)n^t\log n \, dt \\
&= n^{1+\O(\epsilon)}\int_{0}^{\frac{1}{\beta-1}} n^{t\left(\min\left(1,2\log_{\frac{1}{\beta-2}}\left(\frac{1}{2t}\right)+\frac{1}{3-\beta}\right)-\beta+1\right)}dt \\
&= n^{1+\O(\epsilon)+\max{t \in \left[0,\frac{1}{\beta-1}\right]} t\left(\min\left(1,2\log_{\frac{1}{\beta-2}}\left(\frac{1}{2t}\right)+\frac{1}{3-\beta}\right)-\beta+1\right)}.
\end{align*}
Then, the average query time is $n^{\max_{t \in \left[0,\frac{1}{\beta-1}\right]} t\left(\min\left(1,2\log_{\frac{1}{\beta-2}}\left(\frac{1}{2t}\right)+\frac{1}{3-\beta}\right)-\beta+1\right)+\O(\epsilon)}$. We are not able to find an analytic form for this function, but the result is plotted in \Cref{fig:plot}. The exponent of the total space occupancy is this function, plus one.

\begin{figure}[htb]
\input{RunningTimeOracle}
\caption{an upper bound on the exponent in the average distance query time of the distance oracle for $2<\beta<3$.} \label{fig:plot}
\end{figure}
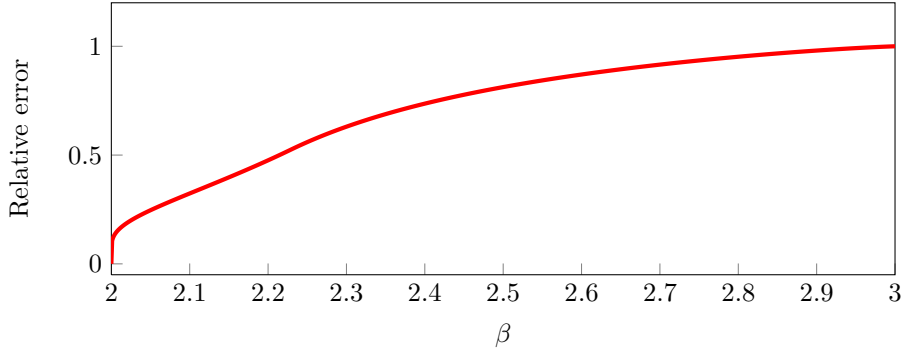

\subsection{The Case $\beta>3$.}

In this case, we prove that the algorithm does not provide a significant improvement: indeed, the time needed for a distance query is $n^{1-\O(\epsilon)}$. To prove this result, it is enough to show that, for a set of graphs that satisfy the four properties, the algorithm is not efficient on this set of graph. The set of graphs we choose is the set of random graphs generated through the CM, or through IRG: for this reason, we are allowed to use theorems that are specific of these models.

In this proof, we show that, if $S$ is the set of vertices with degree between $n^\alpha$ and $n^{\alpha+\epsilon^2}$, then the number of pairs $(s,t) \in S^2$ such that $s \in L(t)$ is $\Omega\left(|S|^2\right)$, and hence the average label size is big, because, by \Cref{ax:deg}, $|S|=n^{1-\alpha(\beta-1)+\o(1)}=n^{1-\o(1)}$ if $\alpha$ tends to $0$.

More formally, by \Cref{ax:touch,ax:dev}, for each pair of vertices $s,t \in S$, $\dist(s,t) \leq \timp s{n^{\frac{1}{2}+\epsilon}}+\timp t{n^{\frac{1}{2}+\epsilon}} \leq (2+2\epsilon)\Td{n^\alpha}{n^{\frac{1}{2}+\epsilon}} \leq (2+2\epsilon)\frac{\log n^{\frac{1}{2}+\epsilon-\alpha}}{\log M_1(\mu)}\leq (1-2\alpha+4\epsilon)\frac{\log n}{\log M_1(\mu)}=:D_S$. We want to prove that, \aas, if $s,t \in S$ and $\deg(t)<\deg(s)$, there is a high chance that $s \in L(t)$. To this purpose, we consider all vertices $v$ with degree bigger than $\deg(s)$, and we count the number of pairs $s,t\in S$ such that $\dist(s,v)+\dist(v,t) \leq D_S$. Then, we sum this contribution over all vertices $v$: if this sum is $\o\left(|S|\right)$, it means that $\Omega(|S|)$ vertices in $S$ have $s$ in their label.

More formally, we start by estimating, for each vertex $v$, the number of vertices $s,t\in S$ such that $\dist(s,v)+\dist(v,t) \leq D_S$.

\begin{lemma}[for a proof, see \Cref{lem:cutvert}]
Let $v$ be a vertex with degree $\omega(1)$. Then, the number of pairs of vertices $s,t \in S$ such that $\dist(s,v)+\dist(v,t) \leq D_S$, and $\dist(s,w)+\dist(w,t) > D_S$ for each $w$ such that $\deg(w) > \deg(v)$, is at most $\deg(v)^2|S|^2n^{-1+\O(\epsilon)}$.
\end{lemma}

Let us consider the ordering of all vertices $s_1,\dots,s_n$, and let us estimate:
\begin{align*}
&|\{(s_i,s_j) \in S^2: i<j,s_i \notin L(s_j)\}| \\
&\leq
|\{(s_i,s_j) \in S^2: i<j,\exists k<i, \dist(s_i,s_j)=\dist(s_i,s_k)+\dist(s_k,s_j)\}| \\
& \leq |\{(s_i,s_j) \in S^2: i<j,\exists k<i, \dist(s_i,s_k)+\dist(s_k,s_j) \leq D_S\}| \\
&\leq \sum_{k < i} \deg(s_k)^2|S|^2n^{-1+\O(\epsilon)} \\
&\leq n^{1-\alpha(\beta-3)+\epsilon}|S|^2n^{-1+\O(\epsilon)} \\
&=\o(|S|^2).
\end{align*}
We used the fact that $\sum_{k < i} \deg(s_k)^2 \leq n^{1-\alpha(\beta-3)+\epsilon}$: let us prove it formally, using Abel's summation technique and \Cref{ax:deg}.

\begin{align*}
\sum_{k < i} \deg(s_k)^2 &= \sum_{d=n^\alpha}^{+\infty} d^2|\{v:\deg(v)=n^\alpha\}| \\
& =\sum_{d=n^\alpha}^{+\infty} d^2|\{v:\deg(v)\geq d\}|-\sum_{d=n^\alpha+1}^{+\infty} (d-1)^2|\{v:\deg(v) \geq d\}| \\
&\leq n^{2\alpha}|\{v:\deg(v) \geq n^\alpha\}|+\sum_{d=n^\alpha}^{+\infty} 2d|\{v:\deg(v)\geq d\}| \\
&\leq n^{2\alpha}\frac{n}{n^{\alpha(\beta-1)}}+\sum_{d=n^\alpha}^{+\infty} 2d\frac{n}{d^{\beta-1}} = \O\left(n^{1-\alpha(\beta-3)}\right).
\end{align*}
We have proved that $|\{(s_i,s_j) \in S^2: i<j,s_i \in L(s_j)\}| =\Omega\left(|S|^2\right)$, and consequently the total label size is at least $\Omega\left(|S|^2\right) \geq n^{2-3\alpha(\beta-1)}$. We claim that this means that there are many labels with size bigger than $n^{1-4\alpha(\beta-1)}$, because no label has size bigger than $n$. Indeed, if $\ell _i$ is the size of label $i$, $n^{2-3\alpha(\beta-1)} \leq \sum_{i=1}^n \ell _i \leq n^{1-4\alpha(\beta-1)}|\{i:\ell _i\leq n^{1-4\alpha(\beta-1)}\}|+n|\{i:\ell _i > n^{1-4\alpha(\beta-1)}\}| \leq n^{2-4\alpha(\beta-1)}+n|\{i:\ell _i > n^{1-4\alpha(\beta-1)}\}|$, and hence $|\{i:\ell _i > n^{1-4\alpha(\beta-1)}\}| \geq n^{1-3\alpha(\beta-1)}-n^{1-4\alpha(\beta-1)} \geq n^{1-4\alpha(\beta-1)}$.

We have proved that, for each $\alpha'=4\alpha(\beta-1)$, there are at least $n^{1-\alpha'}$ labels of size $n^{1-\alpha'}$: consequently, the expected time to perform a distance query is at least $\frac{n^{2-2\alpha'}}{n^2}n^{1-\alpha'}$, because the probability that we hit two vertices $s,t$ whose labels are bigger than $n^{1-\alpha'}$ is at least $\frac{n^{2-2\alpha'}}{n^2}$.  If we let $\alpha'=\O(\epsilon)$, the average time for a distance query becomes at least $n^{1-\O(\epsilon)}$. Similarly, the space occupied is $n^{2-\O(\epsilon)}$.

\end{document}

%% file: EfficiencyPlot.tex
\pgfplotstableset{
  col sep=tab,
}

\newcommand{\plotfile}[1]{
    \pgfplotstableread{#1}{\table}
    \pgfplotstablegetcolsof{#1}
    \pgfmathtruncatemacro\numberofcols{\pgfplotsretval-1}
    \pgfplotsinvokeforeach{0,...,\numberofcols}{
        \pgfplotstablegetcolumnnamebyindex{##1}\of{\table}\to{\colname}
        \addplot table [y index=##1] {#1}; 
    }
}
\newcommand{\plotfilelegend}[1]{
    \pgfplotstableread{#1}{\table}
    \pgfplotstablegetcolsof{#1}
    \pgfmathtruncatemacro\numberofcols{\pgfplotsretval-1}
    \pgfplotsinvokeforeach{1,...,\numberofcols}{
        \pgfplotstablegetcolumnnamebyindex{##1}\of{\table}\to{\colname}
        \addplot table [y index=##1] {#1}; 
        \addlegendentryexpanded{\textsc{\colname}}
    }
}
\pgfplotscreateplotcyclelist{mycolorlist}{%
{red},{black,densely dotted, thick},{blue,dashed},{green!50!black,dashdotted,thick}}
\centering
\begin{tikzpicture}
\begin{groupplot}[
    group style={
        group size=3 by 1,
        xlabels at=edge bottom,
        xticklabels at=edge bottom,
        ylabels at=edge left,
        yticklabels at=edge left,
        vertical sep=0pt,
        horizontal sep=0pt,
    },
    title style={at={(0.5,1)}, yshift=-.6\baselineskip, anchor=north, font=\footnotesize},
    cycle list name=mycolorlist,
    width=6.1cm,
    height=5.4cm,
    xlabel={$\beta$},
    ymin=-0.1, ymax=2.4,
    xmin=1, xmax=5.3,
    tickpos=left,
    ytick align=inside,
    xtick align=inside,
    legend,
    legend pos=north east,
    legend style={at={(1,1)},anchor=north east,font=\tiny}
]
\nextgroupplot[ymin=-0.05,ymax=1.2,
title={Diameter ($\epsilon_{\text{rel}}$)},
title style={at={(0.3,1)}}, ylabel={Relative error}]
\plotfilelegend{DiameterRelativeErrors.dat}
\nextgroupplot[legend style={at={(1,0)},anchor=south east},
title={Diameter (running time)}]
\plotfilelegend{DiameterExact.dat}
\nextgroupplot[title={Other alg.~(time, space)},legend style={at={(1,0)},anchor=south east},yticklabel pos=right,ytick align=inside,ytick pos=right,yticklabels={0,0,1,2},ylabel={Exponent}
]
\plotfilelegend{OtherAlgorithms.dat}
\end{groupplot}

\end{tikzpicture}

\vspace{-\baselineskip}

%% file: Tail.tex
\pgfplotstableset{
  col sep=tab,
}

\newcommand{\plotfile}[1]{
    \pgfplotstableread{#1}{\table}
    \pgfplotstablegetcolsof{#1}
    \pgfmathtruncatemacro\numberofcols{\pgfplotsretval-1}
    \pgfplotsinvokeforeach{1,...,\numberofcols}{
        \pgfplotstablegetcolumnnamebyindex{##1}\of{\table}\to{\colname}
        \addplot table [y index=##1] {#1}; 
        \addlegendentryexpanded{\colname}
    }
}
\centering
\begin{tikzpicture}
\begin{axis}[
    title style={at={(0,1)}, yshift=-.6\baselineskip, anchor=north west},
    cycle list name=color list,
    width=10cm,
    height=7cm,
	ymode=log,
    xlabel={$k$},
    ylabel=Fraction of nodes,
    ymax=1,
    xmin=-2, xmax=11,
    tickpos=left,
    ytick align=outside,
    xtick align=outside,
    legend,
    legend pos=north east,
    legend style={at={(1,1)},anchor=north west,font=\tiny}
]
\plotfile{Tail.dat}
\end{axis}

\end{tikzpicture}

%% file: Touch.tex
\pgfplotstableset{
  col sep=tab,
}

\newcommand{\plotfile}[1]{
    \pgfplotstableread{#1}{\table}
    \pgfplotstablegetcolsof{#1}
    \pgfmathtruncatemacro\numberofcols{\pgfplotsretval-1}
    \pgfplotsinvokeforeach{1,...,\numberofcols}{
        \pgfplotstablegetcolumnnamebyindex{##1}\of{\table}\to{\colname}
        \addplot table [y index=##1] {#1}; 
    }
}
\newcommand{\plotfilelegend}[1]{
    \pgfplotstableread{#1}{\table}
    \pgfplotstablegetcolsof{#1}
    \pgfmathtruncatemacro\numberofcols{\pgfplotsretval-1}
    \pgfplotsinvokeforeach{1,...,\numberofcols}{
        \pgfplotstablegetcolumnnamebyindex{##1}\of{\table}\to{\colname}
        \addplot table [y index=##1] {#1}; 
        \addlegendentryexpanded{\colname}
    }
}
\centering
\begin{tikzpicture}
\begin{groupplot}[
    group style={
        group size=2 by 2,
        xlabels at=edge bottom,
        xticklabels at=edge bottom,
        ylabels at=edge left,
        yticklabels at=edge left,
        vertical sep=0pt,
        horizontal sep=0pt,
    },
    title style={at={(0,1)}, yshift=-.6\baselineskip, anchor=north west},
    cycle list name=color list,
    width=7.3cm,
    height=5.5cm,
    xlabel={$\timp v{n^x}+\timp w{n^y}-d(v,w)$},
    ylabel=Fraction of nodes,
    ymin=0, ymax=0.799,
    tickpos=left,
    ytick align=outside,
    xtick align=outside,
    legend,
    legend pos=north east,
    legend style={at={(1,1)},anchor=north west,font=\tiny}
]
\nextgroupplot[title={$x=0.6, y=0.6$}]
\plotfile{Touch0.6-0.6.dat}
\nextgroupplot[title={$x=0.5, y=0.7$}]
\plotfilelegend{Touch0.5-0.7.dat}
\nextgroupplot[title={$x=0.4, y=0.8$}]
\plotfile{Touch0.4-0.8.dat}
\nextgroupplot[title={$x=0.3, y=0.9$}]
\plotfile{Touch0.3-0.9.dat}
\end{groupplot}

\end{tikzpicture}

%% file: Untouch.tex
\pgfplotstableset{
  col sep=tab,
}

\newcommand{\plotfile}[1]{
    \pgfplotstableread{#1}{\table}
    \pgfplotstablegetcolsof{#1}
    \pgfmathtruncatemacro\numberofcols{\pgfplotsretval-1}
    \pgfplotsinvokeforeach{0,...,\numberofcols}{
        \pgfplotstablegetcolumnnamebyindex{##1}\of{\table}\to{\colname}
        \addplot table [y index=##1] {#1}; 
    }
}
\newcommand{\plotfilelegend}[1]{
    \pgfplotstableread{#1}{\table}
    \pgfplotstablegetcolsof{#1}
    \pgfmathtruncatemacro\numberofcols{\pgfplotsretval-1}
    \pgfplotsinvokeforeach{0,...,\numberofcols}{
        \pgfplotstablegetcolumnnamebyindex{##1}\of{\table}\to{\colname}
        \addplot table [y index=##1] {#1}; 
        \addlegendentryexpanded{\colname}
    }
}
\centering
\begin{tikzpicture}
\begin{groupplot}[
    group style={
        group size=2 by 2,
        xlabels at=edge bottom,
        xticklabels at=edge bottom,
        ylabels at=edge left,
        yticklabels at=edge left,
        vertical sep=0pt,
        horizontal sep=0pt,
    },
    title style={at={(0,1)}, yshift=-.6\baselineskip, anchor=north west, font=\footnotesize},
    cycle list name=color list,
    width=7.3cm,
    height=6cm,
    xlabel={$z$},
    ylabel={$1+\frac{\log{N_z}-\log{T}}{n}$},
    ymin=0, ymax=0.999,
    xmin=0, xmax=1.8,
    tickpos=left,
    ytick align=outside,
    xtick align=outside,
    legend,
    legend pos=north east,
    legend style={at={(1,1)},anchor=north west,font=\tiny}
]
\nextgroupplot[title={All vertices}]
\plotfile{Untouch0D.dat}
\nextgroupplot[title={$\timp t{n^{\frac{1}{2}}}<\frac{D}{6}$}]
\plotfilelegend{Untouch0D6.dat}
\nextgroupplot[title={$\frac{D}{6}<\timp t{n^{\frac{1}{2}}}<\frac{D}{3}$}]
\plotfile{UntouchD6D3.dat}
\nextgroupplot[title={$\timp t{n^{\frac{1}{2}}}>\frac{D}{3}$}]
\plotfile{UntouchD3D.dat}
\end{groupplot}

\end{tikzpicture}

%% file: RunningTimeOracle.tex
\pgfplotstableset{
  col sep=tab,
}

\newcommand{\plotfile}[1]{
    \pgfplotstableread{#1}{\table}
    \pgfplotstablegetcolsof{#1}
    \pgfmathtruncatemacro\numberofcols{\pgfplotsretval-1}
    \pgfplotsinvokeforeach{0,...,\numberofcols}{
        \pgfplotstablegetcolumnnamebyindex{##1}\of{\table}\to{\colname}
        \addplot table [y index=##1] {#1}; 
    }
}
\newcommand{\plotfilelegend}[1]{
    \pgfplotstableread{#1}{\table}
    \pgfplotstablegetcolsof{#1}
    \pgfmathtruncatemacro\numberofcols{\pgfplotsretval-1}
    \pgfplotsinvokeforeach{1,...,\numberofcols}{
        \pgfplotstablegetcolumnnamebyindex{##1}\of{\table}\to{\colname}
        \addplot table [y index=##1] {#1}; 
        \addlegendentryexpanded{\textsc{\colname}}
    }
}
\pgfplotscreateplotcyclelist{mycolorlist}{%
{red,ultra thick}}
\centering
\begin{tikzpicture}
\begin{axis}[
    title style={at={(0.5,1)}, yshift=-.6\baselineskip, anchor=north, font=\footnotesize},
    cycle list name=mycolorlist,
    width=12cm,
    height=5.2cm,
    xlabel={$\beta$},
    ylabel={Exponent},
    ymin=0.9, ymax=2.1,
    xmin=2, xmax=3,
    tickpos=left,
    ytick align=inside,
    xtick align=inside,
    legend,
    legend pos=north east,
    legend style={at={(1,1)},anchor=north east,font=\tiny},
	,ymin=-0.05,ymax=1.2,
	title style={at={(0.3,1)}}, 
	ylabel={Relative error},
]
\plotfile{RunningTimeOracle.dat}
\end{axis}

\end{tikzpicture}

\vspace{-\baselineskip}